\documentclass[11pt, reqno, psamsfonts]{amsart}
\pdfoutput=1

\usepackage{amssymb}
\usepackage{amsthm}
\usepackage{amsmath}
\usepackage{latexsym}
\usepackage[T1]{fontenc}
\usepackage[utf8]{inputenc}
\usepackage[russian, french, english]{babel}

\usepackage{graphicx}
\usepackage{wrapfig}
\usepackage{mathtools}
\usepackage{tikz}
\usepackage{amsbsy}
\usepackage[inline]{enumitem}
\usepackage{mathrsfs}
\usetikzlibrary{shapes,snakes}
\usetikzlibrary{arrows.meta}
\usetikzlibrary{decorations.pathmorphing}
\usetikzlibrary{patterns}
\usetikzlibrary{calc}
\usepackage{float}
\usepackage{array}
\usepackage{subcaption}
\usepackage{multicol}
\usepackage{stmaryrd}
\usepackage{cancel} 
\usepackage{algorithm}
\usepackage{algpseudocode}
\makeatletter
\algnewcommand{\LineComment}[1]{\Statex \hskip\ALG@thistlm {\color{gray}\texttt{// #1}}}
\makeatother

\counterwithin{algorithm}{section}
\usepackage[justification=centering, labelfont=bf]{caption}

\usepackage{lmodern}
\usepackage{mathabx}

\usepackage{upgreek}
\usepackage{titlesec}
\usepackage{bbm}
\usepackage[spacing=true,kerning=true,babel=true,tracking=true]{microtype}
\usepackage[shortcuts]{extdash}
\usepackage[foot]{amsaddr}
\usepackage[left=1in,right=1in,top=1in,bottom=1in,bindingoffset=0cm]{geometry}
\usepackage{bm}
\usepackage{centernot}
\usepackage{mdframed}
\usepackage{hyperref}
\usepackage{mathdots}
\usepackage{xspace}
\hypersetup{
    colorlinks=true,
    linkcolor=blue,
    filecolor=magenta,      
    urlcolor=red,
    citecolor=cyan,
}
\allowdisplaybreaks

\usepackage[backend=biber, style=alphabetic, sorting=nyt, maxnames=100,backref=true, firstinits=true]{biblatex}
\title{Fast algorithms for Vizing's theorem \\ on bounded degree graphs}
\date{}
\author{\lsstyle Anton~Bernshteyn}
\email{bernshteyn@math.ucla.edu}
\author{\lsstyle Abhishek~Dhawan}
\email{adhawan2@illinois.edu}
\address{\normalfont{}(AB) \textls{Department of Mathematics, University of California, Los Angeles, CA, USA}}
\address{\normalfont{}(AD) \textls{Department of Mathematics, University of Illinois Urbana-Champaign, Urbana, IL, USA}}
\thanks{AB's research was partially supported by the NSF grant DMS-2045412 and the NSF CAREER grant DMS-2528522.
AD's research was partially supported by the NSF CAREER grant DMS-2239187 (PI: Anton Bernshteyn) and the NSF RTG grant DMS-1937241.}

\newtheoremstyle{bfnote}%
{}{}%
{\slshape}{}%
{\bfseries}{\bfseries.}%
{ }%
{\thmname{#1}\thmnumber{ #2}\thmnote{ \ep{\normalfont{}#3}}}

\newtheoremstyle{claim}%
{}{}%
{\slshape}{}%
{\itshape}{.}%
{ }%
{\thmname{#1}\thmnumber{ #2}\thmnote{ \ep{\normalfont{}#3}}}

\theoremstyle{bfnote}
\newtheorem{theo}{Theorem}[section]
\newtheorem*{theo*}{Theorem}

\newtheorem{Lemma}[theo]{Lemma}

\newtheorem{conj}[theo]{Conjecture}

\newtheorem*{corl*}{Corollary}

\theoremstyle{definition}
\newtheorem{defn}[theo]{Definition}
\newtheorem*{defn*}{Definition}

\newtheorem{ques}[theo]{Question}

\newtheorem*{exmp*}{Example}

\theoremstyle{remark}
\newtheorem*{ques*}{Question}
\newtheorem*{remk*}{Remark}

\theoremstyle{claim}

\newcounter{ForClaims}[section]

\newtheorem*{claim*}{Claim}

\makeatletter
\newcommand{\neutralize}[1]{\expandafter\let\csname c@#1\endcsname\count@}
\makeatother

\newenvironment{claimproof}{\noindent$\rhd$\hspace{1em}}{\hfill$\lhd$}

\newcommand{\0}{\varnothing}
\newcommand{\set}[1]{\{#1\}}
\newcommand{\N}{{\mathbb{N}}}
\newcommand{\Z}{\mathbb{Z}}

\renewcommand{\P}{\mathbb{P}}
\newcommand{\E}{\mathbb{E}}
\renewcommand{\epsilon}{\varepsilon}
\newcommand{\eps}{\epsilon}
\renewcommand{\phi}{\varphi}
\renewcommand{\theta}{\vartheta}
\renewcommand{\leq}{\leqslant}
\renewcommand{\geq}{\geqslant}
\newcommand{\defeq}{\coloneqq}

\newcommand{\bemph}[1]{{\normalfont#1}} 
\newcommand{\ep}[1]{\bemph{(}#1\bemph{)}} 

\newcommand{\pto}{\dashrightarrow}
\newcommand{\emphdef}[1]{\textbf{\textit{{#1}}}}
\newcommand{\blank}{\mathsf{blank}}

\newcommand{\dom}{\mathsf{dom}}
\newcommand{\Shift}{\mathsf{Shift}}
\newcommand{\vend}{\mathsf{vEnd}}
\newcommand{\vstart}{\mathsf{vStart}}

\numberwithin{equation}{section}
\newcommand{\emphd}[1]{\emphdef{#1}}

\newcommand{\poly}{\mathsf{poly}}
\newcommand{\Start}{\mathsf{Start}}
\newcommand{\End}{\mathsf{End}}
\newcommand{\Pivot}{\mathsf{Pivot}}
\newcommand{\length}{\mathsf{length}}
\newcommand{\visited}{\mathsf{visited}}

\newcommand{\aug}{\mathsf{Aug}}
\newcommand{\LOCAL}{$\mathsf{LOCAL}$\xspace}
\newcommand{\bbone}{\mathbbm{1}}

\newcommand{\IE}{E_{\mathsf{int}}}
\newcommand{\IV}{V_{\mathsf{int}}}
\newcommand{\val}{\mathsf{val}}
\newcommand{\wt}{\mathsf{wt}}

\newcommand{\graphrand}{\Gamma}
\newcommand{\graphdet}{\mathcal{H}}

\newcommand{\algsize}{\small}

\titleformat{\section}[block]{\scshape}{\thesection.}{1ex}{}
\titleformat{\subsection}[block]{\bfseries}{\thesubsection.}{1ex}{}
\titleformat{\subsection}[block]{\bfseries}{\thesubsection.}{1ex}{}
\titleformat{\subsubsection}[runin]{\itshape}{\bfseries\upshape\thesubsubsection.}{1ex}{}[.---]

\titlespacing*{\section}{0pt}{*3}{*1}
\titlespacing*{\subsection}{0pt}{*3}{*1}
\titlespacing*{\subsubsection}{0pt}{*1.5}{*0}

\renewbibmacro{in:}{}

\renewbibmacro*{volume+number+eid}{%
	\printfield{volume}%
	\setunit*{\addnbspace}
	\printfield{number}%
	\setunit{\addcomma\space}%
	\printfield{eid}}

\DeclareFieldFormat[article]{volume}{\textbf{#1}\space}
\DeclareFieldFormat[article]{number}{\mkbibparens{#1}}

\DeclareFieldFormat{journaltitle}{#1,}
\DeclareFieldFormat[thesis]{title}{\mkbibemph{#1}\addperiod}
\DeclareFieldFormat[article, unpublished, thesis]{title}{\mkbibemph{#1},}
\DeclareFieldFormat[book]{title}{\mkbibemph{#1}\addperiod}
\DeclareFieldFormat[unpublished]{howpublished}{#1, }

\DeclareFieldFormat{pages}{#1}

\DeclareFieldFormat[article]{series}{Ser.~#1\addcomma}

\setlist{topsep=3pt,itemsep=3pt}

\makeatletter
\newenvironment{breakablealgorithm}
  {
   \begin{center}
     \refstepcounter{algorithm}
     \hrule height.8pt depth0pt \kern2pt
     \renewcommand{\caption}[2][\relax]{
       {\raggedright\textbf{\ALG@name~\thealgorithm} ##2\par}%
       \ifx\relax##1\relax 
         \addcontentsline{loa}{algorithm}{\protect\numberline{\thealgorithm}##2}%
       \else 
         \addcontentsline{loa}{algorithm}{\protect\numberline{\thealgorithm}##1}%
       \fi
       \kern2pt\hrule\kern2pt
     }
  }{
     \kern2pt\hrule\relax
   \end{center}
  }
\makeatother

\makeatletter

\usepackage{framed}
\renewenvironment{shaded}{%
  \MakeFramed{\advance\hsize-\width \FrameRestore\FrameRestore}}%
 {\endMakeFramed}
\definecolor{shadecolor}{gray}{0.9}

\begin{document}

\maketitle


\begin{abstract}
    Vizing's theorem states that every graph $G$ of maximum degree $\Delta$ can be properly edge-colored using $\Delta + 1$ colors. The fastest currently known $(\Delta+1)$-edge-coloring algorithm for general graphs is due to Sinnamon and runs in time $O(m\sqrt{n})$, where $n \defeq |V(G)|$ and $m \defeq |E(G)|$. We investigate the case when $\Delta$ is constant, i.e., $\Delta = O(1)$. In this regime, the runtime of Sinnamon's algorithm is $O(n^{3/2})$, which can be improved to $O(n \log n)$, as shown by Gabow, Nishizeki, Kariv, Leven, and Terada. Here we give an algorithm whose running time is only $O(n)$, which is obviously best possible. Prior to this work, no linear-time $(\Delta+1)$-edge-coloring algorithm was known for any $\Delta \geq 4$. Using some of the same ideas, we also develop new algorithms for $(\Delta+1)$-edge-coloring in the $\mathsf{LOCAL}$ model of distributed computation. Namely, when $\Delta$ is constant, we design a deterministic \LOCAL algorithm with running time $\tilde{O}(\log^5 n)$ and a randomized \LOCAL algorithm with running time $O(\log ^2 n)$. Although our focus is on the constant $\Delta$ regime, our results remain interesting for $\Delta$ up to $\log^{o(1)} n$, since the dependence of their running time on $\Delta$ is polynomial. The key new ingredient in our algorithms is a novel application of the entropy compression method.
\end{abstract}

\vspace{40pt}

\section{Introduction}

    
    Edge-coloring is one of the most fundamental and well-studied problems in graph theory, having both great theoretical significance and a wide array of applications. Virtually every graph theory textbook contains a chapter on edge-coloring, see, e.g., \cite[\S17]{BondyMurty}, \cite[\S5.3]{Diestel}. A thorough review of the subject can be found in the monograph \cite{EdgeColoringMonograph} by Stiebitz, Scheide, Toft, and Favrholdt. The reader may also enjoy the brief historical survey by Toft and Wilson \cite{ToftWilsonSurvey} enhanced with personal memories of some of the key figures in the edge-coloring theory.
    
    In this paper we investigate edge-coloring from an algorithmic standpoint. To begin with, we need a few definitions. Unless explicitly stated otherwise, all graphs in this paper are 
    undirected and simple. For a positive integer $q \in \N^+$, we use the standard notation $[q] \defeq \set{1,2,\ldots, q}$.
    
    \begin{defn}
        Given a graph $G$, we say that two edges $e$, $f \in E(G)$ are \emphd{adjacent}, or \emphd{neighbors} of each other, if $e \neq f$ and they share an endpoint. A \emphd{proper $q$-edge-coloring} of $G$ is a mapping $\phi \colon E(G) \to [q]$ such that $\phi(e) \neq \phi(f)$ for every pair of adjacent edges $e$, $f$. The \emphd{chromatic index} of $G$, denoted by $\chi'(G)$, is the minimum $q$ such that $G$ admits a proper $q$-edge-coloring.
    \end{defn}
    
    If $G$ is a graph of maximum degree $\Delta$, then $\chi'(G) \geq \Delta$, since all edges incident to a vertex of degree $\Delta$ must receive distinct colors. Vizing famously showed that $\chi'(G)$ is always at most $1$ away from this obvious lower bound:
    
    \begin{theo}[{Vizing \cite{Vizing}}]\label{theo:Vizing}
        If $G$ is a graph of maximum degree $\Delta$, then $\chi'(G) \leq \Delta + 1$.
    \end{theo}
    
    See \cite[\S{}A.1]{EdgeColoringMonograph} for an English translation of Vizing's original paper and \cites[\S17.2]{BondyMurty}[\S5.3]{Diestel} for modern textbook presentations. As shown by Holyer \cite{Holyer}, it is \textsf{NP}-hard to determine whether $\chi'(G)$ is equal to $\Delta$ or $\Delta + 1$, even when $\Delta = 3$.
    
    A natural algorithmic problem arises: Given a graph $G$ of maximum degree $\Delta$, how efficiently can we find a proper $(\Delta+1)$-edge-coloring of $G$? As is customary in computer science, we use the RAM model of computation to analyze the time complexity of algorithms on graphs and assume that the input graph is given in the form of an adjacency list \cite[\S{}VI]{Algorithms}. The original proof of Theorem~\ref{theo:Vizing} due to Vizing is essentially algorithmic. Several variations and simplifications of Vizing's original algorithm appeared later, notably by Bollob\'as \cite[94]{Bollobas}, Rao and Dijkstra \cite{RD}, and Misra and Gries \cite{MG}. The running time of all these algorithms is $O(mn)$, where $n \defeq |V(G)|$ and $m \defeq |E(G)|$. All these algorithms iteratively color the edges of $G$ one at a time. In order to color each next edge, the algorithm may have to modify at most $O(n)$ of the edges that have been colored already, which leads to the bound $O(mn)$ on the total running time.
    
    The $O(mn)$ bound was improved to $O(m \sqrt{n \log n})$ by Arjomandi \cite{Arjomandi} and, independently, by Gabow, Nishizeki, Kariv, Leven, and Terada \cite{GNKLT} \ep{however, see the comments in \cite[41]{GNKLT} regarding a possible gap in Arjomandi's analysis}. The fastest currently known algorithm for general graphs, due to Sinnamon, has running time $O(m \sqrt{n})$:
    
    \begin{theo}[{Sinnamon \cite{Sinnamon}}]\label{theo:Sinnamon}
        There is a deterministic sequential algorithm that, given a graph $G$ with $n$ vertices, $m$ edges, and maximum degree $\Delta$, outputs a proper $(\Delta + 1)$-edge-coloring of $G$ in time $O(m \sqrt{n})$.
    \end{theo}
    
    Sinnamon's paper also presents a simpler randomized algorithm for the $(\Delta+1)$-edge-coloring that achieves the same running time (i.e., $O(m \sqrt{n})$) with probability at least $1 - e^{-\sqrt{m}}$.
    
    In this paper we focus on the case when $\Delta$ is constant:
    
    \begin{shaded}
        \begin{quote}\textsl{How fast can we find a proper $(\Delta+1)$-edge-coloring of an $n$-vertex graph of maximum degree $\Delta$ in the regime when $\Delta = O(1)$?}\end{quote}
    \end{shaded}
    
    Our aim will be to optimize the dependence of the running time on $n$ (the dependence on $\Delta$ in our results will be polynomial, but we will not attempt to sharpen it).
    In the $\Delta = O(1)$ regime, \hyperref[theo:Sinnamon]{Sinnamon's algorithm} operates in time $O(n^{3/2})$, since in that case $m = O(n)$.
    It turns out that the dependence on $n$ can be further reduced, and the fastest previously known algorithm, due to Gabow \emph{et al.}\ \cite{GNKLT}, has running time $O(n \log n)$:
    
    \begin{theo}[{Gabow--Nishizeki--Kariv--Leven--Terada \cite{GNKLT}}]\label{theo:parallel-color}
        There is a deterministic sequential algorithm that, given a graph $G$ with $n$ vertices, $m$ edges,
        and maximum degree $\Delta$, outputs a proper $(\Delta + 1)$-edge-coloring of $G$ in time 
        $O(\Delta\, m \log n) = O(\Delta^2 \, n \log n)$.
    \end{theo}
    
    In \S\ref{sec:nlogn} we show that, by modifying Sinnamon's randomized algorithm and tailoring its analysis to the case $\Delta = O(1)$,
    it is also possible to achieve the running time $O(n \log n)$ with high probability:
    
    \begin{theo}[Simple sequential algorithm]\label{theo:nlogn}
        The randomized algorithm presented in \S\ref{sec:nlogn} \ep{namely Algorithm \ref{alg:seq_log}} finds a proper $(\Delta + 1)$-edge-coloring of an $n$-vertex graph $G$ of maximum degree $\Delta$ in time $O(\Delta^3 \, n \log n)$ with probability at least $1 - 1/\poly(n)$.
    \end{theo}
    
    The reason we include Theorem~\ref{theo:nlogn} in this paper, even though it does not yield any improvement on \hyperref[theo:parallel-color]{Gabow \emph{et al.}'s algorithm}, is that 
    both the algorithm itself and its analysis are very simple. Furthermore, some of the ideas in the proof of Theorem~\ref{theo:nlogn} inform the approach we take to obtain our main results.
    
    The algorithms in Theorems~\ref{theo:parallel-color} and \ref{theo:nlogn}
    have running time 
    that, when $\Delta = O(1)$, exceeds the obvious lower bound by a factor of $\log n$. One of the main results of this paper is a randomized algorithm that reduces the running time to \emph{linear} in $n$, which is clearly optimal: 
    
    \begin{theo}[Linear-time algorithm for $\Delta = O(1)$]\label{theo:seq}
        There is a randomized sequential algorithm that, given a graph $G$ with $n$ vertices and maximum degree $\Delta$, outputs a proper $(\Delta + 1)$-edge-coloring of $G$ in time $\poly(\Delta)\,n$ with probability at least $1 - 1/\Delta^n$.
    \end{theo}

    While we are mostly interested in the case when $\Delta$ is constant, it is worth pointing out that since the dependence of the running time of our algorithm on $\Delta$ is polynomial, it beats the \hyperref[theo:parallel-color]{Gabow \emph{et al.}'s bound} for $\Delta$ up to $\log^{o(1)}n$. The exact bound our proof yields is $O(\Delta^{17} \, n)$. The exponent of $\Delta$ can likely be reduced by a more careful analysis (we made no attempt to optimize it), although significant new ideas would be required replace $\Delta^{17}$ by just $\Delta$.

    Previously, algorithms for $(\Delta + 1)$-edge-coloring with linear dependence of the runtime on $n$ were known for bipartite graphs (due to Schrijver \cite{Schrijver} and improved by Cole, Ost, and Schirra \cite{bipartite}; their algorithms only use $\Delta$ colors), planar graphs with $\Delta \geq 8$ (due to Cole and Kowalik \cite{planar}, with the case $\Delta \geq 19$ handled earlier by Chrobak and Yung \cite{ChrobakYung}; for $\Delta \geq 9$ only $\Delta$ colors are used), graphs of constant treewidth (due to Bodlaender \cite{Bodlaender} and improved by Zhou, Nakano, and Nishizeki \cite{ZNN}; these algorithms use $\chi'(G)$ colors), and for graphs with $\Delta \leq 3$ (due to Skulrattanakulchai \cite{subcubic}; note also that this case of Vizing's theorem follows from Brooks's theorem on vertex-coloring, for which general linear-time algorithms exist \cite{BrooksLovasz, BrooksOther}). Already for $\Delta = 4$, no linear-time algorithm for Vizing's theorem was known before.

    Building on some of the ideas that go into the proof of Theorem~\ref{theo:seq}, we also develop improved distributed algorithms for edge-coloring. We work in the so-called \LOCAL model of distributed computation introduced by Linial in \cite{Linial} \ep{although there are some earlier related results, e.g., by Alon, Babai, and Itai \cite{ABI}, Luby \cite{Luby}, and Goldberg, Plotkin, and Shannon \cite{GPSh}}.
    For an introduction to this subject, see the book \cite{BE} by Barenboim and Elkin. In the \LOCAL model an $n$-vertex graph $G$ abstracts a communication network where each vertex plays the role of a processor and edges represent communication links. The computation proceeds in \emphd{rounds}. During each round, the vertices first perform some local computations and then synchronously broadcast messages to all their neighbors. There are no restrictions on the complexity of the local computations and the length of the messages. After a number of rounds, every vertex must generate its own part of the output of the algorithm. For example, if the goal of the algorithm is to find a proper edge-coloring of $G$, then each vertex must eventually decide on the colors of its incident edges. The efficiency of a \LOCAL algorithm is measured by the number of communication rounds required to produce the output.
    
    A crucial feature of the \LOCAL model is that every vertex of $G$ is executing \emph{the same} algorithm. Therefore, to make this model nontrivial, there must be a way to distinguish the vertices from each other. There are two standard symmetry-breaking approaches, leading to the distinction between deterministic and randomized \LOCAL algorithms:
    \begin{itemize}
        \item In the \emphd{deterministic} version of the \LOCAL model, each vertex $v \in V(G)$ is assigned, as part of its input, an identifier $\mathsf{ID}(v)$, which is a string of $\Theta(\log n)$ bits. It is guaranteed that the identifiers assigned to different vertices are distinct. The algorithm must always output a correct solution to the problem, regardless of the specific assignment of identifiers.
        
        \item In the \emphd{randomized} version of the \LOCAL model, each vertex may generate an arbitrarily long finite string of independent uniformly distributed random bits. The algorithm must yield a correct solution to the problem with probability at least $1 - 1/\poly(n)$.
    \end{itemize}
    
    We remark that, unsurprisingly, the randomized version of the model is more computationally powerful than the deterministic one. To simulate a deterministic \LOCAL algorithm in the randomized model, each vertex can simply generate a random sequence of $O(\log n)$ bits and use it as an identifier---with a suitable choice of the implied constant in the big-$O$ notation, the probability that two identifiers generated in this way coincide can be made less than $1/n^c$ for any given $c > 0$. Furthermore, there exist many instances where a problem's randomized \LOCAL complexity is significantly lower than its deterministic \LOCAL complexity \cite{CKP}.
    
    The study of \LOCAL algorithms for edge-coloring has a long history. We direct the reader to \cite{CHLPU, GKMU} for thorough surveys. We only highlight a few benchmark results here, to put our own work into context.
    
    The $(2\Delta-1)$-edge-coloring problem has been a focus of research since the early days of distributed computing. The reason it is natural to specifically consider $(2\Delta-1)$-edge-colorings is that $2\Delta - 1$ colors are needed for a greedy \ep{sequential} edge-coloring algorithm. Thus, the $(2\Delta-1)$-edge-coloring problem is a special case of the more general $(\Delta + 1)$-\emph{vertex}-coloring problem, which has attracted a lot of attention in its own right. Already in 1988, Goldberg, Plotkin, and Shannon \cite{GPSh} designed a deterministic $(\Delta + 1)$-vertex-coloring algorithm that runs in $O(\Delta^2) + \log^\ast n$ rounds \ep{here $\log^\ast n$ is the iterated logarithm of $n$, i.e., the number of times the logarithm function must be applied to $n$ before the result becomes at most $1$}. Notice that the Goldberg--Plotkin--Shannon algorithm is extremely fast when $\Delta$ is small compared to $n$ (for instance when $\Delta$ is constant, i.e., in the main regime of interest for this paper), but for large $\Delta$ it may be rather slow due to the $O(\Delta^2)$ term. On the other hand, in 1986 Alon, Babai, and Itai \cite{ABI} and independently Luby \cite{Luby} designed a randomized $O(\log n)$-round algorithm for $(\Delta + 1)$-vertex-coloring, thus making the running time independent of $\Delta$. Developing a \emph{deterministic} $(\Delta + 1)$-vertex- or $(2\Delta - 1)$-edge-coloring algorithm that runs in $\poly(\log n)$ rounds regardless of the value of $\Delta$ has been a major challenge in the field \cite[Problems 11.2 and 11.4]{BE}. After a long line of contributions by numerous researchers, such algorithms have been discovered, first by Fischer, Ghaffari, and Kuhn \cite{FGK} for edge-coloring and then by Rozho\v{n} and Ghaffari \cite{RG} for vertex-coloring.
    
    Once the number of colors is reduced below $2\Delta - 1$, the situation starts to depend in subtle ways on the exact number of colors and on whether the model is randomized or deterministic. For example, Chang, He, Li, Pettie, and Uitto \cite{CHLPU} developed a randomized algorithm for $(\Delta + \sqrt{\Delta} \, \poly(\log \Delta))$-edge-coloring that runs in $\poly(\Delta, \log \log n)$ rounds.\footnote{Technically, this bound on the running time is not stated explicitly in \cite{CHLPU}. The algorithm presented there relies on efficiently solving instances of the Lov\'{a}sz Local Lemma, which can be done by a randomized \LOCAL algorithm in $\poly(\Delta, \log \log n)$ rounds due to a recent result of Rozho\v{n} and Ghaffari \cite{RG}.} By way of contrast, in the same paper Chang \emph{et al.}\ showed that every deterministic \LOCAL algorithm for $(2\Delta - 2)$-edge-coloring requires at least $\Omega(\log n/ \log \Delta)$ rounds, even when the input graph $G$ is a tree.
    
    Until very recently, there were no efficient \LOCAL algorithms that allow reducing the number of colors for edge-coloring all the way to \hyperref[theo:Vizing]{Vizing's bound} of $\Delta + 1$, or even to $\Delta + c$ for some $c$ independent of $\Delta$. The first such algorithm was found by Su and Vu \cite{SV}. Namely, they developed a randomized \LOCAL algorithm that finds a proper $(\Delta + 2)$-edge-coloring of $G$ in $\poly(\Delta) \log^3 n$ rounds.
    Thanks to the general derandomization technique of Rozho\v{n} and Ghaffari \cite{RG}, Su and Vu's algorithm can be converted, in a black-box manner, into a deterministic \LOCAL algorithm for $(\Delta+2)$-edge-coloring that runs in $\poly(\Delta, \log n)$ rounds.
    
    The first efficient \LOCAL algorithm for $(\Delta + 1)$-edge-coloring was presented by the first author in \cite{VizingChain}. The running time stated in \cite{VizingChain} is simply $\poly(\Delta, \log n)$, but a more careful examination yields the bounds $\poly(\Delta, \log\log n) \log^{11} n$ in the deterministic \LOCAL model and $\poly(\Delta)\log^5 n$ in the randomized \LOCAL model \ep{the randomized algorithm is not explicitly described in \cite{VizingChain}, but can be extracted from the arguments in \cite[\S\S5 and 6]{VizingChain}}. Very recently, Christiansen reduced the bound in the deterministic setting to $\poly(\Delta, \log\log n) \log^6 n$:
    
    \begin{theo}[{Christiansen \cite{Christ}}]\label{theo:Christ}
        There is a deterministic \LOCAL algorithm that finds a proper $(\Delta + 1)$-edge-coloring of an $n$-vertex graph of maximum degree $\Delta$ in $\poly(\Delta,\log\log n) \log^6 n$ rounds.
    \end{theo}
    
    After this discussion of the history of the problem, we can state our results, namely faster \LOCAL algorithms for $(\Delta + 1)$-edge-coloring in both the randomized and the deterministic setting:
    
    \begin{theo}[\LOCAL algorithms]\label{theo:dist}
        \mbox{}
    
        \begin{enumerate}[label=\ep{\normalfont\arabic*}]
            \item\label{item:dist_det} There is a deterministic \LOCAL algorithm that finds a proper $(\Delta + 1)$-edge-coloring of an $n$-vertex graph of maximum degree $\Delta$ in $\poly(\Delta, \log\log n)\log^5 n$ rounds.
            
            \item\label{item:dist_rand} There is a randomized \LOCAL algorithm that finds a proper $(\Delta + 1)$-edge-coloring of an $n$-vertex graph of maximum degree $\Delta$ in $\poly(\Delta)\log^2 n$ rounds.
        \end{enumerate}
    \end{theo}

    As in Theorem~\ref{theo:seq}, we made no attempt to optimize the exact dependence of the running time of the algorithms in Theorem~\ref{theo:dist} on $\Delta$, since we are primarily interested in the $\Delta = O(1)$ case. The explicit bounds our proof gives are $\tilde{O}(\Delta^{84}\,\log^5n)$ in the deterministic setting and $O(\Delta^{36}\,\log^2n)$ in the randomized one, but the rather large exponents of $\Delta$ can almost certainly be improved with a little more care. That being said, we doubt the dependence on $\Delta$ can be made, say, linear without substantial modifications to the proof approach.
    
    Let us now say a few words about the significance of Theorem~\ref{theo:dist}. For this discussion, we assume that $\Delta$ is constant. By the result of Chang \emph{et al.} \cite{CHLPU} mentioned earlier, every deterministic \LOCAL algorithm for edge-coloring with fewer than $2\Delta - 1$ colors requires at least $\Omega(\log n)$ rounds. The result of \cite{VizingChain} gave the first polylogarithmic upper bound $\tilde{O}(\log^{11}n)$, which was reduced to $\tilde{O}(\log^6 n)$ by Christiansen (Theorem~\ref{theo:Christ}). Our deterministic \LOCAL algorithm takes $\tilde{O}(\log^5 n)$ rounds, getting one more $\log n$ factor closer to the lower bound. At the same time, the randomized algorithm in Theorem~\ref{theo:dist} is faster than \emph{all} previously known \LOCAL algorithms---deterministic or randomized---for edge-coloring with $\Delta + c$ colors where $c$ is independent of $\Delta$ (including the Su--Vu $(\Delta + 2)$-edge-coloring algorithm \cite{SV}). Moreover, all the currently known distributed algorithms for $(\Delta + 1)$-edge-coloring are based on iteratively extending partial colorings using augmenting subgraphs. As we explain in \S\ref{subsec:overview_disjoint}, $O(\log^2 n)$ appears to be a natural barrier for the round complexity of algorithms following this strategy, so improving our randomized \LOCAL algorithm seems to require a conceptually different approach. 
    Nevertheless, the following natural problem is still open:
    
    \begin{ques}
        Fix a constant $\Delta \in \N$. Is there a deterministic \LOCAL algorithm that finds a proper $(\Delta + 1)$-edge-coloring of an $n$-vertex graph of maximum degree $\Delta$ in $O(\log n)$ rounds? What about a randomized \LOCAL algorithm? A randomized \LOCAL algorithm that only takes $o(\log n)$ rounds?
    \end{ques}

    Since the proofs of our main results are somewhat technical, we shall first give an informal high-level overview of our approach in \S\ref{sec:overview}. Then, in \S\ref{sec:notation}, we introduce some terminology and background facts that will be used throughout the rest of the paper. In \S\ref{sec:nlogn} we prove Theorem~\ref{theo:nlogn}. In \S\ref{sec:algorithms} we describe certain key subroutines used in our algorithms, whose performance is analyzed in \S\ref{sec:msva_analysis}. Finally, \S\S\ref{sec:sequential} and \ref{sec:distributed} contain the proofs of Theorems~\ref{theo:seq} and \ref{theo:dist}, respectively.

    \subsection*{Addendum}

A preliminary version of this article appeared in March 2023.
Since then, significant further progress on edge-coloring algorithms has been achieved, in part building upon the ideas of this paper. For the reader's convenience, we briefly summarize the main highlights of these new results here.

Theorem~\ref{theo:Sinnamon} was improved by Bhattacharya, Carmon, Costa, Solomon, and Zhang \cite{bhattacharya2024faster}, who designed an $O(mn^{1/3})$-time randomized algorithm for Vizing's theorem. That was a major breakthrough, as it gave the first polynomial improvement on the $O(m\sqrt{n \log n})$ time bound for general graphs due to Arjomandi \cite{Arjomandi} and Gabow \emph{et al.} \cite{GNKLT} mentioned earlier. 
Subsequently, Bhattacharya, Costa, Solomon, and Zhang \cite{bhattacharya2025even} improved the bound further to $O(mn^{1/4})$ by employing an approach that relies on so-called \emph{multi-step Vizing chains}---which play a similarly central role in our paper.
Focusing on the ``dense regime,'' Assadi gave an $O(n^2\poly(\log n))$-time randomized algorithm for Vizing's theorem \cite{assadi2025faster}.
Finally, Assadi, Behnezhad, Bhattacharya, Costa, Solomon, and Zhang very recently were able to find an $O(m\log \Delta)$-time randomized algorithm for $(\Delta + 1)$-edge-coloring \cite{assadi2025vizing}. (We note that this result uses our algorithm as a subroutine.)

There has also been a lot of recent work on so-called ``near-Vizing'' edge-coloring algorithms, i.e., $(1+\eps)\Delta$-edge-coloring where $\eps \in (0, 1)$.
Employing the result of this paper as a subroutine, Elkin and Khuzman designed an $O(m(\log (\eps\Delta) + \poly(1/\eps)))$-time randomized algorithm for this problem \cite{elkin2024deterministic}.
For $\Delta = \omega(\log n/\eps)$, Assadi gave an algorithm with expected running time $O(m\log (1/\eps))$, improving upon the Elkin--Khuzman result for large $\Delta$ \cite{assadi2025faster}.
Focusing on simplicity in the $\Delta = \Omega(\log n/\epsilon)$ regime, the second named author gave an $O(m\log^3 n/\eps^3)$-time randomized algorithm \cite{dhawan2024simple}, which he subsequently improved to $O(m\log^3\Delta/\eps^2)$ via a multi-step approach \cite{dhawan2025fast}. 
Taking the ideas of this paper further, the authors recently developed an $O(m\log (1/\eps)/\eps^4)$-time randomized algorithm with no restrictions on $\Delta$ other than $\Delta \geq 1/\eps$ (so that the number of colors is at least as large as in Vizing's theorem) \cite{bernshteyn2024linear}, thus answering a question of \cite{bhattacharya2024nibbling}. For constant $\epsilon$, this is an $O(m)$-time (i.e., truly linear) algorithm. Note that it is the first such linear-time algorithm even for greedy edge-coloring, i.e., using $2\Delta - 1$ colors.

Emboldened by the amazing progress outlined above, we wish to highlight the following two ambitious questions that remain open but no longer seem completely out of reach (although some innovative new ideas would be required to answer them):
\begin{itemize}
    \item Is there an $O(m)$-time (randomized) algorithm for Vizing's theorem? In other words, can the near-linear time bound obtained in \cite{assadi2025vizing} be made truly linear?

    \item What about deterministic edge-coloring algorithms? In spite of the recent progress in the randomized setting, the best known deterministic algorithm for Vizing's theorem on general graphs remains the $O(m\sqrt{n})$-time one due to Sinnamon \cite{Sinnamon} (see Theorem~\ref{theo:Sinnamon}). 
\end{itemize}

Going beyond simple graphs, the second named author generalized the results of this paper to multigraphs \cite{dhawan2024edge}.
The sequential edge-coloring results of that paper were then improved in \cite{assadi2025vizing}.
Finally, related techniques have also appeared in works concerning measurable and Borel edge-coloring \cite{bernshteyn2025borel, grebik2025measurable}.
As evidenced by the results mentioned in this section, we suspect that this circle of ideas will be useful in other facets of edge-coloring as well.

    \section{Informal Overview}\label{sec:overview}

    In this section we give an informal overview of our algorithms and the main ideas employed in their analysis. It should be understood that the presentation in this section deliberately ignores certain minor technical issues, and so the actual algorithms and formal definitions given in the rest of the paper may be slightly different from how they are described here. However, the differences do not affect the general conceptual framework underlying our arguments.
    
    \subsection{Augmenting subgraphs and Vizing chains}\label{subsec:overviewaug}
    
    Throughout the remainder of the paper, $G$ is a graph with $n$ vertices, $m$ edges, and of maximum degree $\Delta$. For brevity, we let $V \defeq V(G)$ and $E \defeq E(G)$. A key role in our algorithms (as well as in most other approaches to \hyperref[theo:Vizing]{Vizing's theorem}) is played by augmenting subgraphs.
    
    \begin{defn}[Augmenting subgraphs for partial colorings]\label{defn:aug}
        Let $\phi \colon E \pto [\Delta + 1]$ be a proper partial $(\Delta + 1)$-edge-coloring with domain $\dom(\phi) \subset E$. A subgraph $H \subseteq G$ is \emphd{$e$-augmenting} for an uncolored edge $e \in E \setminus \dom(\phi)$ if $e \in E(H)$ and there is a proper coloring $\phi' \colon \dom(\phi) \cup \set{e} \to [\Delta + 1]$ that agrees with $\phi$ on the edges that are not in $E(H)$; in other words, by only modifying the colors of the edges of $H$, it is possible to add $e$ to the set of colored edges. We refer to such modification operation as \emphd{augmenting} $\phi$ using $H$.
    \end{defn}
    
    The following template provides a natural model for a $(\Delta+1)$-edge-coloring algorithm:
    
    {
    \floatname{algorithm}{Algorithm Template}
    \begin{algorithm}[H]\algsize
        \caption{A sequential $(\Delta + 1)$-edge-coloring algorithm}\label{temp:seq}
        \begin{flushleft}
            \textbf{Input}: A graph $G = (V,E)$ of maximum degree $\Delta$. \\
            \textbf{Output}: A proper $(\Delta + 1)$-edge-coloring of $G$.
        \end{flushleft}
        \begin{algorithmic}[1]
            \State $\phi \gets$ the empty coloring
            \While{there are uncolored edges}
                \State Pick an uncolored edge $e$.
                \State Find an $e$-augmenting subgraph $H$.
                \State Augment $\phi$ using $H$ (thus adding $e$ to the set of colored edges).
            \EndWhile
            \State \Return $\phi$
        \end{algorithmic}
    \end{algorithm}
    }

    \begin{figure}[t]
		\centering
		\begin{subfigure}[t]{.2\textwidth}
			\centering
			\begin{tikzpicture}
			\node[circle,fill=black,draw,inner sep=0pt,minimum size=4pt] (a) at (0,0) {};
			\node[circle,fill=black,draw,inner sep=0pt,minimum size=4pt] (b) at (1,0) {};
			\path (b) ++(150:1) node[circle,fill=black,draw,inner sep=0pt,minimum size=4pt] (c) {};
			\path (b) ++(120:1) node[circle,fill=black,draw,inner sep=0pt,minimum size=4pt] (d) {};
			\node[circle,fill=black,draw,inner sep=0pt,minimum size=4pt] (e) at (1,1) {};
			\node[circle,fill=black,draw,inner sep=0pt,minimum size=4pt] (f) at (1.8,1.4) {};
			\node[circle,fill=black,draw,inner sep=0pt,minimum size=4pt] (g) at (1,1.8) {};
			\node[circle,fill=black,draw,inner sep=0pt,minimum size=4pt] (h) at (1.8,2.2) {};
			\node[circle,fill=black,draw,inner sep=0pt,minimum size=4pt] (i) at (1,2.6) {};
			\node[circle,fill=black,draw,inner sep=0pt,minimum size=4pt] (j) at (1.8,3) {};
			\node[circle,fill=black,draw,inner sep=0pt,minimum size=4pt] (k) at (1,3.4) {};
			\node[circle,fill=black,draw,inner sep=0pt,minimum size=4pt] (l) at (1.8,3.8) {};
			\node[circle,fill=black,draw,inner sep=0pt,minimum size=4pt] (m) at (1,4.2) {};
			\node[circle,fill=black,draw,inner sep=0pt,minimum size=4pt] (n) at (1.8,4.6) {};

			\draw[ thick,dotted] (a) to node[midway,anchor=north] {$e$} (b);
			\draw[ thick] (b) -- (c) (b) -- (d) (b) -- (e) -- (f) -- (g) -- (h) -- (i) -- (j) -- (k) -- (l) -- (m) -- (n);

			\draw[decoration={brace,amplitude=10pt,mirror},decorate] (2,0.9) -- node [midway,below,sloped,yshift=-7pt] {path} (2,4.7);
			
			\draw[decoration={brace,amplitude=10pt},decorate] (-0.3,0) -- node [midway,above,sloped,yshift=8pt] {fan} (1,1.3);
			\end{tikzpicture}
			\caption{A Vizing chain.}\label{fig:chains:Vizing}
		\end{subfigure}%
		\qquad%
		\begin{subfigure}[t]{.3\textwidth}
			\centering
			\begin{tikzpicture}
			\node[circle,fill=black,draw,inner sep=0pt,minimum size=4pt] (a) at (0,0) {};
			\node[circle,fill=black,draw,inner sep=0pt,minimum size=4pt] (b) at (1,0) {};
			\path (b) ++(150:1) node[circle,fill=black,draw,inner sep=0pt,minimum size=4pt] (c) {};
			\path (b) ++(120:1) node[circle,fill=black,draw,inner sep=0pt,minimum size=4pt] (d) {};
			\node[circle,fill=black,draw,inner sep=0pt,minimum size=4pt] (e) at (1,1) {};
			\node[circle,fill=black,draw,inner sep=0pt,minimum size=4pt] (f) at (1.8,1.4) {};
			\node[circle,fill=black,draw,inner sep=0pt,minimum size=4pt] (g) at (1,1.8) {};
			\node[circle,fill=black,draw,inner sep=0pt,minimum size=4pt] (h) at (1.8,2.2) {};
			\node[circle,fill=black,draw,inner sep=0pt,minimum size=4pt] (i) at (1,2.6) {};
			\node[circle,fill=black,draw,inner sep=0pt,minimum size=4pt] (j) at (1.8,3) {};
			\node[circle,fill=black,draw,inner sep=0pt,minimum size=4pt] (k) at (1,3.4) {};
			\node[circle,fill=black,draw,inner sep=0pt,minimum size=4pt] (l) at (1,4.4) {};

			\path (k) ++(60:1) node[circle,fill=black,draw,inner sep=0pt,minimum size=4pt] (m) {};
			\path (k) ++(30:1) node[circle,fill=black,draw,inner sep=0pt,minimum size=4pt] (n) {};
			\path (k) ++(0:1) node[circle,fill=black,draw,inner sep=0pt,minimum size=4pt] (o) {};
			
			\node[circle,fill=black,draw,inner sep=0pt,minimum size=4pt] (p) at (2.4,4.2) {};
			\node[circle,fill=black,draw,inner sep=0pt,minimum size=4pt] (q) at (2.8,3.4) {};
			\node[circle,fill=black,draw,inner sep=0pt,minimum size=4pt] (r) at (3.2,4.2) {};
			\node[circle,fill=black,draw,inner sep=0pt,minimum size=4pt] (s) at (3.6,3.4) {};
			\node[circle,fill=black,draw,inner sep=0pt,minimum size=4pt] (t) at (4,4.2) {};
			\node[circle,fill=black,draw,inner sep=0pt,minimum size=4pt] (u) at (4.4,3.4) {};
			
			\draw[ thick,dotted] (a) to node[midway,anchor=north] {$e$} (b);
			\draw[ thick] (b) -- (c) (b) -- (d) (b) -- (e) -- (f) -- (g) -- (h) -- (i) -- (j) -- (k) -- (l);
			
			\draw[ thick] (k) -- (m) (k) -- (n) (k) -- (o) -- (p) -- (q) -- (r) -- (s) -- (t) -- (u);
			\end{tikzpicture}
			\caption{A two-step Vizing chain.}\label{fig:chains:two}
		\end{subfigure}%
		\qquad%
		\begin{subfigure}[t]{.3\textwidth}
			\centering
			\begin{tikzpicture}
			\node[circle,fill=black,draw,inner sep=0pt,minimum size=4pt] (a) at (0,0) {};
			\node[circle,fill=black,draw,inner sep=0pt,minimum size=4pt] (b) at (1,0) {};
			\path (b) ++(150:1) node[circle,fill=black,draw,inner sep=0pt,minimum size=4pt] (c) {};
			\path (b) ++(120:1) node[circle,fill=black,draw,inner sep=0pt,minimum size=4pt] (d) {};
			\node[circle,fill=black,draw,inner sep=0pt,minimum size=4pt] (e) at (1,1) {};
			\node[circle,fill=black,draw,inner sep=0pt,minimum size=4pt] (f) at (1.8,1.4) {};
			\node[circle,fill=black,draw,inner sep=0pt,minimum size=4pt] (g) at (1,1.8) {};
			\node[circle,fill=black,draw,inner sep=0pt,minimum size=4pt] (h) at (1.8,2.2) {};
			\node[circle,fill=black,draw,inner sep=0pt,minimum size=4pt] (i) at (1,2.6) {};
			\node[circle,fill=black,draw,inner sep=0pt,minimum size=4pt] (j) at (1.8,3) {};
			\node[circle,fill=black,draw,inner sep=0pt,minimum size=4pt] (k) at (1,3.4) {};
			\node[circle,fill=black,draw,inner sep=0pt,minimum size=4pt] (l) at (1,4.4) {};
			
			\path (k) ++(60:1) node[circle,fill=black,draw,inner sep=0pt,minimum size=4pt] (m) {};
			\path (k) ++(30:1) node[circle,fill=black,draw,inner sep=0pt,minimum size=4pt] (n) {};
			\path (k) ++(0:1) node[circle,fill=black,draw,inner sep=0pt,minimum size=4pt] (o) {};
			
			\node[circle,fill=black,draw,inner sep=0pt,minimum size=4pt] (p) at (2.4,4.2) {};
			\node[circle,fill=black,draw,inner sep=0pt,minimum size=4pt] (q) at (2.8,3.4) {};
			\node[circle,fill=black,draw,inner sep=0pt,minimum size=4pt] (r) at (3.2,4.2) {};
			\node[circle,fill=black,draw,inner sep=0pt,minimum size=4pt] (s) at (3.6,3.4) {};
			\node[circle,fill=black,draw,inner sep=0pt,minimum size=4pt] (t) at (4,4.2) {};

			\node[circle,fill=black,draw,inner sep=0pt,minimum size=4pt] (v) at (5,4.2) {};
			\path (t) ++(-30:1) node[circle,fill=black,draw,inner sep=0pt,minimum size=4pt] (w) {};
			\path (t) ++(-60:1) node[circle,fill=black,draw,inner sep=0pt,minimum size=4pt] (x) {};
			\path (t) ++(-90:1) node[circle,fill=black,draw,inner sep=0pt,minimum size=4pt] (y) {};

			\node[circle,fill=black,draw,inner sep=0pt,minimum size=4pt] (z) at (4.8,2.8) {};
			\node[circle,fill=black,draw,inner sep=0pt,minimum size=4pt] (aa) at (4,2.4) {};
			\node[circle,fill=black,draw,inner sep=0pt,minimum size=4pt] (ab) at (4.8,2) {};
			\node[circle,fill=black,draw,inner sep=0pt,minimum size=4pt] (ac) at (4,1.6) {};
			\node[circle,fill=black,draw,inner sep=0pt,minimum size=4pt] (ad) at (4.8,1.2) {};
			\node[circle,fill=black,draw,inner sep=0pt,minimum size=4pt] (ae) at (4.8,0.2) {};
			
			\path (ad) ++(-120:1) node[circle,fill=black,draw,inner sep=0pt,minimum size=4pt] (ag) {};
			\path (ad) ++(-150:1) node[circle,fill=black,draw,inner sep=0pt,minimum size=4pt] (ah) {};
			\path (ad) ++(-180:1) node[circle,fill=black,draw,inner sep=0pt,minimum size=4pt] (ai) {};
			
			\node[circle,fill=black,draw,inner sep=0pt,minimum size=4pt] (aj) at (3.4,0.4) {};
			\node[circle,fill=black,draw,inner sep=0pt,minimum size=4pt] (ak) at (3,1.2) {};
			\node[circle,fill=black,draw,inner sep=0pt,minimum size=4pt] (al) at (2.6,0.4) {};
			\node[circle,fill=black,draw,inner sep=0pt,minimum size=4pt] (am) at (2.2,1.2) {};

			\draw[ thick,dotted] (a) to node[midway,anchor=north] {$e$} (b);
			\draw[ thick] (b) -- (c) (b) -- (d) (b) -- (e) -- (f) -- (g) -- (h) -- (i) -- (j) -- (k) -- (l);
			
			\draw[ thick] (k) -- (m) (k) -- (n) (k) -- (o) -- (p) -- (q) -- (r) -- (s) -- (t) -- (v);
			
			\draw[ thick] (t) -- (w) (t) -- (x) (t) -- (y) -- (z) -- (aa) -- (ab) -- (ac) -- (ad) -- (ae);
			
			\draw[ thick] (ad) -- (ag) (ad) -- (ah) (ad) -- (ai) -- (aj) -- (ak) -- (al) -- (am);
			\end{tikzpicture}
			\caption{A multi-step Vizing chain.}\label{fig:chains:multi}
		\end{subfigure}%
		\caption{Types of $e$-augmenting subgraphs.}\label{fig:typesofchains}
	\end{figure}
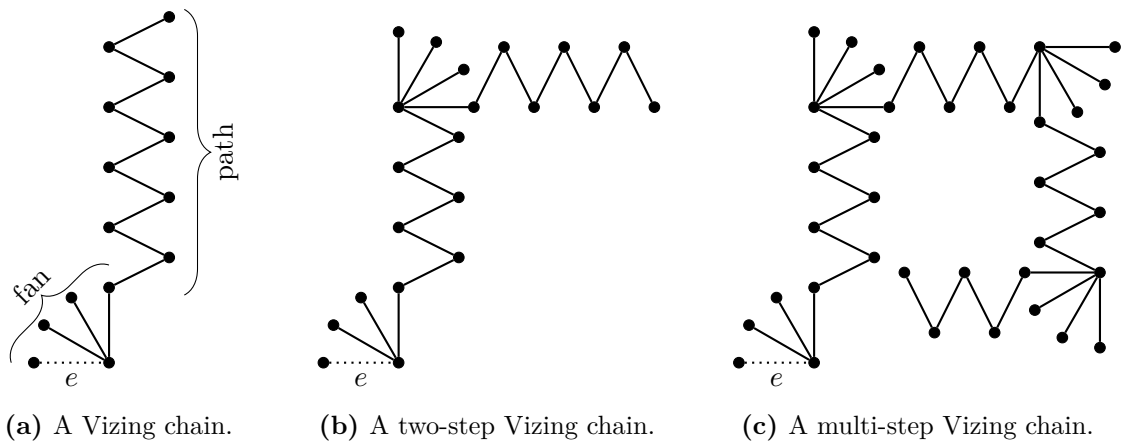
    
    The usual proofs of \hyperref[theo:Vizing]{Vizing's theorem} essentially follow this scheme, by arguing that it is always possible to find an $e$-augmenting subgraph $H$ of a certain special form, called a \emph{Vizing chain} \ep{see Fig.~\ref{fig:chains:Vizing} for an illustration}. A Vizing chain $C$ consists of a \emph{fan}---i.e., a set of edges that share the same common vertex---and a \emph{path} that is \emph{alternating} in the sense that the sequence of colors of its edges has the form $\alpha$, $\beta$, $\alpha$, $\beta$, \ldots{} for some $\alpha$, $\beta \in [\Delta + 1]$. The precise way in which the fan and the alternating path fit together varies slightly across different sources in the literature. For a formal definition of the type of Vizing chains used in this paper, see \S\ref{subsec:Vizingdefn}.
    
    The running time of an algorithm that follows Template~\ref{temp:seq} depends on how quickly we can find an $e$-augmenting subgraph $H$ and how long it takes to augment $\phi$ using $H$. Both these parameters can typically be made proportional to the number of edges in $H$. For example, the alternating path in a Vizing chain may, in principle, contain as many as $\Theta(n)$ edges. If each iteration of the \textsf{while} loop in Template~\ref{temp:seq} takes $O(n)$ time, then the entire algorithm would run in time $O(mn) = O(\Delta n^2)$. This is the core of the algorithmic analysis of \hyperref[theo:Vizing]{Vizing's theorem} given by Bollob\'as \cite[94]{Bollobas}, Rao and Dijkstra \cite{RD}, and Misra and Gries \cite{MG}.
    
    To improve this bound, we need to somehow find augmenting subgraphs with significantly fewer than $\Theta(n)$ edges. An observation that helps us here is that, although the Vizing chain for a \emph{particular} uncolored edge $e$ may contain up to $\Theta(n)$ edges, if many edges are uncolored, then it is impossible for \emph{all} of their Vizing chains to have size linear in $n$. Specifically, one can show that a vertex of $G$ may belong to at most $\poly(\Delta)$ Vizing chains. This implies that, if $s \geq 1$ edges are uncolored, the average size of a Vizing chain corresponding to an uncolored edge is at most $\poly(\Delta) n/s$, which is considerably less than $n$ when $s$ is large. This motivates the following variation on Template~\ref{temp:seq}:
    
    {
    \floatname{algorithm}{Algorithm Sketch}
    \begin{algorithm}[H]\algsize
        \caption{A randomized sequential $(\Delta + 1)$-edge-coloring algorithm using Vizing chains; see Algorithm~\ref{alg:seq_log} for the formal details}\label{inf:simple_seq}
        \begin{flushleft}
            \textbf{Input}: A graph $G = (V,E)$ of maximum degree $\Delta$. \\
            \textbf{Output}: A proper $(\Delta + 1)$-edge-coloring of $G$.
        \end{flushleft}
        \begin{algorithmic}[1]
            \State $\phi \gets$ the empty coloring
            \While{there are uncolored edges}
                \State Pick an uncolored edge $e$ uniformly at random.
                \State Construct an $e$-augmenting Vizing chain $C$.
                \State Augment $\phi$ using $C$.
            \EndWhile
            \State \Return $\phi$
        \end{algorithmic}
    \end{algorithm}
    }
    
    Algorithm~\ref{inf:simple_seq} is very similar to the randomized edge-coloring algorithm considered by Sinnamon in \cite{Sinnamon}. Since the number of uncolored edges decreases by $1$ with each iteration of the \textsf{while} loop, the expected running time of this algorithm is bounded by
    \[
        \poly(\Delta) \,\sum_{i=1}^{m} \frac{n}{m - i + 1} \,\leq\, \poly(\Delta) n \log n.
    \]
    (Here we use that $1 + 1/2 + 1/3 + \cdots + 1/m \approx \log m$ and $m \leq \Delta n/2$.) At this point, we remark that by standard arguments using Markov's inequality, repeating Algorithm~\ref{inf:simple_seq} independently $O(\log n)$ times yields an algorithm that succeeds in time at most $\poly(\Delta) n \log^2 n$ with probability at least $1 - 1/\poly(n)$. In fact, using a version of Azuma's inequality for supermartingales due to Kuszmaul and Qi \cite{azuma}, we are able to show that the running time of Algorithm~\ref{inf:simple_seq} itself is already bounded by $\poly(\Delta) n \log n$ with high probability, so the extra $\log n$ factor is not needed. 
    This is, roughly, our proof of Theorem~\ref{theo:nlogn}; the details are given in \S\ref{sec:nlogn}.

    \subsection{Multi-step Vizing chains}\label{subsec:MSVC_history}
    
    To improve on the $\poly(\Delta) n \log n$ running time given by Theorems~\ref{theo:parallel-color} and \ref{theo:nlogn}, we turn away from Vizing chains and investigate ways of building smaller augmenting subgraphs. In \cite{CHLPU}, Chang, He, Li, Pettie, and Uitto showed that, in general, given an uncolored edge $e$, there may not exist an $e$-augmenting subgraph of diameter less than $\Omega(\Delta \log(n/\Delta)) = \poly(\Delta) \log n$. When $\Delta \ll n$, this lower bound is much smaller than the $O(n)$ upper bound given by Vizing chains; nevertheless, better upper bounds have only been discovered recently. A breakthrough was achieved by Greb\'ik and Pikhurko \cite{GP}, who in the course of their study of measurable edge-colorings invented a technique for building $e$-augmenting subgraphs with just $\poly(\Delta) \sqrt{n}$ edges.\footnote{This bound is not explicit in \cite{GP}, since that paper is concerned with infinite graphs. The derivation of the $\poly(\Delta) \sqrt{n}$ bound using the Greb\'ik--Pikhurko approach is sketched in \cite[\S3]{VizingChain}.} Their idea was to consider so-called \emph{two-step Vizing chains} 
    obtained by ``gluing'' two Vizing chains \ep{see Fig.\ \ref{fig:chains:two}}.
    
    Developing this idea further, the first author introduced \emph{multi-step Vizing chains} in \cite{VizingChain} \ep{see Fig.~\ref{fig:chains:multi}}. A multi-step Vizing chain has the form $C = F_0 + P_0 + F_1 + P_1 + \cdots + F_k + P_k$, where $F_0$, $F_1$, \ldots, $F_k$ are fans and $P_0$, $P_1$, \ldots, $P_k$ are alternating paths that fit together in a specified way \ep{see \S\ref{subsec:Vizingdefn} for a precise definition}. It turns out to be important for the analysis to prohibit intersections of certain types between $F_i + P_i$ and $F_j + P_j$ for $i < j$; 
    we call multi-step Vizing chains that avoid such intersections \emph{non-intersecting} (the precise definition is given in Definition~\ref{defn:non-int}). A central result of \cite{VizingChain} is that one can always find an $e$-augmenting non-intersecting multi-step Vizing chain $C = F_0 + P_0 + \cdots + F_k + P_k$ with $k \leq \poly(\Delta) \log n$ and $\length(P_i) \leq \poly(\Delta) \log n$ for all $i$. This implies that the total number of edges in $C$ is at most $\poly(\Delta) \log^2 n$. A somewhat stronger version of this fact was used in \cite{VizingChain} to develop a \LOCAL algorithm for $(\Delta+1)$-edge-coloring that runs in $\poly(\Delta, \log n)$ rounds.
    
    Very recently, Christiansen \cite{Christ} proved that one can actually find $e$-augmenting non-intersecting multi-step Vizing chains with just $\poly(\Delta) \log n$ edges, matching the lower bound of Chang, He, Li, Pettie, and Uitto modulo the exact dependence on $\Delta$:
    
    \begin{theo}[{Christiansen \cite{Christ}}]\label{theo:christlog}
        Let $G$ be an $n$-vertex graph of maximum degree $\Delta$ and let $\phi$ be a proper partial $(\Delta + 1)$-edge-coloring of $G$. Then for each uncolored edge $e$, there exists an $e$-augmenting subgraph with at most $\poly(\Delta) \log n$ edges.
    \end{theo}
    
    Theorem~\ref{theo:christlog} plays a crucial role in Christiansen's proof of Theorem~\ref{theo:Christ}.
    
    \subsection{Algorithmic construction of short multi-step Vizing chains}\label{subsec:MSVC_inf}
    
    Christiansen proved Theorem~\ref{theo:christlog} using a clever counting argument. The downside of this approach, from our perspective, is that it is not constructive: although it demonstrates the existence of a small $e$-augmenting subgraph, it does not yield an efficient algorithm to find it \ep{or any algorithm at all other than exhaustive search}. Here we remedy this problem and develop a randomized algorithm that finds an $e$-augmenting subgraph with at most $\poly(\Delta) \log n$ edges in time $\poly(\Delta) \log n$ with high probability.
    
    Recall that we are looking for an $e$-augmenting subgraph of a special form, namely a multi-step Vizing chain $C = F_0 + P_0 + \cdots + F_k + P_k$, where each $F_i$ is a fan and each $P_i$ is an alternating path. To control the number of edges in $C$, we would like to make sure that the length of each path $P_i$ is at most $\ell$, where $\ell$ is some parameter fixed in advance. To achieve this, if the length of $P_i$ exceeds $\ell$, we shorten it to a random length $\ell' \leq \ell$ by only keeping the first $\ell'$ of its edges.\footnote{As mentioned earlier, we ignore certain technicalities in this informal discussion. For example, it turns out to be more convenient to shorten $P_i$ to a random length between $\ell$ and $2\ell - 1$, so that it is neither too long nor too short.}
    This idea is captured in Algorithm Sketch \ref{inf:MSVC1}.

\begin{figure}[b]
    \centering
    \begin{tikzpicture}[xscale = 0.7,yscale=0.6]
            \node[circle,fill=black,draw,inner sep=0pt,minimum size=4pt] (a) at (0,0) {};
        	\path (a) ++(-45:1) node[circle,fill=black,draw,inner sep=0pt,minimum size=4pt] (b) {};
        	\path (a) ++(0:1) node[circle,fill=black,draw,inner sep=0pt,minimum size=4pt] (c) {};
        	\path (a) ++(45:1) node[circle,fill=black,draw,inner sep=0pt,minimum size=4pt] (d) {};
        	\path (a) ++(90:1) node[circle,fill=black,draw,inner sep=0pt,minimum size=4pt] (e) {};

            \path (c) ++(0:3) node[circle,fill=black,draw,inner sep=0pt,minimum size=4pt] (f) {};
            
        	\path (f) ++(-45:1) node[circle,fill=black,draw,inner sep=0pt,minimum size=4pt] (g) {};
        	\path (f) ++(0:1) node[circle,fill=black,draw,inner sep=0pt,minimum size=4pt] (h) {};
        	\path (f) ++(45:1) node[circle,fill=black,draw,inner sep=0pt,minimum size=4pt] (i) {};
        	\path (f) ++(90:1) node[circle,fill=black,draw,inner sep=0pt,minimum size=4pt] (j) {};
        	
        	\path (h) ++(0:3) node[circle,fill=black,draw,inner sep=0pt,minimum size=4pt] (k) {};

            \path (k) ++(-45:1) node[circle,fill=black,draw,inner sep=0pt,minimum size=4pt] (l) {};
            \path (k) ++(0:1) node[circle,fill=black,draw,inner sep=0pt,minimum size=4pt] (m) {};
        	\path (k) ++(45:1) node[circle,fill=black,draw,inner sep=0pt,minimum size=4pt] (n) {};
        	\path (k) ++(90:1) node[circle,fill=black,draw,inner sep=0pt,minimum size=4pt] (o) {};
        	
        	\path (m) ++(0:7) node[circle,fill=black,draw,inner sep=0pt,minimum size=4pt] (p) {};

        	\draw[thick,dotted] (a) to node[font=\fontsize{8}{8},midway,inner sep=1pt,outer sep=1pt,minimum size=4pt,fill=white] {$e$} (b);
        	
        	\draw[thick, decorate,decoration=zigzag] (e) to[out=10,in=135] (f) (j) to[out=10,in=135] (k) (o) to[out=10,in=160] node[font=\fontsize{8}{8},midway,inner sep=1pt,outer sep=1pt,minimum size=4pt,fill=white] {$P$} (p);
        	
        	\draw[thick] (a) -- (c) (a) -- (d) (a) -- (e) (f) -- (g) (f) -- (h) (f) -- (i) (f) -- (j) (k) -- (l) (k) -- (m) (k) to node[font=\fontsize{8}{8},midway,inner sep=1pt,outer sep=1pt,minimum size=4pt,fill=white] {$F$} (n) (k) -- (o);

            \draw[decoration={brace,amplitude=10pt},decorate] (8,1.7) -- node [font=\fontsize{8}{8},midway,above,sloped,yshift=7pt] {$> \ell$} (16,0.5);

        \begin{scope}[yshift=-2.5cm]
            \draw[-{Stealth[length=3mm,width=2mm]},very thick,decoration = {snake,pre length=3pt,post length=7pt,},decorate] (9,1) -- (9,-1);
        \end{scope}
        
        \begin{scope}[yshift=-6cm]
            \node[circle,fill=black,draw,inner sep=0pt,minimum size=4pt] (a) at (0,0) {};
        	\path (a) ++(-45:1) node[circle,fill=black,draw,inner sep=0pt,minimum size=4pt] (b) {};
        	\path (a) ++(0:1) node[circle,fill=black,draw,inner sep=0pt,minimum size=4pt] (c) {};
        	\path (a) ++(45:1) node[circle,fill=black,draw,inner sep=0pt,minimum size=4pt] (d) {};
        	\path (a) ++(90:1) node[circle,fill=black,draw,inner sep=0pt,minimum size=4pt] (e) {};

            \path (c) ++(0:3) node[circle,fill=black,draw,inner sep=0pt,minimum size=4pt] (f) {};
            
        	\path (f) ++(-45:1) node[circle,fill=black,draw,inner sep=0pt,minimum size=4pt] (g) {};
        	\path (f) ++(0:1) node[circle,fill=black,draw,inner sep=0pt,minimum size=4pt] (h) {};
        	\path (f) ++(45:1) node[circle,fill=black,draw,inner sep=0pt,minimum size=4pt] (i) {};
        	\path (f) ++(90:1) node[circle,fill=black,draw,inner sep=0pt,minimum size=4pt] (j) {};
        	
        	\path (h) ++(0:3) node[circle,fill=black,draw,inner sep=0pt,minimum size=4pt] (k) {};

            \path (k) ++(-45:1) node[circle,fill=black,draw,inner sep=0pt,minimum size=4pt] (l) {};
            \path (k) ++(0:1) node[circle,fill=black,draw,inner sep=0pt,minimum size=4pt] (m) {};
        	\path (k) ++(45:1) node[circle,fill=black,draw,inner sep=0pt,minimum size=4pt] (n) {};
        	\path (k) ++(90:1) node[circle,fill=black,draw,inner sep=0pt,minimum size=4pt] (o) {};
        	
        	\path (m) ++(0:3) node[circle,fill=black,draw,inner sep=0pt,minimum size=4pt] (p) {};
        	
        	\path (p) ++(-45:1) node[circle,fill=black,draw,inner sep=0pt,minimum size=4pt] (q) {};

        	\draw[thick,dotted] (a) to node[font=\fontsize{8}{8},midway,inner sep=1pt,outer sep=1pt,minimum size=4pt,fill=white] {$e$} (b);
        	
        	\draw[thick, decorate,decoration=zigzag] (e) to[out=10,in=135] (f) (j) to[out=10,in=135] (k) (o) to[out=10,in=135] node[font=\fontsize{8}{8},midway,inner sep=1pt,outer sep=1pt,minimum size=4pt,fill=white] {$P_k$} (p);
        	
        	\draw[thick] (a) -- (c) (a) -- (d) (a) -- (e) (f) -- (g) (f) -- (h) (f) -- (i) (f) -- (j) (k) -- (l) (k) -- (m) (k) to node[font=\fontsize{8}{8},midway,inner sep=1pt,outer sep=1pt,minimum size=4pt,fill=white] {$F_k$} (n) (k) -- (o) (p) -- (q);
        	
        	\draw[decoration={brace,amplitude=10pt},decorate] (8.5,1.7) -- node [font=\fontsize{8}{8},midway,above,sloped,yshift=7pt] {$\ell'$} (12.5,0.1);
        \end{scope}

        \begin{scope}[yshift=6cm]
            \node[circle,fill=black,draw,inner sep=0pt,minimum size=4pt] (a) at (0,0) {};
        	\path (a) ++(-45:1) node[circle,fill=black,draw,inner sep=0pt,minimum size=4pt] (b) {};
        	\path (a) ++(0:1) node[circle,fill=black,draw,inner sep=0pt,minimum size=4pt] (c) {};
        	\path (a) ++(45:1) node[circle,fill=black,draw,inner sep=0pt,minimum size=4pt] (d) {};
        	\path (a) ++(90:1) node[circle,fill=black,draw,inner sep=0pt,minimum size=4pt] (e) {};

            \path (c) ++(0:3) node[circle,fill=black,draw,inner sep=0pt,minimum size=4pt] (f) {};
            
        	\path (f) ++(-45:1) node[circle,fill=black,draw,inner sep=0pt,minimum size=4pt] (g) {};
        	\path (f) ++(0:1) node[circle,fill=black,draw,inner sep=0pt,minimum size=4pt] (h) {};
        	\path (f) ++(45:1) node[circle,fill=black,draw,inner sep=0pt,minimum size=4pt] (i) {};
        	\path (f) ++(90:1) node[circle,fill=black,draw,inner sep=0pt,minimum size=4pt] (j) {};
        	
        	\path (h) ++(0:3) node[circle,fill=black,draw,inner sep=0pt,minimum size=4pt] (k) {};

            \path (k) ++(-45:1) node[circle,fill=black,draw,inner sep=0pt,minimum size=4pt] (l) {};
        	
        	\draw[thick,dotted] (a) to node[font=\fontsize{8}{8},midway,inner sep=1pt,outer sep=1pt,minimum size=4pt,fill=white] {$e$} (b);
        	
        	\draw[thick, decorate,decoration=zigzag] (e) to[out=10,in=135] (f) (j) to[out=10,in=135] (k);
        	
        	\draw[thick] (a) -- (c) (a) -- (d) (a) -- (e) (f) -- (g) (f) -- (h) (f) -- (i) (f) -- (j) (k) -- (l);
        \end{scope}

        \begin{scope}[yshift=3.5cm]
            \draw[-{Stealth[length=3mm,width=2mm]},very thick,decoration = {snake,pre length=3pt,post length=7pt,},decorate] (9,1) -- (9,-1);
        \end{scope}
        \end{tikzpicture}
    \caption{An iteration of the \textsf{while} loop in Algorithm \ref{inf:MSVC1} when $\length(P) > \ell$.}
    \label{fig:infiteration1}
\end{figure}
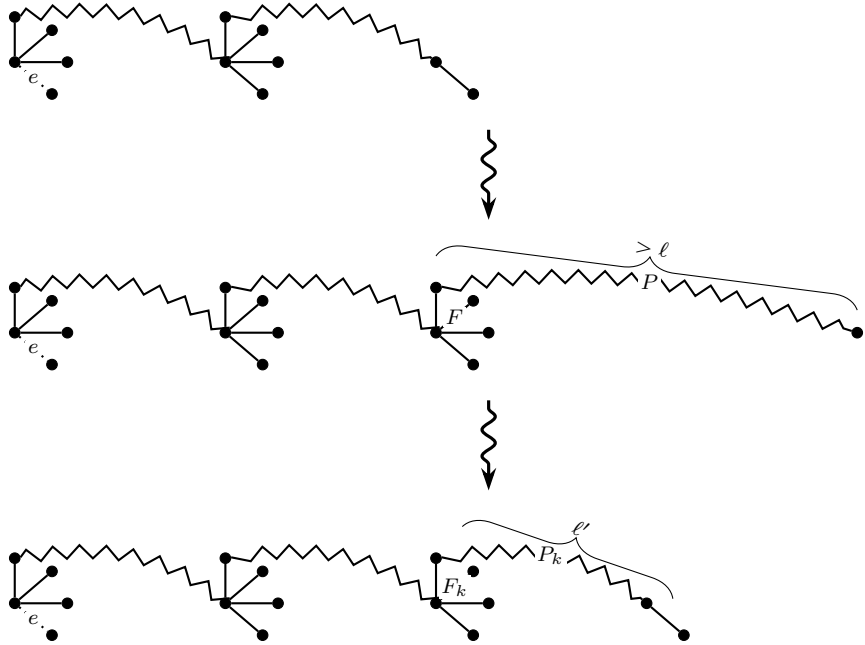

{
    \floatname{algorithm}{Algorithm Sketch}
    \begin{algorithm}[h]\algsize
        \caption{A randomized algorithm for multi-step Vizing chains: first attempt}\label{inf:MSVC1}
        \begin{flushleft}
            \textbf{Input}: A graph $G = (V,E)$ of maximum degree $\Delta$, a proper partial $(\Delta+1)$-edge-coloring $\phi$, an uncolored edge $e$, and a parameter $\ell \in \N$. \\
            \textbf{Output}: A multi-step Vizing chain $C = F_0 + P_0 + \cdots + F_k + P_k$ with $\length(P_i) \leq \ell$ for all $i$.
        \end{flushleft}
        \begin{algorithmic}[1]
            \State $e_0 \gets e$, \quad $k \gets 0$
            \While{true}
                \State $F+P \gets$ a Vizing chain starting from $e_k$
                \If{$\length(P) \leq \ell$}
                    \State $F_k \gets F$, \quad $P_k \gets P$
                    \State \Return $C = F_0 + P_0 + \cdots + F_k + P_k$
                \Else
                    \State Pick a random natural number $\ell' \leq \ell$.
                    \State $F_k \gets F$, \quad $P_k \gets$ the first $\ell'$ edges of $P$ \Comment{Randomly shorten the path} 
                    \State $e_{k+1} \gets$ the last edge of $P_k$, \quad $k \gets k+1$ \Comment{Move on to the next step} 
                \EndIf
            \EndWhile
        \end{algorithmic}
    \end{algorithm}
    }
    
    The execution of the \textsf{while} loop of Algorithm \ref{inf:MSVC1} is illustrated in Fig.~\ref{fig:infiteration1}. This is essentially the algorithm
    studied in \cite{VizingChain}, where it is proved that for $\ell = \poly(\Delta) \log n$, it terminates after at most $\poly(\Delta) \log n$ iterations of the \textsf{while} loop with high probability and returns a chain $C$ with at most $\ell \cdot \poly(\Delta) \log n \leq \poly(\Delta) \log^2 n$ edges. The assumption that $\ell = \poly(\Delta) \log n$ is invoked in \cite{VizingChain} to show that with high probability, the chain $C$ is non-intersecting; the rest of the argument goes through with just the bound $\ell \geq \poly(\Delta)$. Unfortunately, if we cannot guarantee that $C$ is non-intersecting, then the entire analysis of Algorithm~\ref{inf:MSVC1} given in \cite{VizingChain} breaks down. Thus, we are faced with the problem: \textsl{Is it possible to ensure that $C$ is non-intersecting while only assuming $\ell \geq \poly(\Delta)$?} To this end, we modify Algorithm~\ref{inf:MSVC1} in a way that always results in a non-intersecting chain by going back to an earlier stage of the process whenever an intersection occurs. This idea is presented in Algorithm {Sketch}~\ref{inf:MSVC}.

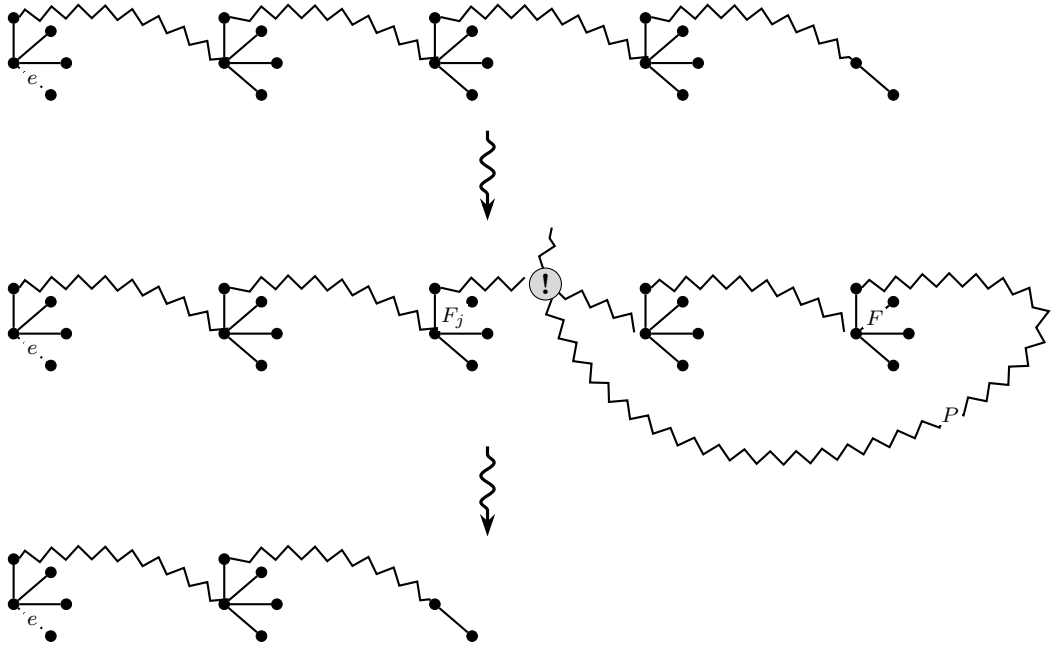
\begin{figure}[b]
    \centering
    \begin{tikzpicture}[xscale = 0.7,yscale=0.6]
            \clip (-0.5,-7) rectangle (20,7.5);
    
            \node[circle,fill=black,draw,inner sep=0pt,minimum size=4pt] (a) at (0,0) {};
        	\path (a) ++(-45:1) node[circle,fill=black,draw,inner sep=0pt,minimum size=4pt] (b) {};
        	\path (a) ++(0:1) node[circle,fill=black,draw,inner sep=0pt,minimum size=4pt] (c) {};
        	\path (a) ++(45:1) node[circle,fill=black,draw,inner sep=0pt,minimum size=4pt] (d) {};
        	\path (a) ++(90:1) node[circle,fill=black,draw,inner sep=0pt,minimum size=4pt] (e) {};

            \path (c) ++(0:3) node[circle,fill=black,draw,inner sep=0pt,minimum size=4pt] (f) {};
            
        	\path (f) ++(-45:1) node[circle,fill=black,draw,inner sep=0pt,minimum size=4pt] (g) {};
        	\path (f) ++(0:1) node[circle,fill=black,draw,inner sep=0pt,minimum size=4pt] (h) {};
        	\path (f) ++(45:1) node[circle,fill=black,draw,inner sep=0pt,minimum size=4pt] (i) {};
        	\path (f) ++(90:1) node[circle,fill=black,draw,inner sep=0pt,minimum size=4pt] (j) {};
        	
        	\path (h) ++(0:3) node[circle,fill=black,draw,inner sep=0pt,minimum size=4pt] (k) {};

            \path (k) ++(-45:1) node[circle,fill=black,draw,inner sep=0pt,minimum size=4pt] (l) {};
        	\path (k) ++(0:1) node[circle,fill=black,draw,inner sep=0pt,minimum size=4pt] (m) {};
        	\path (k) ++(45:1) node[circle,fill=black,draw,inner sep=0pt,minimum size=4pt] (n) {};
        	\path (k) ++(90:1) node[circle,fill=black,draw,inner sep=0pt,minimum size=4pt] (o) {};
        	
        	\path (m) ++(0:3) node[circle,fill=black,draw,inner sep=0pt,minimum size=4pt] (p) {};
        	
        	\path (p) ++(-45:1) node[circle,fill=black,draw,inner sep=0pt,minimum size=4pt] (q) {};
        	\path (p) ++(0:1) node[circle,fill=black,draw,inner sep=0pt,minimum size=4pt] (r) {};
        	\path (p) ++(45:1) node[circle,fill=black,draw,inner sep=0pt,minimum size=4pt] (s) {};
        	\path (p) ++(90:1) node[circle,fill=black,draw,inner sep=0pt,minimum size=4pt] (t) {};
        	
        	\path (r) ++(0:3) node[circle,fill=black,draw,inner sep=0pt,minimum size=4pt] (u) {};
        	
        	\path (u) ++(-45:1) node[circle,fill=black,draw,inner sep=0pt,minimum size=4pt] (v) {};sep=0pt,minimum size=4pt] (q) {};
        	\path (u) ++(0:1) node[circle,fill=black,draw,inner sep=0pt,minimum size=4pt] (w) {};
        	\path (u) ++(45:1) node[circle,fill=black,draw,inner sep=0pt,minimum size=4pt] (x) {};
        	\path (u) ++(90:1) node[circle,fill=black,draw,inner sep=0pt,minimum size=4pt] (y) {};
        	
        	\node[circle,inner sep=1pt] (z) at (10, 1.1) {\textbf{!}};
        	\path (z) ++(80:1.5) node (aa) {};

        	\draw[thick,dotted] (a) to node[font=\fontsize{8}{8},midway,inner sep=1pt,outer sep=1pt,minimum size=4pt,fill=white] {$e$} (b);
        	
        	\draw[thick, decorate,decoration=zigzag] (e) to[out=10,in=135] (f) (j) to[out=10,in=135] (k) (o) to[out=10,in=180] (z) to[out=-20,in=135] (p) (t) to[out=10,in=135] (u) (y) to[out=10,in=-70, looseness=4] node[font=\fontsize{8}{8},midway,inner sep=1pt,outer sep=1pt,minimum size=3pt,fill=white] {$P$} (z) -- (aa);
        	
        	\draw[thick] (a) -- (c) (a) -- (d) (a) -- (e) (f) -- (g) (f) -- (h) (f) -- (i) (f) -- (j) (k) -- (l) (k) -- (m) (k) to node[font=\fontsize{8}{8},midway,inner sep=1pt,outer sep=1pt,minimum size=4pt,fill=white] {$F_j$} (n) (k) -- (o) (p) -- (q) (p) -- (r) (p) -- (s) (p) -- (t) (u) -- (v) (u) -- (w) (u) to node[font=\fontsize{8}{8},midway,inner sep=1pt,outer sep=1pt,minimum size=4pt,fill=white] {$F$} (x) (u) -- (y);

            \node[circle,fill=gray!30,draw,inner sep=1.3pt] at (10.1, 1.1) {\textbf{!}};
        
        \begin{scope}[yshift=-3.5cm]
            \draw[-{Stealth[length=3mm,width=2mm]},very thick,decoration = {snake,pre length=3pt,post length=7pt,},decorate] (9,1) -- (9,-1);
        \end{scope}
        
        \begin{scope}[yshift=-6cm]
            \node[circle,fill=black,draw,inner sep=0pt,minimum size=4pt] (a) at (0,0) {};
        	\path (a) ++(-45:1) node[circle,fill=black,draw,inner sep=0pt,minimum size=4pt] (b) {};
        	\path (a) ++(0:1) node[circle,fill=black,draw,inner sep=0pt,minimum size=4pt] (c) {};
        	\path (a) ++(45:1) node[circle,fill=black,draw,inner sep=0pt,minimum size=4pt] (d) {};
        	\path (a) ++(90:1) node[circle,fill=black,draw,inner sep=0pt,minimum size=4pt] (e) {};

            \path (c) ++(0:3) node[circle,fill=black,draw,inner sep=0pt,minimum size=4pt] (f) {};
            
        	\path (f) ++(-45:1) node[circle,fill=black,draw,inner sep=0pt,minimum size=4pt] (g) {};
        	\path (f) ++(0:1) node[circle,fill=black,draw,inner sep=0pt,minimum size=4pt] (h) {};
        	\path (f) ++(45:1) node[circle,fill=black,draw,inner sep=0pt,minimum size=4pt] (i) {};
        	\path (f) ++(90:1) node[circle,fill=black,draw,inner sep=0pt,minimum size=4pt] (j) {};
        	
        	\path (h) ++(0:3) node[circle,fill=black,draw,inner sep=0pt,minimum size=4pt] (k) {};

            \path (k) ++(-45:1) node[circle,fill=black,draw,inner sep=0pt,minimum size=4pt] (l) {};

        	\draw[thick,dotted] (a) to node[font=\fontsize{8}{8},midway,inner sep=1pt,outer sep=1pt,minimum size=4pt,fill=white] {$e$} (b);
        	
        	\draw[thick, decorate,decoration=zigzag] (e) to[out=10,in=135] (f) (j) to[out=10,in=135] (k);
        	
        	\draw[thick] (a) -- (c) (a) -- (d) (a) -- (e) (f) -- (g) (f) -- (h) (f) -- (i) (f) -- (j) (k) -- (l);
        \end{scope}

        \begin{scope}[yshift=6cm]
            \node[circle,fill=black,draw,inner sep=0pt,minimum size=4pt] (a) at (0,0) {};
        	\path (a) ++(-45:1) node[circle,fill=black,draw,inner sep=0pt,minimum size=4pt] (b) {};
        	\path (a) ++(0:1) node[circle,fill=black,draw,inner sep=0pt,minimum size=4pt] (c) {};
        	\path (a) ++(45:1) node[circle,fill=black,draw,inner sep=0pt,minimum size=4pt] (d) {};
        	\path (a) ++(90:1) node[circle,fill=black,draw,inner sep=0pt,minimum size=4pt] (e) {};

            \path (c) ++(0:3) node[circle,fill=black,draw,inner sep=0pt,minimum size=4pt] (f) {};
            
        	\path (f) ++(-45:1) node[circle,fill=black,draw,inner sep=0pt,minimum size=4pt] (g) {};
        	\path (f) ++(0:1) node[circle,fill=black,draw,inner sep=0pt,minimum size=4pt] (h) {};
        	\path (f) ++(45:1) node[circle,fill=black,draw,inner sep=0pt,minimum size=4pt] (i) {};
        	\path (f) ++(90:1) node[circle,fill=black,draw,inner sep=0pt,minimum size=4pt] (j) {};
        	
        	\path (h) ++(0:3) node[circle,fill=black,draw,inner sep=0pt,minimum size=4pt] (k) {};

            \path (k) ++(-45:1) node[circle,fill=black,draw,inner sep=0pt,minimum size=4pt] (l) {};
        	\path (k) ++(0:1) node[circle,fill=black,draw,inner sep=0pt,minimum size=4pt] (m) {};
        	\path (k) ++(45:1) node[circle,fill=black,draw,inner sep=0pt,minimum size=4pt] (n) {};
        	\path (k) ++(90:1) node[circle,fill=black,draw,inner sep=0pt,minimum size=4pt] (o) {};
        	
        	\path (m) ++(0:3) node[circle,fill=black,draw,inner sep=0pt,minimum size=4pt] (p) {};
        	
        	\path (p) ++(-45:1) node[circle,fill=black,draw,inner sep=0pt,minimum size=4pt] (q) {};
        	\path (p) ++(0:1) node[circle,fill=black,draw,inner sep=0pt,minimum size=4pt] (r) {};
        	\path (p) ++(45:1) node[circle,fill=black,draw,inner sep=0pt,minimum size=4pt] (s) {};
        	\path (p) ++(90:1) node[circle,fill=black,draw,inner sep=0pt,minimum size=4pt] (t) {};
        	
        	\path (r) ++(0:3) node[circle,fill=black,draw,inner sep=0pt,minimum size=4pt] (u) {};
        	
        	\path (u) ++(-45:1) node[circle,fill=black,draw,inner sep=0pt,minimum size=4pt] (v) {};
        	
        	\draw[thick,dotted] (a) to node[font=\fontsize{8}{8},midway,inner sep=1pt,outer sep=1pt,minimum size=4pt,fill=white] {$e$} (b);
        	
        	\draw[thick, decorate,decoration=zigzag] (e) to[out=10,in=135] (f) (j) to[out=10,in=135] (k) (o) to[out=10,in=135] (p) (t) to[out=10,in=135] (u);
        	
        	\draw[thick] (a) -- (c) (a) -- (d) (a) -- (e) (f) -- (g) (f) -- (h) (f) -- (i) (f) -- (j) (k) -- (l) (k) -- (m) (k) -- (n) (k) -- (o) (p) -- (q) (p) -- (r) (p) -- (s) (p) -- (t) (u) -- (v);
        \end{scope}

        \begin{scope}[yshift=3.5cm]
            \draw[-{Stealth[length=3mm,width=2mm]},very thick,decoration = {snake,pre length=3pt,post length=7pt,},decorate] (9,1) -- (9,-1);
        \end{scope}
        \end{tikzpicture}
    \caption{An iteration of the \textsf{while} loop in Algorithm \ref{inf:MSVC} with an intersection.}
    \label{fig:infiteration2}
\end{figure}

\newpage
    {
    \floatname{breakablealgorithm}{Algorithm Sketch}
    \begin{breakablealgorithm}\algsize
        \caption{A randomized algorithm for non-intersecting multi-step Vizing chains; see Algorithm~\ref{alg:multi_viz_chain} for the formal details}\label{inf:MSVC}
        \begin{flushleft}
            \textbf{Input}: A graph $G = (V,E)$ of maximum degree $\Delta$, a proper partial $(\Delta+1)$-edge-coloring $\phi$, an uncolored edge $e$, and a parameter $\ell \in \N$. \\
            \textbf{Output}: A non-intersecting chain $C = F_0 + P_0 + \cdots + F_k + P_k$ with $\length(P_i) \leq \ell$ for all $i$.
        \end{flushleft}
        \begin{algorithmic}[1]
            \State $e_0 \gets e$, \quad $k \gets 0$
            \While{true}
                \State $F + P \gets$ a Vizing chain starting from $e_k$
                \If{$F + P$ intersects $F_j + P_j$ for some $j < k$}
                    \State $k \gets j$ \label{line:back} \Comment{Return to step $j$}
                \ElsIf{$\length(P) \leq \ell$}
                        \State $F_k \gets F$, \quad $P_k \gets P$
                        \State \Return $C = F_0 + P_0 + \cdots + F_k + P_k$
                \Else
                    \State Pick a random natural number $\ell' \leq \ell$. \label{line:forward}\label{line:short}
                    \State $F_k \gets F$, \quad $P_k \gets $ the first $\ell'$ edges of $P$ \Comment{Randomly shorten the path}
                    \State $e_{k+1} \gets $ the last edge of $P_k$, \quad $k \gets k+1$ \Comment{Move on to the next step}
                \EndIf
            \EndWhile
        \end{algorithmic}
    \end{breakablealgorithm}
    }
\vspace{0.1in}

    An iteration of the \textsf{while} loop in Algorithm~\ref{inf:MSVC} when $F+P$ intersects some $F_j + P_j$ is illustrated in Fig.~\ref{fig:infiteration2}.
    The crucial technical result of this paper, namely Theorem~\ref{theo:entropy_compression_dist}, is that with $\ell \geq \poly(\Delta)$, Algorithm~\ref{inf:MSVC} (or, more accurately, its slightly modified implementation presented as Algorithm~\ref{alg:multi_viz_chain}) terminates after $\log n$ iterations of the \textsf{while} loop with probability at least $1 - 1/\poly(n)$.

    In order to analyze Algorithm~\ref{inf:MSVC} we employ the \emph{entropy compression method}. This method was introduced by Moser and Tardos in their seminal paper \cite{MT} in order to prove an algorithmic version of the Lov\'asz Local Lemma \ep{the name ``entropy compression method'' was given to the Moser--Tardos technique by Tao~\cite{Tao}}. Later it was discovered \ep{first by Grytczuk, Kozik, and Micek in their study of nonrepetitive sequences \cite{Grytczuk}} that the entropy compression method can sometimes lead to improved combinatorial results if applied directly, with no explicit mention of the Local Lemma. Nevertheless, even in such applications there is usually a corresponding weaker result for which the Local Lemma is sufficient. In this regard, our application is somewhat atypical, as there does not seem to be a natural way to substitute an appeal to the Local Lemma in place of our analysis.

    The entropy compression method is used to prove that a given randomized algorithm terminates after a specified time with high probability. The idea of the method is to encode the execution process of the algorithm in such a way that the original sequence of random inputs can be recovered from the resulting encoding.
    In our application, for each iteration of the \textsf{while} loop, we include the following information in the encoding:
    \begin{itemize}
        \item whether the algorithm went to line \ref{line:back} or \ref{line:forward} on that iteration,
        \item if the algorithm went to line \ref{line:back}, the value $k - j$ (i.e., how many steps back were taken).
    \end{itemize}
    We also record the edge $e_k$ after the last iteration. Although this is insufficient to reconstruct the execution process uniquely, we prove an upper bound on the number of different processes that can result in the same record. We then use this to show that there are at most $\poly(\Delta)^t n$ sequences of random inputs for which the \textsf{while} loop is iterated at least $t$ times. On the other hand, the total number of input sequences for $t$ iterations is approximately $\ell^t$ \ep{the random inputs in Algorithm~\ref{inf:MSVC} are the numbers $\ell'$ generated in line~\ref{line:short}}. Therefore, for $\ell \geq \poly(\Delta)$ and $t \geq \log n$, the algorithm will terminate after fewer than $t$ iterations for most of the random input strings. The details of this analysis are presented in \S\ref{sec:msva_analysis}.

    \subsection{Linear-time sequential algorithm}

    We can now sketch our linear-time $(\Delta + 1)$-edge-coloring algorithm: 

    {
    \floatname{algorithm}{Algorithm Sketch}
    \begin{algorithm}[H]\algsize
        \caption{A randomized sequential $(\Delta + 1)$-edge-coloring algorithm using multi-step Vizing chains and a random sequence of edges; see Algorithm~\ref{alg:seq} for the formal details}\label{inf:seq}
        \begin{flushleft}
            \textbf{Input}: A graph $G = (V,E)$ of maximum degree $\Delta$. \\
            \textbf{Output}: A proper $(\Delta + 1)$-edge-coloring of $G$.
        \end{flushleft}
        \begin{algorithmic}[1]
            \State $\phi \gets$ the empty coloring
            \While{there are uncolored edges}
                \State Pick an uncolored edge $e$ uniformly at random.
                \State Construct an $e$-augmenting multi-step Vizing chain $C$ using Algorithm \ref{inf:MSVC} with $\ell = \poly(\Delta)$.
                \State Augment $\phi$ using $C$.
            \EndWhile
            \State \Return $\phi$
        \end{algorithmic}
    \end{algorithm}
    }

    As discussed in \S\ref{subsec:overviewaug}, Algorithm~\ref{inf:MSVC} finds an $e$-augmenting multi-step Vizing chain in time $\poly(\Delta) \log n$ with high probability. It follows that Algorithm~\ref{inf:seq} obtains a proper $(\Delta + 1)$-edge-coloring of $G$ in time $\poly(\Delta) n \log n$ with high probability---which coincides with the running time of Algorithm~\ref{inf:simple_seq}. It turns out that a better upper bound on the running time of Algorithm~\ref{inf:seq} can be derived by utilizing the randomness of the order in which the uncolored edges are processed. Specifically, it turns out that if $s \geq 1$ edges are currently uncolored, then for a randomly chosen uncolored edge $e$, the expected running time of Algorithm~\ref{inf:MSVC} is at most $\poly(\Delta) \log (n/s)$ rather than $\poly(\Delta) \log n$. Since the number of uncolored edges decreases by $1$ with each iteration, the expected running time of Algorithm~\ref{inf:seq} is bounded by
    \[
        \poly(\Delta) \sum_{i = 1}^m \log\left(\frac{n}{m- i + 1}\right) \,=\, \poly(\Delta) \log\left(\frac{n^m}{m!}\right) \,\leq\, \poly(\Delta) n,
    \]
    where we use that $m \leq \Delta n/2$ and $\log(m!) = m \log m - m + O(\log m)$ by Stirling's formula. With a somewhat more careful analysis, we will show in \S\ref{sec:sequential} that the running time of Algorithm~\ref{inf:seq} does not exceed $\poly(\Delta) n$ with probability $1 - 1/\Delta^n$.

    \subsection{Disjoint augmenting subgraphs and distributed algorithms}\label{subsec:overview_disjoint}

    An additional challenge we face when designing distributed algorithms for edge-coloring is that in the distributed setting, uncolored edges cannot be processed one at a time; instead, at every iteration a substantial fraction of the uncolored edges need to be colored simultaneously. To achieve this, we need to be able to find a ``large'' collection of vertex-disjoint augmenting subgraphs, which can then be simultaneously used to augment the current partial coloring without creating any conflicts. This approach was employed in \cite{VizingChain} to design the first \LOCAL algorithm for $(\Delta + 1)$-edge-coloring with running time $\poly(\Delta, \log n)$, see \cite[Theorem 1.3]{VizingChain}. The result of \cite{VizingChain} was sharpened by Christiansen as follows:

    \begin{theo}[{Christiansen \cite{Christ}}]\label{theo:disjoint_Christ}
        Let $\phi \colon E \pto [\Delta + 1]$ be a proper partial $(\Delta + 1)$-edge-coloring with domain $\dom(\phi) \subset E$ and let $U \defeq E \setminus \dom(\phi)$ be the set of uncolored edges. Then there exists a set $W \subseteq U$ of size $|W| \geq |U|/(\poly(\Delta) \log n)$ such that it is possible to assign a connected $e$-augmenting subgraph ${H_e}$ to each $e \in W$ so that:
        \begin{itemize}
            \item the graphs ${H_e}$, $e \in W$ are pairwise vertex-disjoint, and
            \item $|E({H_e})| \leq \poly(\Delta) \log n$ for all $e \in W$.
        \end{itemize}
    \end{theo}

    Note that Theorem~\ref{theo:disjoint_Christ} is an existence result, i.e., its proof does not yield an algorithm to find the set $W$ and the subgraphs ${H_e}$. Nevertheless, it is possible to ``algorithmize'' Theorem~\ref{theo:disjoint_Christ} by reducing the problem to finding a large matching in a certain auxiliary hypergraph and applying the \LOCAL approximation algorithm for hypergraph maximum matching due to Harris \cite{harrisdistributed} (this idea was introduced in \cite[\S2]{VizingChain}). This is how Christiansen proves Theorem~\ref{theo:Christ}, i.e., the existence of a $\poly(\Delta, \log \log n) \log^6 n$-round deterministic \LOCAL algorithm for $(\Delta+1)$-edge-coloring.

    In \S\ref{subsec:many_aug_subgraphs}, we establish the following further improvement to Theorem~\ref{theo:disjoint_Christ}:

    \begin{theo}[Many disjoint augmenting subgraphs]\label{theo:disjoint}
        Let $\phi \colon E \pto [\Delta + 1]$ be a proper partial $(\Delta + 1)$-edge-coloring with domain $\dom(\phi) \subset E$ and let $U \defeq E \setminus \dom(\phi)$ be the set of uncolored edges. There exists a randomized \LOCAL algorithm that in $\poly(\Delta) \log n$ rounds outputs a set $W \subseteq U$ of expected size $\E[|W|] \geq |U|/\poly(\Delta)$ and an assignment of 
        connected $e$-augmenting subgraphs ${H_e}$ to the edges $e \in W$ such that:
        \begin{itemize}
            \item the graphs ${H_e}$, $e \in W$ are pairwise vertex-disjoint, and
            \item $|E({H_e})| \leq \poly(\Delta) \log n$ for all $e \in W$.
        \end{itemize}
    \end{theo}

    Theorem~\ref{theo:disjoint} strengthens Theorem~\ref{theo:disjoint_Christ} in two separate ways. First, it gives a better lower bound on $|W|$: $|U|/ \poly(\Delta)$ instead of $|U|/(\poly(\Delta) \log n)$. Using this improved bound and utilizing Harris's hypergraph matching algorithm, we obtain a deterministic \LOCAL $(\Delta+1)$-edge-coloring algorithm that runs in $\poly(\Delta, \log\log n) \log^5 n$ rounds, improving the running time of Christiansen's algorithm in Theorem~\ref{theo:Christ} by a $\log n$ factor (this is Theorem~\ref{theo:dist}\ref{item:dist_det}). Second, Theorem~\ref{theo:disjoint} is constructive, in the sense that it yields a $\poly(\Delta) \log n$-round randomized \LOCAL algorithm. This means that in the randomized setting, there is no need to employ Harris's hypergraph matching algorithm. Instead we can start with the empty coloring and iterate the algorithm in Theorem~\ref{theo:disjoint} $t$ times, which would reduce the expected number of remaining uncolored edges to at most
    \[
        \left(1 - \frac{1}{\poly(\Delta)}\right)^t m \,\leq\, e^{-t/\poly(\Delta)} m.
    \]
    In particular, if $t = \poly(\Delta) \log n$, the expected number of uncolored edges will become less than $1/\poly(n)$, which implies that all edges will be colored with high probability. Since each iteration takes $\poly(\Delta) \log n$ rounds, the total running time of this randomized \LOCAL algorithm is at most $\poly(\Delta) \log^2 n$, as claimed in Theorem~\ref{theo:dist}\ref{item:dist_rand}.

    As mentioned in the introduction, when $\Delta$ is constant, $O(\log^2 n)$ appears to be a natural barrier for the round complexity of \LOCAL $(\Delta + 1)$-edge-coloring algorithms based on iteratively extending partial colorings using augmenting subgraphs. Indeed, as mentioned in \S\ref{subsec:MSVC_history}, Chang \emph{et al.}~\cite{CHLPU} showed that the smallest diameter of an augmenting subgraph for a given uncolored edge may be as large as $\Omega(\log n)$. This means that we should not expect a single iteration to take fewer than $\Omega(\log n)$ rounds. Furthermore, since the augmenting subgraphs used at every iteration must be vertex-disjoint, the edges that become colored after each iteration form a matching. Without some strong assumptions on how the uncolored edges are distributed in the graph, such a matching can, at best, include a constant fraction of the uncolored edges. Therefore, we will need to perform at least $\Omega(\log n)$ iterations to color every edge, with the total round complexity of the algorithm being at least $\Omega(\log^2 n)$. Although this is a non-rigorous, heuristic analysis that ignores certain subtle aspects of the problem (for example, while there are still more than $n^{1-o(1)}$ uncolored edges, a large fraction of them will have augmenting subgraphs of size $o(\log n)$), it seems to indicate that a conceptually different approach is needed to bring the running time to below $O(\log^2 n)$.

\section{Notation and Preliminaries on Augmenting Chains}\label{sec:notation}

\subsection{Chains: general definitions}

As mentioned in \S\ref{sec:overview}, throughout the remainder of the paper, $G$ is a graph with $n$ vertices, $m$ edges, and of maximum degree $\Delta$, and we write $V \defeq V(G)$ and $E \defeq E(G)$. 
For a vertex $v$, we let $N_G(v)$ denote the neighborhood of $v$, and $N_G[v] \defeq N_G(v) \cup \set{v}$ denote the closed neighborhood of $v$.
Without loss of generality, we shall assume that $\Delta \geq 2$.
In all our algorithms, we tacitly treat $G$, $n$, and $\Delta$ as global variables included in the input.
We call a function $\phi\,:\, E\to [\Delta + 1]\cup \{\blank\}$ a \emphd{partial $(\Delta+1)$-edge-coloring} (or simply a \emphd{partial coloring}) of $G$. Here $\phi(e) = \blank$ indicates that the edge $e$ is uncolored.

Given a partial coloring $\phi$ and $x\in V$, we let \[M(\phi, x) \defeq [\Delta+1]\setminus\{\phi(xy)\,:\, xy \in E\}\] be the set of all the \emphd{missing} colors at $x$ under the coloring $\phi$. Since there are $\Delta+1$ available colors, $M(\phi, x)$ is always nonempty. An uncolored edge $xy$ is \emphd{$\phi$-happy} if $M(\phi, x)\cap M(\phi, y)\neq \0$. If $e = xy$ is $\phi$-happy, we can extend the coloring $\phi$ by assigning any color in $M(\phi, x)\cap M(\phi, y)$ to $e$.

\begin{figure}[t]
	\centering
	\begin{tikzpicture}
	\begin{scope}
	\node[circle,fill=black,draw,inner sep=0pt,minimum size=4pt] (a) at (0,0) {};
	\node[circle,fill=black,draw,inner sep=0pt,minimum size=4pt] (b) at (-1.25,0.625) {};
	\node[circle,fill=black,draw,inner sep=0pt,minimum size=4pt] (c) at (0,1.625) {};
	\node[circle,fill=black,draw,inner sep=0pt,minimum size=4pt] (d) at (1.25,0.625) {};
	\node[circle,fill=black,draw,inner sep=0pt,minimum size=4pt] (e) at (2.5,0) {};
	\node[circle,fill=black,draw,inner sep=0pt,minimum size=4pt] (f) at (3.75,0.625) {};
	\node[circle,fill=black,draw,inner sep=0pt,minimum size=4pt] (g) at (3.75,-1) {};
	\node[circle,fill=black,draw,inner sep=0pt,minimum size=4pt] (h) at (5,0) {};
	
	\draw[thick,dotted] (a) to node[midway,inner sep=0pt,minimum size=4pt] (i) {} (b);
	\draw[thick] (a) to node[midway,inner sep=1pt,outer sep=1pt,minimum size=4pt,fill=white] (j) {$\alpha$} (c);
	\draw[thick] (a) to node[midway,inner sep=1pt,outer sep=1pt,minimum size=4pt,fill=white] (k) {$\beta$} (d);
	\draw[thick] (d) to node[midway,inner sep=1pt,outer sep=1pt,minimum size=4pt,fill=white] (l) {$\gamma$} (e);
	\draw[thick] (f) to node[midway,inner sep=1pt,outer sep=1pt,minimum size=4pt,fill=white] (m) {$\delta$} (e);
	\draw[thick] (f) to node[midway,inner sep=1pt,outer sep=1pt,minimum size=4pt,fill=white] (n) {$\epsilon$} (g);
	\draw[thick] (f) to node[midway,inner sep=1pt,outer sep=1pt,minimum size=4pt,fill=white] (o) {$\zeta$} (h);
	
	\draw[-{Stealth[length=1.6mm]}] (j) to[bend right] (i);
	\draw[-{Stealth[length=1.6mm]}] (k) to[bend right] (j);
	\draw[-{Stealth[length=1.6mm]}] (l) to[bend left] (k);
	\draw[-{Stealth[length=1.6mm]}] (m) to[bend right] (l);
	\draw[-{Stealth[length=1.6mm]}] (n) to[bend left] (m);
	\draw[-{Stealth[length=1.6mm]}] (o) to[bend left] (n);
	\end{scope}
	
	\draw[-{Stealth[length=1.6mm]},very thick,decoration = {snake,pre length=3pt,post length=7pt,},decorate] (5.4,0.3125) -- (6.35,0.3125);
	
	\begin{scope}[xshift=8cm]
	\node[circle,fill=black,draw,inner sep=0pt,minimum size=4pt] (a) at (0,0) {};
	\node[circle,fill=black,draw,inner sep=0pt,minimum size=4pt] (b) at (-1.25,0.625) {};
	\node[circle,fill=black,draw,inner sep=0pt,minimum size=4pt] (c) at (0,1.625) {};
	\node[circle,fill=black,draw,inner sep=0pt,minimum size=4pt] (d) at (1.25,0.625) {};
	\node[circle,fill=black,draw,inner sep=0pt,minimum size=4pt] (e) at (2.5,0) {};
	\node[circle,fill=black,draw,inner sep=0pt,minimum size=4pt] (f) at (3.75,0.625) {};
	\node[circle,fill=black,draw,inner sep=0pt,minimum size=4pt] (g) at (3.75,-1) {};
	\node[circle,fill=black,draw,inner sep=0pt,minimum size=4pt] (h) at (5,0) {};
	
	\draw[thick] (a) to node[midway,inner sep=1pt,outer sep=1pt,minimum size=4pt,fill=white] (i) {$\alpha$} (b);
	\draw[thick] (a) to node[midway,inner sep=1pt,outer sep=1pt,minimum size=4pt,fill=white] (j) {$\beta$} (c);
	\draw[thick] (a) to node[midway,inner sep=1pt,outer sep=1pt,minimum size=4pt,fill=white] (k) {$\gamma$} (d);
	\draw[thick] (d) to node[midway,inner sep=1pt,outer sep=1pt,minimum size=4pt,fill=white] (l) {$\delta$} (e);
	\draw[thick] (f) to node[midway,inner sep=1pt,outer sep=1pt,minimum size=4pt,fill=white] (m) {$\epsilon$} (e);
	\draw[thick] (f) to node[midway,inner sep=1pt,outer sep=1pt,minimum size=4pt,fill=white] (n) {$\zeta$} (g);
	\draw[thick,dotted] (f) to node[midway,inner sep=0pt,minimum size=4pt] (o) {} (h);
	\end{scope}
	\end{tikzpicture}
	\caption{Shifting a coloring along a chain (Greek letters represent colors).}\label{fig:shift}
\end{figure}
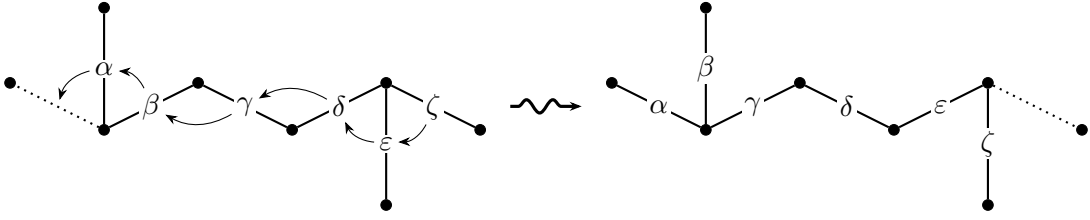

Given a proper partial coloring, we wish to modify it in order to create a partial coloring with a happy edge. To do so, we define the $\Shift$ operation. Given two adjacent edges $e_0 = xy$, $e_1 = xz$ such that $\phi(e_0) = \blank$ and $\phi(e_1) \neq \blank$, we define a new partial coloring $\psi \defeq \Shift(\phi, e_0, e_1)$ as follows:
\[\psi(e) \defeq\left\{\begin{array}{cc}
    \phi(e_1) & e = e_0; \\
    \blank & e = e_1; \\
    \phi(e) & \text{otherwise.}
\end{array}\right.\]
In other words, $\Shift(\phi, e_0, e_1)$ ``shifts'' the color from $e_1$ to $e_0$, leaves $e_1$ uncolored, and keeps the coloring on the other edges unchanged. Such an operation only yields a proper coloring if $\phi(e_1) \in M(\phi, y)$. If this condition holds, we call $(e_0, e_1)$ a \emphd{$\phi$-shiftable pair}. Note that if $(e_0, e_1)$ is $\phi$-shiftable, then $(e_1, e_0)$ is $\Shift(\phi, e_0, e_1)$-shiftable and we have
\begin{align}\label{eqn:shift_equivalence}
    \Shift(\Shift(\phi, e_0, e_1), e_1, e_0) \,=\, \phi.    
\end{align}

A \emphd{chain} of length $k$ is a sequence of edges $C = (e_0, \ldots, e_{k-1})$ such that $e_i$ and $e_{i+1}$ are adjacent for every $0 \leq i < k-1$. Let $\Start(C) \defeq e_0$ and $\End(C) \defeq e_{k-1}$ denote the first and the last edges of $C$ respectively and let $\length(C) \defeq k$ be the length of $C$. We also let $E(C) \defeq \set{e_0, \ldots, e_{k-1}}$ be the edge set of $C$ and $V(C)$ be the set of vertices incident to the edges in $E(C)$. Now we define $\Shift(\phi, C)$ iteratively as follows (see Fig.~\ref{fig:shift} for an example):
\begin{align*}
    \Shift_0(\phi, C) \,&\defeq\, \phi, \\
    \Shift_{i+1}(\phi, C) \,&\defeq\, \Shift(\Shift_{i}(\phi, C), e_{i}, e_{i+1}), \, 0 \leq i < k-1, \\
    \Shift(\phi, C) \,&\defeq\, \Shift_{k-1}(\phi, C).
\end{align*}
We say that $C$ is \emphd{$\phi$-shiftable} if the pair $(e_i,e_{i+1})$ is $\Shift_{i}(\phi, C)$-shiftable for all $0\leq i < k-1$. This in particular implies that the unique uncolored edge in $\set{e_0, \ldots, e_{k-1}}$ is $e_0$. For $C = (e_0, \ldots, e_{k-1})$, we let $C^* = (e_{k-1}, \ldots, e_0)$. By \eqref{eqn:shift_equivalence}, it is clear that if $C$ is $\phi$-shiftable, then $C^*$ is $\Shift(\phi, C)$-shiftable and $\Shift(\Shift(\phi, C), C^*) = \phi$.

The next definition captures the class of chains that can be used to create a happy edge: 

\begin{defn}[Happy chains]
    We say that a chain $C$ is \emphd{$\phi$-happy} for a partial coloring $\phi$ if it is $\phi$-shiftable and the edge $\End(C)$ is $\Shift(\phi, C)$-happy.
\end{defn}

Using the terminology of Definition~\ref{defn:aug}, we observe that a $\phi$-happy chain with $\Start(C) = e$ forms an $e$-augmenting subgraph of $G$ with respect to the partial coloring $\phi$. For a $\phi$-happy chain $C$, we let $\aug(\phi, C)$ be a proper coloring obtained from $\Shift(\phi, C)$ by assigning to $\End(C)$ some valid color. Our aim now becomes to develop algorithms for constructing $\phi$-happy chains for a given partial coloring $\phi$.

For a chain $C = (e_0, \ldots, e_{k-1})$ and $1 \leq j \leq k$, we define the \emphd{initial segment} of $C$ as
\[C|j \,\defeq\, (e_0, \ldots, e_{j-1}).\]
If $C$ is $\phi$-shiftable, then $C|j$ is also $\phi$-shiftable and $\Shift(\phi, C|j) = \Shift_{j-1}(\phi, C)$. Given two chains $C = (e_0, \ldots, e_{k-1})$ and $C' = (f_0, \ldots, f_{j-1})$ with $e_{k-1} = f_0$, we define their \emphd{sum} as follows:
\[C+C' \,\defeq\, (e_0, \ldots, e_{k-1} = f_0, \ldots, f_{j-1}).\]
Note that if $C$ is $\phi$-shiftable and $C'$ is $\Shift(\phi, C)$-shiftable, then $C+C'$ is $\phi$-shiftable.

With these definitions in mind, we are now ready to describe the types of chains that will be used as building blocks in our algorithms.

\subsection{Path chains}\label{subsec:pathchains}

The first special type of chains we will consider are path chains:

\begin{defn}[Path chains]
    A chain $P = (e_0, \ldots, e_{k-1})$ is a \emphd{path chain} if the edges $e_1$, \ldots, $e_{k-1}$ form a path in $G$, i.e., if there exist distinct vertices $x_1$, \ldots, $x_k$ such that $e_i = x_ix_{i+1}$ for all $1 \leq i \leq k-1$. We let $x_0$ be the vertex in $e_0 \setminus e_1$ and let $\vstart(P) \defeq x_0$, $\vend(P) \defeq x_k$ denote the first and last vertices on the path chain respectively. (Note that the vertex $x_0$ may coincide with $x_i$ for some $3 \leq i \leq k$; see Fig.~\ref{subfig:unsucc} for an example.) Note that the vertices $\vstart(P)$ and $\vend(P)$ are uniquely determined unless $P$ is a single edge.
\end{defn}

    \begin{defn}[Internal vertices and edges]\label{defn:internal}
        An edge of a path chain $P$ is \emphd{internal} if it is distinct from $\Start(P)$ and $\End(P)$. We let $\IE(P)$ denote the set of all internal edges of $P$. A vertex of $P$ is \emphd{internal} if it is not incident to $\Start(P)$ or $\End(P)$. We let $\IV(P)$ denote the set of all internal vertices of $P$.
    \end{defn}

We are particularly interested in path chains containing at most $2$ colors, except possibly at their first edge. Specifically, given a proper partial coloring $\phi$ and $\alpha$, $\beta \in [\Delta+1]$, we say that a path chain $P$ is an \emphd{$\alpha\beta$-path} under $\phi$ if all edges of $P$ except possibly $\Start(P)$ are colored $\alpha$ or $\beta$. To simplify working with $\alpha\beta$-paths, we introduce the following notation.
Let $G(\phi, \alpha\beta)$ be the spanning subgraph of $G$ with 
\[E(G(\phi, \alpha\beta)) \,\defeq\, \{e\in E\,:\, \phi(e) \in \{\alpha, \beta\}\}.\]
Since $\phi$ is proper, the maximum degree of $G(\phi, \alpha\beta)$ is at most $2$. Hence, the connected components of $G(\phi, \alpha\beta)$ are paths or cycles (an isolated vertex is a path of length $0$).
For $x\in V$, let $G(x;\phi, \alpha\beta)$ denote the connected component of $G(\phi, \alpha\beta)$ containing $x$ and $\deg(x; \phi, \alpha\beta)$ denote the degree of $x$ in $G(\phi, \alpha\beta)$. We say that $x$, $y \in V$ are \emphd{$(\phi, \alpha\beta)$-related} if $G(x;\phi, \alpha\beta) = G(y;\phi, \alpha\beta)$, i.e., if $y$ is reachable from $x$ by a path in $G(\phi, \alpha\beta)$.

\begin{defn}[Hopeful and successful edges]
    Let $\phi$ be a proper partial coloring of $G$ and let $\alpha$, $\beta \in [\Delta + 1]$. Let $xy \in E$ be an edge such that $\phi(xy) = \blank$. We say that $xy$ is \emphd{$(\phi, \alpha\beta)$-hopeful} if $\deg(x;\phi, \alpha\beta) < 2$ and $\deg(y;\phi, \alpha\beta) < 2$. Further, we say that $xy$ is \emphd{$(\phi, \alpha\beta)$-successful} if it is $(\phi, \alpha\beta)$-hopeful and $x$ and $y$ are not $(\phi, \alpha\beta)$-related.
\end{defn}

    Let $\phi$ be a proper partial coloring and let $\alpha$, $\beta \in [\Delta + 1]$. Consider a $(\phi, \alpha\beta)$-hopeful edge $e = xy$. Depending on the degrees of $x$ and $y$ in $G(\phi, \alpha\beta)$, the following two situations are possible:

\begin{figure}[t]
	\centering
	\begin{subfigure}[t]{.4\textwidth}
		\centering
		\begin{tikzpicture}
		\node[draw=none,minimum size=2.5cm,regular polygon,regular polygon sides=7] (P) {};

		\node[circle,fill=black,draw,inner sep=0pt,minimum size=4pt] (x) at (P.corner 4) {};
		\node[circle,fill=black,draw,inner sep=0pt,minimum size=4pt] (y) at (P.corner 5) {};
		\node[circle,fill=black,draw,inner sep=0pt,minimum size=4pt] (a) at (P.corner 6) {};
		\node[circle,fill=black,draw,inner sep=0pt,minimum size=4pt] (b) at (P.corner 7) {};
		\node[circle,fill=black,draw,inner sep=0pt,minimum size=4pt] (c) at (P.corner 1) {};
		\node[circle,fill=black,draw,inner sep=0pt,minimum size=4pt] (d) at (P.corner 2) {};
		\node[circle,fill=black,draw,inner sep=0pt,minimum size=4pt] (e) at (P.corner 3) {};
		
		\node[anchor=north] at (x) {$x$};
		\node[anchor=north] at (y) {$y$};
		
		\draw[thick,dotted] (x) to (y);
		\draw[thick] (y) to node[midway,inner sep=1pt,outer sep=1pt,minimum size=4pt,fill=white] {$\alpha$} (a);
		\draw[thick] (a) to node[midway,inner sep=1pt,outer sep=1pt,minimum size=4pt,fill=white] {$\beta$} (b);
		\draw[thick] (b) to node[midway,inner sep=1pt,outer sep=1pt,minimum size=4pt,fill=white] {$\alpha$} (c);
		\draw[thick] (c) to node[midway,inner sep=1pt,outer sep=1pt,minimum size=4pt,fill=white] {$\beta$} (d);
		\draw[thick] (d) to node[midway,inner sep=1pt,outer sep=1pt,minimum size=4pt,fill=white] {$\alpha$} (e);
		\draw[thick] (e) to node[midway,inner sep=1pt,outer sep=1pt,minimum size=4pt,fill=white] {$\beta$} (x);
		\end{tikzpicture}
		\caption{The edge $xy$ is not $(\phi, \alpha \beta)$-successful.}\label{subfig:unsucc}
	\end{subfigure}%
	\qquad%
	\begin{subfigure}[t]{.4\textwidth}
		\centering
		\begin{tikzpicture}
		\node[draw=none,minimum size=2.5cm,regular polygon,regular polygon sides=7] (P) {};

		\node[circle,fill=black,draw,inner sep=0pt,minimum size=4pt] (x) at (P.corner 4) {};
		\node[circle,fill=black,draw,inner sep=0pt,minimum size=4pt] (y) at (P.corner 5) {};
		\node[circle,fill=black,draw,inner sep=0pt,minimum size=4pt] (a) at (P.corner 6) {};
		\node[circle,fill=black,draw,inner sep=0pt,minimum size=4pt] (b) at (P.corner 7) {};
		\node[circle,fill=black,draw,inner sep=0pt,minimum size=4pt] (c) at (P.corner 1) {};
		\node[circle,fill=black,draw,inner sep=0pt,minimum size=4pt] (d) at (P.corner 2) {};
		\node[circle,fill=black,draw,inner sep=0pt,minimum size=4pt] (e) at (P.corner 3) {};
		\node[circle,fill=black,draw,inner sep=0pt,minimum size=4pt] (f) at (-2.2,0) {}; 
		
		\node[anchor=north] at (x) {$x$};
		\node[anchor=north] at (y) {$y$};
		
		\draw[thick,dotted] (x) to (y);
		\draw[thick] (y) to node[midway,inner sep=1pt,outer sep=1pt,minimum size=4pt,fill=white] {$\alpha$} (a);
		\draw[thick] (a) to node[midway,inner sep=1pt,outer sep=1pt,minimum size=4pt,fill=white] {$\beta$} (b);
		\draw[thick] (b) to node[midway,inner sep=1pt,outer sep=1pt,minimum size=4pt,fill=white] {$\alpha$} (c);
		\draw[thick] (c) to node[midway,inner sep=1pt,outer sep=1pt,minimum size=4pt,fill=white] {$\beta$} (d);
		\draw[thick] (d) to node[midway,inner sep=1pt,outer sep=1pt,minimum size=4pt,fill=white] {$\alpha$} (e);
		\draw[thick] (e) to node[midway,inner sep=1pt,outer sep=1pt,minimum size=4pt,fill=white] {$\beta$} (f);
		\end{tikzpicture}
		\caption{The edge $e = xy$ is $(\phi, \alpha \beta)$-successful.}
	\end{subfigure}
	\caption{The chain $P(e; \phi, \alpha\beta)$.}\label{fig:path}
\end{figure}
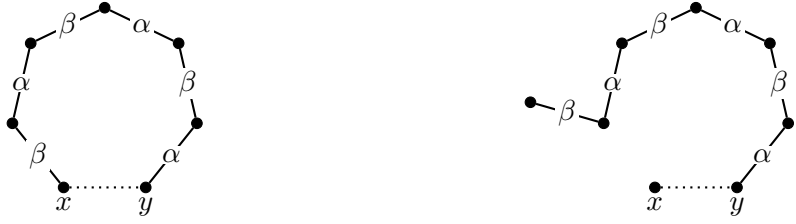

\begin{enumerate}[label=\ep{\textbf{Case \arabic*}},wide]
    \item If $\deg(x;\phi, \alpha\beta) = 0$ or $\deg(y; \phi, \alpha\beta) = 0$ (or both), then either $\alpha$ or $\beta$ is missing at both $x$ and $y$, so
    $e$ is $\phi$-happy.

    \item If $\deg(x;\phi, \alpha\beta) = \deg(y; \phi, \alpha\beta) = 1$, then both $x$ and $y$ miss exactly one of the colors $\alpha$, $\beta$. If they miss the same color, $e$ is $\phi$-happy. Otherwise suppose that, say, $\alpha \in M(\phi, x)$ and $\beta \in M(\phi, y)$.
    We define a path chain $P(e; \phi, \alpha\beta)$ by
    \[
        P(e; \phi, \alpha\beta) \,\defeq\, (e_0=e, e_1, \ldots, e_k),
    \]
    where $(e_1, \ldots, e_k)$ is the maximal path in $G(\phi, \alpha\beta)$ starting at $y$. Note that we have $\phi(e_1) = \alpha$ (in particular, the order of $\alpha$ and $\beta$ in the notation $P(e; \phi, \alpha\beta)$ matters, as $P(e; \phi, \beta\alpha)$ is the maximal path in $G(\phi, \alpha\beta)$ starting at $x$).
    The chain $P(e; \phi, \alpha\beta)$ is $\phi$-shiftable \cite[Fact 4.4]{VizingChain}. Moreover, if $e$ is $(\phi, \alpha\beta)$-successful, then $P(e; \phi, \alpha\beta)$ is $\phi$-happy by \cite[Fact 4.5]{VizingChain} and the maximality of $P(e; \phi, \alpha\beta)$. See Fig.~\ref{fig:path} for an illustration.
\end{enumerate}

\subsection{Fan chains}

    The second type of chains we will be working with are fans:

    \begin{defn}[Fans]
        A \emphd{fan} is a chain of the form $F = (xy_0, \ldots, xy_{k-1})$ where $x$ is referred to as the \emphd{pivot} of the fan and $y_0$, \ldots, $y_{k-1}$ are distinct neighbors of $x$. We let $\Pivot(F) \defeq x$, $\vstart(F) \defeq y_0$ and $\vend(F) \defeq y_{k-1}$ denote the pivot, start, and end vertices of a fan $F$. (This notation is uniquely determined unless $F$ is a single edge.)
    \end{defn}

        The process of shifting a fan is illustrated in Fig.~\ref{fig:fan}.

\begin{figure}[h]
	\centering
	\begin{tikzpicture}[xscale=1.1]
	\begin{scope}
	\node[circle,fill=black,draw,inner sep=0pt,minimum size=4pt] (x) at (0,0) {};
	\node[anchor=north] at (x) {$x$};
	
	\coordinate (O) at (0,0);
	\def\radius{2.6cm}
	
	\node[circle,fill=black,draw,inner sep=0pt,minimum size=4pt] (y0) at (190:\radius) {};
	\node at (195:2.9) {$y = y_0$};
	
	\node[circle,fill=black,draw,inner sep=0pt,minimum size=4pt] (y1) at (165:\radius) {};
	\node at (165:2.9) {$y_1$};
	
	\node[circle,fill=black,draw,inner sep=0pt,minimum size=4pt] (y2) at (140:\radius) {};
	\node at (140:2.9) {$y_2$};

	\node[circle,fill=black,draw,inner sep=0pt,minimum size=4pt] (y4) at (90:\radius) {};
	\node at (90:2.9) {$y_{i-1}$};
	
	\node[circle,fill=black,draw,inner sep=0pt,minimum size=4pt] (y5) at (65:\radius) {};
	\node at (65:2.9) {$y_i$};
	
	\node[circle,fill=black,draw,inner sep=0pt,minimum size=4pt] (y6) at (40:\radius) {};
	\node at (40:3) {$y_{i+1}$};

	\node[circle,fill=black,draw,inner sep=0pt,minimum size=4pt] (y8) at (-10:\radius) {};
	\node at (-10:3.1) {$y_{k-1}$};
	
	\node[circle,inner sep=0pt,minimum size=4pt] at (115:2) {$\ldots$}; 
	\node[circle,inner sep=0pt,minimum size=4pt] at (15:2) {$\ldots$}; 
	
	\draw[thick,dotted] (x) to (y0);
	\draw[thick] (x) to node[midway,inner sep=1pt,outer sep=1pt,minimum size=4pt,fill=white] {$\eta_0$} (y1);
	\draw[thick] (x) to node[midway,inner sep=1pt,outer sep=1pt,minimum size=4pt,fill=white] {$\eta_1$} (y2);
	
	\draw[thick] (x) to node[midway,inner sep=1pt,outer sep=1pt,minimum size=4pt,fill=white] {$\eta_{i-2}$} (y4);
	\draw[thick] (x) to node[pos=0.75,inner sep=1pt,outer sep=1pt,minimum size=4pt,fill=white] {$\eta_{i-1}$} (y5);
	\draw[thick] (x) to node[midway,inner sep=1pt,outer sep=1pt,minimum size=4pt,fill=white] {$\eta_i$} (y6);
	
	\draw[thick] (x) to node[midway,inner sep=1pt,outer sep=1pt,minimum size=4pt,fill=white] {$\eta_{k-2}$} (y8);
	\end{scope}
	
	\draw[-{Stealth[length=1.6mm]},very thick,decoration = {snake,pre length=3pt,post length=7pt,},decorate] (2.9,1) to node[midway,anchor=south]{$\Shift$} (5,1);
	
	\begin{scope}[xshift=8.3cm]
	\node[circle,fill=black,draw,inner sep=0pt,minimum size=4pt] (x) at (0,0) {};
	\node[anchor=north] at (x) {$x$};
	
	\coordinate (O) at (0,0);
	\def\radius{2.6cm}
	
	\node[circle,fill=black,draw,inner sep=0pt,minimum size=4pt] (y0) at (190:\radius) {};
	\node at (195:2.9) {$y = y_0$};
	
	\node[circle,fill=black,draw,inner sep=0pt,minimum size=4pt] (y1) at (165:\radius) {};
	\node at (165:2.9) {$y_1$};
	
	\node[circle,fill=black,draw,inner sep=0pt,minimum size=4pt] (y2) at (140:\radius) {};
	\node at (140:2.9) {$y_2$};

	\node[circle,fill=black,draw,inner sep=0pt,minimum size=4pt] (y4) at (90:\radius) {};
	\node at (90:2.9) {$y_{i-1}$};
	
	\node[circle,fill=black,draw,inner sep=0pt,minimum size=4pt] (y5) at (65:\radius) {};
	\node at (65:2.9) {$y_i$};
	
	\node[circle,fill=black,draw,inner sep=0pt,minimum size=4pt] (y6) at (40:\radius) {};
	\node at (40:3) {$y_{i+1}$};

	\node[circle,fill=black,draw,inner sep=0pt,minimum size=4pt] (y8) at (-10:\radius) {};
	\node at (-10:3.1) {$y_{k-1}$};
	
	\node[circle,inner sep=0pt,minimum size=4pt] at (115:2) {$\ldots$}; 
	\node[circle,inner sep=0pt,minimum size=4pt] at (15:2) {$\ldots$}; 
	
	\draw[thick] (x) to node[midway,inner sep=1pt,outer sep=1pt,minimum size=4pt,fill=white] {$\eta_0$} (y0);
	\draw[thick] (x) to node[midway,inner sep=1pt,outer sep=1pt,minimum size=4pt,fill=white] {$\eta_1$} (y1);
	\draw[thick] (x) to node[midway,inner sep=1pt,outer sep=1pt,minimum size=4pt,fill=white] {$\eta_2$} (y2);
	
	\draw[thick] (x) to node[midway,inner sep=1pt,outer sep=1pt,minimum size=4pt,fill=white] {$\eta_{i-1}$} (y4);
	\draw[thick] (x) to node[pos=0.75,inner sep=1pt,outer sep=1pt,minimum size=4pt,fill=white] {$\eta_i$} (y5);
	\draw[thick] (x) to node[midway,inner sep=1pt,outer sep=1pt,minimum size=4pt,fill=white] {$\eta_{i+1}$} (y6);
	
	\draw[thick, dotted] (x) to (y8);
	\end{scope}
	\end{tikzpicture}
	\caption{The process of shifting a fan.}\label{fig:fan}
\end{figure}
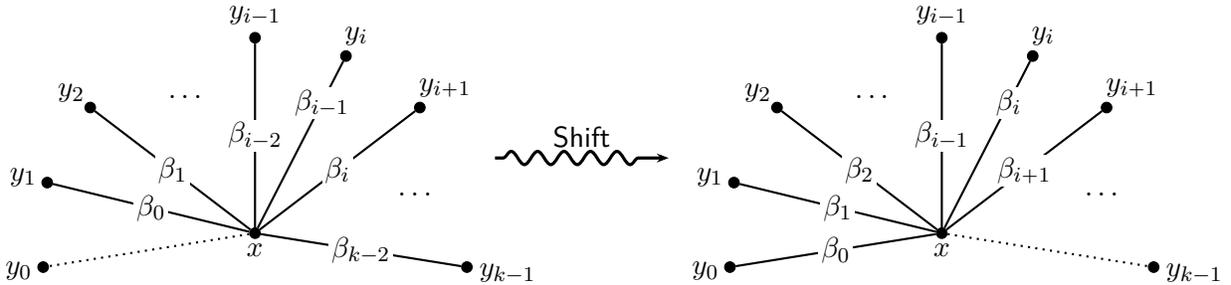

\vspace*{-9pt}

    \begin{defn}[Hopeful, successful, and disappointed fans]\label{defn:hsf}
        Let $\phi$ be a proper partial coloring and let $\alpha$, $\beta \in [\Delta+1]$. Let $F$ be a $\phi$-shiftable fan with $x \defeq \Pivot(F)$ and $y \defeq \vend(F)$ and suppose that $F$ is not $\phi$-happy (which means that the edge $xy$ is not $\Shift(\phi, F)$-happy). We say that $F$ is:
        \begin{itemize}
            \item \emphd{$(\phi, \alpha\beta)$-hopeful} if $\deg(x;\phi, \alpha\beta) < 2$ and $\deg(y;\phi, \alpha\beta) < 2$;
            \item \emphd{$(\phi, \alpha\beta)$-successful} if it is $(\phi, \alpha\beta)$-hopeful and $x$ and $y$ are not $(\Shift(\phi, F), \alpha\beta)$-related;
            \item \emphd{$(\phi, \alpha\beta)$-disappointed} if it is $(\phi, \alpha\beta)$-hopeful but not $(\phi, \alpha\beta)$-successful.
        \end{itemize}
    \end{defn}

Note that, in the setting of Definition~\ref{defn:hsf}, we have $M(\phi, x) = M(\Shift(\phi, F), x)$ and $M(\phi, y) \subseteq M(\Shift(\phi, F), y)$. Therefore, if $F$ is $(\phi, \alpha\beta)$-hopeful (resp.~successful), then
\[\deg(y;\Shift(\phi, F), \alpha\beta) \leq \deg(y;\phi, \alpha\beta) < 2, \quad \deg(x;\phi, \alpha\beta) = \deg(x;\Shift(\phi, F), \alpha\beta) < 2,\]
and so $xy$ is $(\Shift(\phi, F), \alpha\beta)$-hopeful (resp. successful) \cite[Fact 4.7]{VizingChain}.

\subsection{Vizing chains}\label{subsec:Vizingdefn}

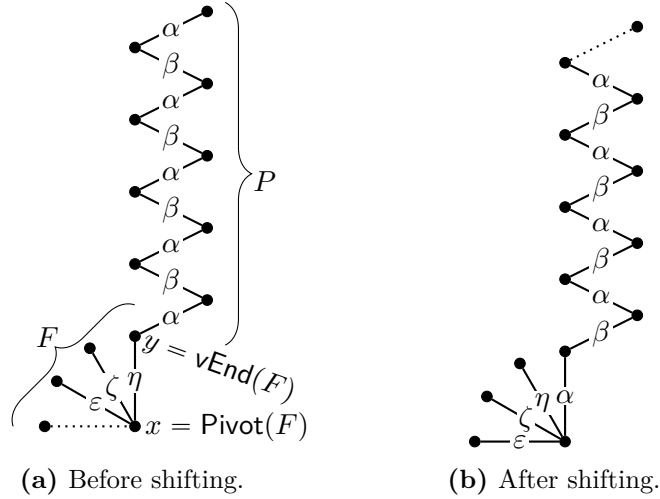
\begin{figure}[t]
	\centering
	\begin{subfigure}[t]{.2\textwidth}
		\centering
		\begin{tikzpicture}[scale=1.2]
		\node[circle,fill=black,draw,inner sep=0pt,minimum size=4pt] (a) at (0,0) {};
		\node[circle,fill=black,draw,inner sep=0pt,minimum size=4pt] (b) at (1,0) {};
		\path (b) ++(150:1) node[circle,fill=black,draw,inner sep=0pt,minimum size=4pt] (c) {}; 
		\path (b) ++(120:1) node[circle,fill=black,draw,inner sep=0pt,minimum size=4pt] (d) {}; 
		\node[circle,fill=black,draw,inner sep=0pt,minimum size=4pt] (e) at (1,1) {};
		\node[circle,fill=black,draw,inner sep=0pt,minimum size=4pt] (f) at (1.8,1.4) {};
		\node[circle,fill=black,draw,inner sep=0pt,minimum size=4pt] (g) at (1,1.8) {};
		\node[circle,fill=black,draw,inner sep=0pt,minimum size=4pt] (h) at (1.8,2.2) {};
		\node[circle,fill=black,draw,inner sep=0pt,minimum size=4pt] (i) at (1,2.6) {};
		\node[circle,fill=black,draw,inner sep=0pt,minimum size=4pt] (j) at (1.8,3) {};
		\node[circle,fill=black,draw,inner sep=0pt,minimum size=4pt] (k) at (1,3.4) {};
		\node[circle,fill=black,draw,inner sep=0pt,minimum size=4pt] (l) at (1.8,3.8) {};
		\node[circle,fill=black,draw,inner sep=0pt,minimum size=4pt] (m) at (1,4.2) {};
		\node[circle,fill=black,draw,inner sep=0pt,minimum size=4pt] (n) at (1.8,4.6) {};

		\draw[ thick,dotted] (a) -- (b);
		\draw[ thick] (b) to node[midway,inner sep=1pt,outer sep=1pt,minimum size=4pt,fill=white] {$\epsilon$} (c) (b) to node[midway,inner sep=0.5pt,outer sep=0.5pt,minimum size=4pt,fill=white] {$\zeta$} (d) (b) to node[midway,inner sep=1pt,outer sep=1pt,minimum size=4pt,fill=white] {$\eta$} (e) to node[midway,inner sep=1pt,outer sep=1pt,minimum size=4pt,fill=white] {$\alpha$} (f) to node[midway,inner sep=1pt,outer sep=1pt,minimum size=4pt,fill=white] {$\beta$} (g) to node[midway,inner sep=1pt,outer sep=1pt,minimum size=4pt,fill=white] {$\alpha$} (h) to node[midway,inner sep=1pt,outer sep=1pt,minimum size=4pt,fill=white] {$\beta$} (i) to node[midway,inner sep=1pt,outer sep=1pt,minimum size=4pt,fill=white] {$\alpha$} (j) to node[midway,inner sep=1pt,outer sep=1pt,minimum size=4pt,fill=white] {$\beta$} (k) to node[midway,inner sep=1pt,outer sep=1pt,minimum size=4pt,fill=white] {$\alpha$} (l) to node[midway,inner sep=1pt,outer sep=1pt,minimum size=4pt,fill=white] {$\beta$} (m) to node[midway,inner sep=1pt,outer sep=1pt,minimum size=4pt,fill=white] {$\alpha$} (n);
		
		\draw[decoration={brace,amplitude=10pt,mirror},decorate] (2,0.9) -- node [midway,below,xshift=15pt,yshift=5pt] {$P$} (2,4.7);
		
		\draw[decoration={brace,amplitude=10pt},decorate] (-0.35,0) -- node [midway,above,yshift=3pt,xshift=-10pt] {$F$} (1,1.35);
		
		\node[anchor=west, rotate=-20] at (e) {$y = \vend(F)$};
		\node[anchor=west] at (b) {$x = \Pivot(F)$};
		\end{tikzpicture}
		\caption{Before shifting.}\label{fig:Viz_col:initial}
	\end{subfigure}%
	\qquad\qquad\qquad%
	\begin{subfigure}[t]{.2\textwidth}
		\centering
		\begin{tikzpicture}[scale=1.2]
		\node[circle,fill=black,draw,inner sep=0pt,minimum size=4pt] (a) at (0,0) {};
		\node[circle,fill=black,draw,inner sep=0pt,minimum size=4pt] (b) at (1,0) {};
		\path (b) ++(150:1) node[circle,fill=black,draw,inner sep=0pt,minimum size=4pt] (c) {}; 
		\path (b) ++(120:1) node[circle,fill=black,draw,inner sep=0pt,minimum size=4pt] (d) {}; 
		\node[circle,fill=black,draw,inner sep=0pt,minimum size=4pt] (e) at (1,1) {};
		\node[circle,fill=black,draw,inner sep=0pt,minimum size=4pt] (f) at (1.8,1.4) {};
		\node[circle,fill=black,draw,inner sep=0pt,minimum size=4pt] (g) at (1,1.8) {};
		\node[circle,fill=black,draw,inner sep=0pt,minimum size=4pt] (h) at (1.8,2.2) {};
		\node[circle,fill=black,draw,inner sep=0pt,minimum size=4pt] (i) at (1,2.6) {};
		\node[circle,fill=black,draw,inner sep=0pt,minimum size=4pt] (j) at (1.8,3) {};
		\node[circle,fill=black,draw,inner sep=0pt,minimum size=4pt] (k) at (1,3.4) {};
		\node[circle,fill=black,draw,inner sep=0pt,minimum size=4pt] (l) at (1.8,3.8) {};
		\node[circle,fill=black,draw,inner sep=0pt,minimum size=4pt] (m) at (1,4.2) {};
		\node[circle,fill=black,draw,inner sep=0pt,minimum size=4pt] (n) at (1.8,4.6) {};

		\draw[ thick] (a) to node[midway,inner sep=1pt,outer sep=1pt,minimum size=4pt,fill=white] {$\epsilon$} (b);
		\draw[ thick] (b) to node[midway,inner sep=0.5pt,outer sep=0.5pt,minimum size=4pt,fill=white] {$\zeta$} (c) (b) to node[midway,inner sep=1pt,outer sep=1pt,minimum size=4pt,fill=white] {$\eta$} (d) (b) to node[midway,inner sep=1pt,outer sep=1pt,minimum size=4pt,fill=white] {$\alpha$} (e) to node[midway,inner sep=1pt,outer sep=1pt,minimum size=4pt,fill=white] {$\beta$} (f) to node[midway,inner sep=1pt,outer sep=1pt,minimum size=4pt,fill=white] {$\alpha$} (g) to node[midway,inner sep=1pt,outer sep=1pt,minimum size=4pt,fill=white] {$\beta$} (h) to node[midway,inner sep=1pt,outer sep=1pt,minimum size=4pt,fill=white] {$\alpha$} (i) to node[midway,inner sep=1pt,outer sep=1pt,minimum size=4pt,fill=white] {$\beta$} (j) to node[midway,inner sep=1pt,outer sep=1pt,minimum size=4pt,fill=white] {$\alpha$} (k) to node[midway,inner sep=1pt,outer sep=1pt,minimum size=4pt,fill=white] {$\beta$} (l) to node[midway,inner sep=1pt,outer sep=1pt,minimum size=4pt,fill=white] {$\alpha$} (m);
  
        \draw[ thick,dotted] (m) -- (n);
		
		\end{tikzpicture}
		\caption{After shifting.}\label{fig:Viz_col:shifted}
	\end{subfigure}%
	\caption{A Vizing chain $C = F + P$ before and after shifting.}\label{fig:Viz_col}
\end{figure}

    The third type of chain we will be working with are Vizing chains, which are formed by combining a fan and a two-colored path:

\begin{defn}[Vizing chains]
    A \emphd{Vizing chain} in a proper partial coloring $\phi$ is a chain of the form $F + P$, where $F$ is a $(\phi, \alpha\beta)$-hopeful fan for some $\alpha$, $\beta \in [\Delta + 1]$ and $P$ is an initial segment of the path chain $P(\End(F);  \Shift(\phi, F), \alpha\beta)$ with $\vstart(P) = \Pivot(F)$.
    In particular, letting $x \defeq \Pivot(F)$ and $y \defeq \vend(F)$, the chain $P$ consists of the edge $xy$ followed by a (not necessarily maximal) path starting at $y$ whose edges are colored $\alpha$ and $\beta$ under the coloring $\Shift(\phi, F)$ such that $\beta \in M(\Shift(\phi, F), y)$. See Fig.~\ref{fig:Viz_col} for an illustration. Note that we allow $P$ to comprise only the single edge $\End(F)$, in which case the Vizing chain coincides with the fan $F$.
\end{defn}

    It follows from standard proofs of \hyperref[theo:Vizing]{Vizing's theorem} that for any uncolored edge $e$, one can find a $\phi$-happy Vizing chain $C = F+P$ with $\Start(F) = e$ (see Algorithm~\ref{alg:viz_chain} below for a construction).

Combining several Vizing chains yields a multi-step Vizing chain:

\begin{defn}[Multi-step Vizing chains]
    A \emphd{$k$-step Vizing chain} is a chain of the form \[C \,=\, C_0 + \cdots + C_{k-1},\] where $C_i = F_i + P_i$ is a Vizing chain in the coloring $\Shift(\phi, C_0 + \cdots + C_{i-1})$ such that $\vend(P_i) = \vstart(F_{i+1})$ for all $0 \leq i < k-1$. See Fig.~\ref{fig:multi_Viz_chain} for an illustration.
\end{defn}

\begin{figure}[t]
    \begin{subfigure}[t]{.4\textwidth}
    	\centering
    	\begin{tikzpicture}[scale=1.3]
    	\node[circle,fill=black,draw,inner sep=0pt,minimum size=4pt] (a) at (0,0) {};
    	\node[circle,fill=black,draw,inner sep=0pt,minimum size=4pt] (b) at (1,0) {};
    	\path (b) ++(150:1) node[circle,fill=black,draw,inner sep=0pt,minimum size=4pt] (c) {};
    	\path (b) ++(120:1) node[circle,fill=black,draw,inner sep=0pt,minimum size=4pt] (d) {};
    	\node[circle,fill=black,draw,inner sep=0pt,minimum size=4pt] (e) at (1,1) {};
    	\node[circle,fill=black,draw,inner sep=0pt,minimum size=4pt] (f) at (1.8,1.4) {};
    	\node[circle,fill=black,draw,inner sep=0pt,minimum size=4pt] (g) at (1,1.8) {};
    	\node[circle,fill=black,draw,inner sep=0pt,minimum size=4pt] (h) at (1.8,2.2) {};
    	\node[circle,fill=black,draw,inner sep=0pt,minimum size=4pt] (i) at (1,2.6) {};
    	\node[circle,fill=black,draw,inner sep=0pt,minimum size=4pt] (j) at (1.8,3) {};
    	\node[circle,fill=black,draw,inner sep=0pt,minimum size=4pt] (k) at (1,3.4) {};
    	\node[circle,fill=black,draw,inner sep=0pt,minimum size=4pt] (l) at (1,4.4) {};
    	
    	\path (k) ++(60:1) node[circle,fill=black,draw,inner sep=0pt,minimum size=4pt] (m) {};
    	\path (k) ++(30:1) node[circle,fill=black,draw,inner sep=0pt,minimum size=4pt] (n) {};
    	\path (k) ++(0:1) node[circle,fill=black,draw,inner sep=0pt,minimum size=4pt] (o) {};
    	
    	\node[circle,fill=black,draw,inner sep=0pt,minimum size=4pt] (p) at (2.4,4.2) {};
    	\node[circle,fill=black,draw,inner sep=0pt,minimum size=4pt] (q) at (2.8,3.4) {};
    	\node[circle,fill=black,draw,inner sep=0pt,minimum size=4pt] (r) at (3.2,4.2) {};
    	\node[circle,fill=black,draw,inner sep=0pt,minimum size=4pt] (s) at (3.6,3.4) {};
    	\node[circle,fill=black,draw,inner sep=0pt,minimum size=4pt] (t) at (4,4.2) {};

    	\node[circle,fill=black,draw,inner sep=0pt,minimum size=4pt] (v) at (5,4.2) {};
    	\path (t) ++(-30:1) node[circle,fill=black,draw,inner sep=0pt,minimum size=4pt] (w) {};
    	\path (t) ++(-60:1) node[circle,fill=black,draw,inner sep=0pt,minimum size=4pt] (x) {};
    	\path (t) ++(-90:1) node[circle,fill=black,draw,inner sep=0pt,minimum size=4pt] (y) {};

    	\node[circle,fill=black,draw,inner sep=0pt,minimum size=4pt] (z) at (4.8,2.8) {};
    	\node[circle,fill=black,draw,inner sep=0pt,minimum size=4pt] (aa) at (4,2.4) {};
    	\node[circle,fill=black,draw,inner sep=0pt,minimum size=4pt] (ab) at (4.8,2) {};
    	\node[circle,fill=black,draw,inner sep=0pt,minimum size=4pt] (ac) at (4,1.6) {};
    	\node[circle,fill=black,draw,inner sep=0pt,minimum size=4pt] (ad) at (4.8,1.2) {};
    	\node[circle,fill=black,draw,inner sep=0pt,minimum size=4pt] (ae) at (4.8,0.2) {};
    	
    	\path (ad) ++(-120:1) node[circle,fill=black,draw,inner sep=0pt,minimum size=4pt] (ag) {};
    	\path (ad) ++(-150:1) node[circle,fill=black,draw,inner sep=0pt,minimum size=4pt] (ah) {};
    	\path (ad) ++(-180:1) node[circle,fill=black,draw,inner sep=0pt,minimum size=4pt] (ai) {};
    	
    	\node[circle,fill=black,draw,inner sep=0pt,minimum size=4pt] (aj) at (3.4,0.4) {};
    	\node[circle,fill=black,draw,inner sep=0pt,minimum size=4pt] (ak) at (3,1.2) {};
    	\node[circle,fill=black,draw,inner sep=0pt,minimum size=4pt] (al) at (2.6,0.4) {};
    	\node[circle,fill=black,draw,inner sep=0pt,minimum size=4pt] (am) at (2.2,1.2) {};
    	
    	\begin{scope}[every node/.style={scale=0.7}]
    	\draw[ thick,dotted] (a) -- (b);
    	\draw[ thick] (b) -- (c) (b) -- (d) (b) -- (e) to node[midway,inner sep=1pt,outer sep=1pt,minimum size=4pt,fill=white] {$\alpha_0$} (f) to node[midway,inner sep=1pt,outer sep=1pt,minimum size=4pt,fill=white] {$\beta_0$} (g) to node[midway,inner sep=1pt,outer sep=1pt,minimum size=4pt,fill=white] {$\alpha_0$} (h) to node[midway,inner sep=1pt,outer sep=1pt,minimum size=4pt,fill=white] {$\beta_0$} (i) to node[midway,inner sep=1pt,outer sep=1pt,minimum size=4pt,fill=white] {$\alpha_0$} (j) to node[midway,inner sep=1pt,outer sep=1pt,minimum size=4pt,fill=white] {$\beta_0$} (k) to node[midway,inner sep=1pt,outer sep=1pt,minimum size=4pt,fill=white] {$\alpha_0$} (l);
    	
    	\draw[ thick] (k) -- (m) (k) -- (n) (k) -- (o) to node[midway,inner sep=1pt,outer sep=1pt,minimum size=4pt,fill=white] {$\alpha_1$} (p) to node[midway,inner sep=1pt,outer sep=1pt,minimum size=4pt,fill=white] {$\beta_1$} (q) to node[midway,inner sep=1pt,outer sep=1pt,minimum size=4pt,fill=white] {$\alpha_1$} (r) to node[midway,inner sep=1pt,outer sep=1pt,minimum size=4pt,fill=white] {$\beta_1$} (s) to node[midway,inner sep=1pt,outer sep=1pt,minimum size=4pt,fill=white] {$\alpha_1$} (t) to node[midway,inner sep=1pt,outer sep=1pt,minimum size=4pt,fill=white] {$\beta_1$} (v);
    	
    	\draw[ thick] (t) -- (w) (t) -- (x) (t) -- (y) to node[midway,inner sep=1pt,outer sep=1pt,minimum size=4pt,fill=white] {$\alpha_2$} (z) to node[midway,inner sep=1pt,outer sep=1pt,minimum size=4pt,fill=white] {$\beta_2$} (aa) to node[midway,inner sep=1pt,outer sep=1pt,minimum size=4pt,fill=white] {$\alpha_2$} (ab) to node[midway,inner sep=1pt,outer sep=1pt,minimum size=4pt,fill=white] {$\beta_2$} (ac) to node[midway,inner sep=1pt,outer sep=1pt,minimum size=4pt,fill=white] {$\alpha_2$} (ad) to node[midway,inner sep=1pt,outer sep=1pt,minimum size=4pt,fill=white] {$\beta_2$} (ae);
    	
    	\draw[thick] (ad) -- (ag) (ad) -- (ah) (ad) -- (ai) to node[midway,inner sep=1pt,outer sep=1pt,minimum size=4pt,fill=white] {$\alpha_3$} (aj) to node[midway,inner sep=1pt,outer sep=1pt,minimum size=4pt,fill=white] {$\beta_3$} (ak) to node[midway,inner sep=1pt,outer sep=1pt,minimum size=4pt,fill=white] {$\alpha_3$} (al) to node[midway,inner sep=1pt,outer sep=1pt,minimum size=4pt,fill=white] {$\beta_3$} (am);
        \end{scope}
    	
    	\end{tikzpicture}
    	\caption{Before shifting}
    \end{subfigure}
    \qquad \qquad
    \begin{subfigure}[t]{.45\textwidth}
    	\centering
    	\begin{tikzpicture}[scale=1.3]
    	\node[circle,fill=black,draw,inner sep=0pt,minimum size=4pt] (a) at (0,0) {};
    	\node[circle,fill=black,draw,inner sep=0pt,minimum size=4pt] (b) at (1,0) {};
    	\path (b) ++(150:1) node[circle,fill=black,draw,inner sep=0pt,minimum size=4pt] (c) {};
    	\path (b) ++(120:1) node[circle,fill=black,draw,inner sep=0pt,minimum size=4pt] (d) {};
    	\node[circle,fill=black,draw,inner sep=0pt,minimum size=4pt] (e) at (1,1) {};
    	\node[circle,fill=black,draw,inner sep=0pt,minimum size=4pt] (f) at (1.8,1.4) {};
    	\node[circle,fill=black,draw,inner sep=0pt,minimum size=4pt] (g) at (1,1.8) {};
    	\node[circle,fill=black,draw,inner sep=0pt,minimum size=4pt] (h) at (1.8,2.2) {};
    	\node[circle,fill=black,draw,inner sep=0pt,minimum size=4pt] (i) at (1,2.6) {};
    	\node[circle,fill=black,draw,inner sep=0pt,minimum size=4pt] (j) at (1.8,3) {};
    	\node[circle,fill=black,draw,inner sep=0pt,minimum size=4pt] (k) at (1,3.4) {};
    	\node[circle,fill=black,draw,inner sep=0pt,minimum size=4pt] (l) at (1,4.4) {};
    	
    	\path (k) ++(60:1) node[circle,fill=black,draw,inner sep=0pt,minimum size=4pt] (m) {};
    	\path (k) ++(30:1) node[circle,fill=black,draw,inner sep=0pt,minimum size=4pt] (n) {};
    	\path (k) ++(0:1) node[circle,fill=black,draw,inner sep=0pt,minimum size=4pt] (o) {};
    	
    	\node[circle,fill=black,draw,inner sep=0pt,minimum size=4pt] (p) at (2.4,4.2) {};
    	\node[circle,fill=black,draw,inner sep=0pt,minimum size=4pt] (q) at (2.8,3.4) {};
    	\node[circle,fill=black,draw,inner sep=0pt,minimum size=4pt] (r) at (3.2,4.2) {};
    	\node[circle,fill=black,draw,inner sep=0pt,minimum size=4pt] (s) at (3.6,3.4) {};
    	\node[circle,fill=black,draw,inner sep=0pt,minimum size=4pt] (t) at (4,4.2) {};

    	\node[circle,fill=black,draw,inner sep=0pt,minimum size=4pt] (v) at (5,4.2) {};
    	\path (t) ++(-30:1) node[circle,fill=black,draw,inner sep=0pt,minimum size=4pt] (w) {};
    	\path (t) ++(-60:1) node[circle,fill=black,draw,inner sep=0pt,minimum size=4pt] (x) {};
    	\path (t) ++(-90:1) node[circle,fill=black,draw,inner sep=0pt,minimum size=4pt] (y) {};

    	\node[circle,fill=black,draw,inner sep=0pt,minimum size=4pt] (z) at (4.8,2.8) {};
    	\node[circle,fill=black,draw,inner sep=0pt,minimum size=4pt] (aa) at (4,2.4) {};
    	\node[circle,fill=black,draw,inner sep=0pt,minimum size=4pt] (ab) at (4.8,2) {};
    	\node[circle,fill=black,draw,inner sep=0pt,minimum size=4pt] (ac) at (4,1.6) {};
    	\node[circle,fill=black,draw,inner sep=0pt,minimum size=4pt] (ad) at (4.8,1.2) {};
    	\node[circle,fill=black,draw,inner sep=0pt,minimum size=4pt] (ae) at (4.8,0.2) {};
    	
    	\path (ad) ++(-120:1) node[circle,fill=black,draw,inner sep=0pt,minimum size=4pt] (ag) {};
    	\path (ad) ++(-150:1) node[circle,fill=black,draw,inner sep=0pt,minimum size=4pt] (ah) {};
    	\path (ad) ++(-180:1) node[circle,fill=black,draw,inner sep=0pt,minimum size=4pt] (ai) {};
    	
    	\node[circle,fill=black,draw,inner sep=0pt,minimum size=4pt] (aj) at (3.4,0.4) {};
    	\node[circle,fill=black,draw,inner sep=0pt,minimum size=4pt] (ak) at (3,1.2) {};
    	\node[circle,fill=black,draw,inner sep=0pt,minimum size=4pt] (al) at (2.6,0.4) {};
    	\node[circle,fill=black,draw,inner sep=0pt,minimum size=4pt] (am) at (2.2,1.2) {};
    	
    	\begin{scope}[every node/.style={scale=0.7}]
    	\draw[ thick] (a) -- (b) -- (c) (b) -- (d) (b) to node[midway,inner sep=1pt,outer sep=1pt,minimum size=4pt,fill=white] {$\alpha_0$} (e) to node[midway,inner sep=1pt,outer sep=1pt,minimum size=4pt,fill=white] {$\beta_0$} (f) to node[midway,inner sep=1pt,outer sep=1pt,minimum size=4pt,fill=white] {$\alpha_0$} (g) to node[midway,inner sep=1pt,outer sep=1pt,minimum size=4pt,fill=white] {$\beta_0$} (h) to node[midway,inner sep=1pt,outer sep=1pt,minimum size=4pt,fill=white] {$\alpha_0$} (i) to node[midway,inner sep=1pt,outer sep=1pt,minimum size=4pt,fill=white] {$\beta_0$} (j) to node[midway,inner sep=1pt,outer sep=1pt,minimum size=4pt,fill=white] {$\alpha_0$} (k);
    	
    	\draw[ thick] (l) -- (k) -- (m) (k) -- (n) (k) to node[midway,inner sep=1pt,outer sep=1pt,minimum size=4pt,fill=white] {$\alpha_1$} (o) to node[midway,inner sep=1pt,outer sep=1pt,minimum size=4pt,fill=white] {$\beta_1$} (p) to node[midway,inner sep=1pt,outer sep=1pt,minimum size=4pt,fill=white] {$\alpha_1$} (q) to node[midway,inner sep=1pt,outer sep=1pt,minimum size=4pt,fill=white] {$\beta_1$} (r) to node[midway,inner sep=1pt,outer sep=1pt,minimum size=4pt,fill=white] {$\alpha_1$} (s) to node[midway,inner sep=1pt,outer sep=1pt,minimum size=4pt,fill=white] {$\beta_1$} (t);
    	
    	\draw[ thick] (v) -- (t) -- (w) (t) -- (x) (t) to node[midway,inner sep=1pt,outer sep=1pt,minimum size=4pt,fill=white] {$\alpha_2$} (y) to node[midway,inner sep=1pt,outer sep=1pt,minimum size=4pt,fill=white] {$\beta_2$} (z) to node[midway,inner sep=1pt,outer sep=1pt,minimum size=4pt,fill=white] {$\alpha_2$} (aa) to node[midway,inner sep=1pt,outer sep=1pt,minimum size=4pt,fill=white] {$\beta_2$} (ab) to node[midway,inner sep=1pt,outer sep=1pt,minimum size=4pt,fill=white] {$\alpha_2$} (ac) to node[midway,inner sep=1pt,outer sep=1pt,minimum size=4pt,fill=white] {$\beta_2$} (ad);
    	
    	\draw[thick] (ae) -- (ad) -- (ag) (ad) -- (ah) (ad) to node[midway,inner sep=1pt,outer sep=1pt,minimum size=4pt,fill=white] {$\alpha_3$} (ai) to node[midway,inner sep=1pt,outer sep=1pt,minimum size=4pt,fill=white] {$\beta_3$} (aj) to node[midway,inner sep=1pt,outer sep=1pt,minimum size=4pt,fill=white] {$\alpha_3$} (ak) to node[midway,inner sep=1pt,outer sep=1pt,minimum size=4pt,fill=white] {$\beta_3$} (al);
            \draw[thick, dotted] (al) -- (am);
        \end{scope}
    	
    	\end{tikzpicture}
    	\caption{After shifting}
    \end{subfigure}
    \caption{A 4-step Vizing chain before and after shifting.}\label{fig:multi_Viz_chain}
\end{figure}
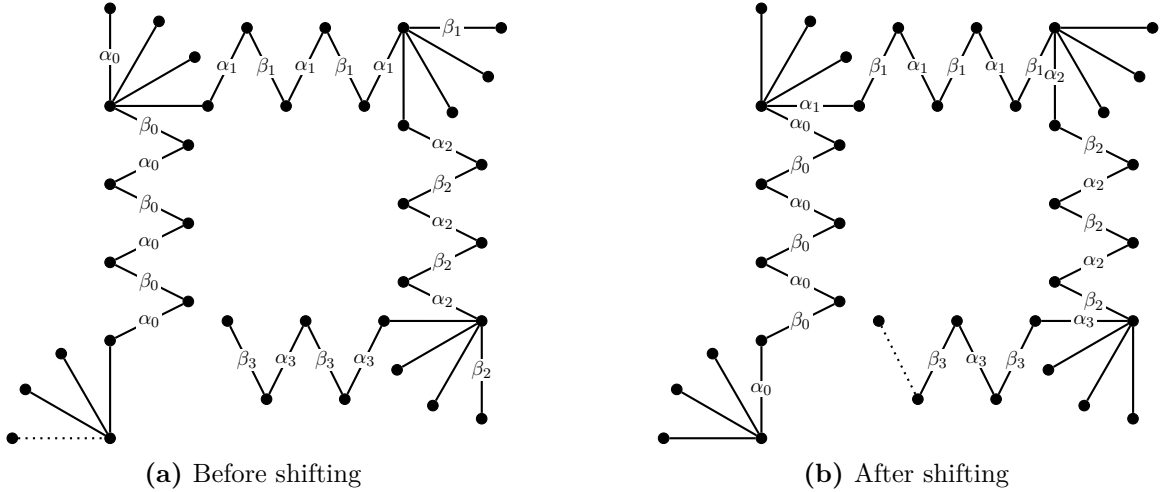

    As discussed in \S\ref{subsec:MSVC_history}, a central result of \cite{VizingChain} is that for each uncolored edge $e \in E$, one can find a $\phi$-happy $k$-step Vizing chain $C = F_0 + P_0 + \cdots + F_{k-1} + P_{k-1}$ with $\Start(C) = e$ such that \[k \,\leq\, \poly(\Delta)\log n \qquad \text{and} \qquad \length(P_i) \,\leq\, \poly(\Delta)\log n \text{ for all } 0 \leq i < k,\] and hence $\length(C) \leq \poly(\Delta) \log^2n$. In \cite{Christ}, Christiansen improved these bounds to \[k \,\leq\, \poly(\Delta)\log n \qquad \text{and} \qquad \length(P_i) \,\leq\, \poly(\Delta) \text{ for all } 0 \leq i < k,\] thus obtaining a $\phi$-happy multi-step Vizing chain $C$ with $\length(C) \leq \poly(\Delta) \log n$. In this paper we will present an efficient randomized algorithm that outputs such a chain.

\subsection{Fan Algorithms}\label{subsec:fan_algos}

We conclude this section by outlining two previously known procedures to construct fan chains, which will be used as subroutines in our main algorithms.

    We start with the \hyperref[alg:first_fan]{First Fan Algorithm}, which is described formally in Algorithm \ref{alg:first_fan}. The algorithm takes as input a proper partial coloring $\phi$, an uncolored edge $e = xy$, and a choice of a pivot vertex $x \in e$.
The output is a tuple $(F, \beta, j)$ such that:
\begin{itemize}
    \item $F$ is a fan with $\Start(F) = e$ and $\Pivot(F) = x$,
    \item $\beta \in [\Delta + 1]$ is a color and $j$ is an index such that $\beta \in M(\phi, \vend(F))\cap M(\phi, \vend(F|j))$.
\end{itemize}
    
We first assign a color $\beta(z) \in M(\phi, z)$ to each $z\in N_G(x)$. 
This assignment can be arbitrary, so we will let $\beta(z)$ be the minimum element in $M(\phi, z)$. To construct $F$, we follow a series of iterations. At the start of each iteration, we have a fan $F = (xy_0, \ldots, xy_k)$ where $xy_0 = e$. 
If $\beta(y_k) \in M(\phi, x)$, then $F$ is $\phi$-happy and we return $(F, \beta(y_k), k+1)$.
If not, let $z \in N_G(x)$ be the unique neighbor of $x$ such that $\phi(xz) = \beta(y_k)$. We now have two cases.
\begin{enumerate}[label=\ep{\textbf{Case \arabic*}},wide]
    \item $z\notin \set{y_0, \ldots, y_k}$. Then we update $F$ to $(xy_0, \ldots, xy_k, xz)$ and continue.
    \item $z = y_j$ for some $0 \leq j \leq k$. Note that $\phi(xy_0) = \blank$ and $\beta(y_k) \in M(\phi, y_k)$, so we must have $1 \leq j < k$. 
    In this case, we return $(F, \beta(y_k), j)$.
\end{enumerate}

\begin{algorithm}[h]\algsize
\caption{First Fan}\label{alg:first_fan}
\begin{flushleft}
\textbf{Input}: A proper partial edge-coloring $\phi$, an uncolored edge $e = xy$, and a vertex $x \in e$. \\
\textbf{Output}: A fan $F$ with $\Start(F) = e$ and $\Pivot(F) = x$, a color $\beta \in [\Delta + 1]$, and an index $j$ such that $\beta \in M(\phi, \vend(F))\cap M(\phi, \vend(F|j))$. 
\end{flushleft}


\begin{algorithmic}[1]
    \State $\mathsf{nbr}(\eta) \gets \blank$ \textbf{for each} $\eta \in [\Delta+1]$, \quad $\beta(z) \gets \blank$ \textbf{for each} $z \in N_G(x)$ \label{step:start}
    \For{$z\in N_G(x)$}\label{step:assign_beta}
        \State $\beta(z) \gets \min M(\phi, z)$
        \State $\mathsf{nbr}(\phi(xz)) \gets z$
    \EndFor
    \medskip
    \State $\mathsf{index}(z) \gets \blank$ \textbf{for each} $z \in N_G(x)$
    \State $F \gets (xy)$, \quad $k \gets 0$, \quad $z \gets y$, \quad $\mathsf{index}(z) \gets k$
    \While{$ k < \deg_G(x)$}
        \State $\eta \gets \beta(z)$
        \If{$\eta \in M(\phi, x)$}
            \State \Return $(F, \eta, k + 1)$
        \EndIf
        \State $z\gets \mathsf{nbr}(\eta)$
        \If{$\mathsf{index}(z) \neq \blank$}
            \State \Return $(F, \eta, \mathsf{index}(z))$
        \EndIf
        \State $k \gets k + 1$, \quad $\mathsf{index}(z) \gets k$
        \State $\mathsf{append}(F, xz)$
    \EndWhile \label{step:end}
\end{algorithmic}
\end{algorithm}

This procedure is identical to the one described and analyzed in \cite[Lemma 4.8]{VizingChain}. In particular, we have the following lemma as a result.

\begin{Lemma}[{\cite[Lemma 4.8]{VizingChain}}]\label{lemma:first_fan_ber}
    Let $(F, \beta, j)$ be the output of Algorithm \ref{alg:first_fan} on input $(\phi, xy, x)$. Then either $\beta \in M(\phi, x)$ and $F$ is $\phi$-happy, or else, for $F' \defeq F|j$ and any $\alpha \in M(\phi, x)$, we have that either $F$ or $F'$ is $(\phi, \alpha\beta)$-successful.
\end{Lemma}

By storing $M(\phi, z)$ as a hash map with key set $V$, we can compute the minimum element of $M(\phi, z)$ in time $O(\Delta)$. 
Therefore, the first loop in Algorithm~\ref{alg:first_fan} can run in $O(\Delta^2)$ time. We store $\beta(\cdot)$ and $\mathsf{nbr}(\cdot)$ as hash maps as well. Each operation in the \textsf{while} loop takes $O(1)$ time, and so the second loop runs in $O(\Delta)$ time, implying the overall runtime is $O(\Delta^2)$.

    The next algorithm we describe is the \hyperref[alg:next_fan]{Next Fan Algorithm}, Algorithm~\ref{alg:next_fan}. It will be used to construct the next fan on every iteration of our procedure for building multi-step Vizing chains;
see \S\ref{sec:algorithms} for details. This algorithm is borrowed from the paper \cite{GP} by Greb\'ik and Pikhurko, where it is employed in the construction of two-step Vizing chains with desirable properties. Algorithm~\ref{alg:next_fan} takes as input a proper partial edge-coloring $\phi$, an uncolored edge $e = xy$ such that $M(\phi, x) \setminus M(\phi, y) \neq \0$, a vertex $x \in e$, and a color $\beta \in M(\phi, y)$.
The output is a tuple $(F, \delta, j)$ such that:
    \begin{itemize}
        \item $F$ is a fan with $\Start(F) = e$ and $\Pivot(F) = x$,
        \item $\delta \in [\Delta + 1]$ is a color and $j$ is an index such that $\delta \in M(\phi, \vend(F))\cap M(\phi, \vend(F|j))$.
    \end{itemize}
The algorithm is similar to Algorithm \ref{alg:first_fan} with a few minor differences. 
First, we use $\delta(\cdot)$ in place of $\beta(\cdot)$ as a notational change.
Second, we have a restriction on the choice of $\delta(y)$, namely that $\delta(y) \neq \beta$; note that as $y$ is incident to an uncolored edge, $|M(\phi, y)| \geq 2$ and so $\delta(y)$ is well-defined. 
Finally, in the loop we additionally check whether $\delta(z) = \beta$, in which case the algorithm returns $(F, \beta, k+1)$.

\vspace{0.1in}
\begin{breakablealgorithm}\algsize
\caption{Next Fan}\label{alg:next_fan}
\begin{flushleft}
\textbf{Input}: A proper partial coloring $\phi$, an uncolored edge $e = xy$, a vertex $x \in e$, a color $\beta \in M(\phi, y)$. \\
\textbf{Output}: A fan $F$ with $\Start(F) = e$ and $\Pivot(F) = x$, a color $\delta \in [\Delta + 1]$, and an index $j$ such that $\delta \in M(\phi, \vend(F)) \cap M(\phi, \vend(F|j))$.
\end{flushleft}
\begin{algorithmic}[1]
    \State $\mathsf{nbr}(\eta) \gets \blank$ \textbf{for each} $\eta \in [\Delta+1]$, \quad $\delta(z) \gets \blank$ \textbf{for each} $z \in N_G(x)\setminus\set{y}$
    \State $\delta(y) \gets \min M(\phi, y)\setminus\{\beta\}$
    \For{$z\in N_G(x)\setminus \{y\}$}\label{step:assign_delta}
        \State $\delta(z) \gets \min M(\phi, z)$
        \State $\mathsf{nbr}(\phi(xz)) \gets z$
    \EndFor
    \medskip
    \State $\mathsf{index}(z) \gets \blank$ \textbf{for each} $z \in N_G(x)$
    \State $F \gets (xy)$, \quad $k \gets 0$, \quad $z \gets y$, \quad $\mathsf{index}(z) \gets k$
    \While{$k < \deg_G(x)$}
        \State $\eta \gets \delta(z)$
        \If{$\eta \in M(\phi, x)$}
            \State \Return $(F, \eta, k+1)$
        \ElsIf{$\eta = \beta$}\label{step:same_colors} 
            \State \Return $(F, \eta, k+1)$
        \EndIf
        \State $z\gets \mathsf{nbr}(\eta)$
        \If{$\mathsf{index}(z) \neq \blank$}
            \State \Return $(F, \eta, \mathsf{index}(z))$
        \EndIf
        \State $k \gets k + 1$, \quad $\mathsf{index}(z) \gets k$
        \State $\mathsf{append}(F, xz)$
    \EndWhile
\end{algorithmic}
\end{breakablealgorithm}
\vspace{0.1in}

This procedure is identical to the one analyzed in the proof of \cite[Lemma 4.9]{VizingChain}, hence: 
\begin{Lemma}[{\cite[Lemma 4.9]{VizingChain}}]\label{lemma:next_fan_ber}
    Let $(F, \delta, j)$ be the output of Algorithm \ref{alg:next_fan} on input $(\phi, xy, x, \beta)$, let $F' \defeq F|j$, and let $\alpha \in M(\phi, x) \setminus M(\phi, y)$ be arbitrary.
    Then no edge in $F$ is colored $\alpha$ or $\beta$ and at least one of the following statements holds:
    \begin{itemize}
        \item $F$ is $\phi$-happy, or
        \item $\delta = \beta$ and the fan $F$ is $(\phi, \alpha\beta)$-hopeful, or
        \item $\delta \neq \beta$ and either $F$ or $F'$ is $(\phi, \gamma\delta)$-successful for all $\gamma \in M(\phi, x) \setminus \{\alpha\}$.
    \end{itemize}
\end{Lemma}

The upshot is that Lemma~\ref{lemma:next_fan_ber} gives us extra information about the pair of colors $\gamma$, $\delta$ such that the fan $F$ or $F'$ is $(\phi, \gamma\delta)$-hopeful/successful. Namely, it guarantees that either $\set{\gamma, \delta} = \set{\alpha, \beta}$ or $\set{\gamma, \delta} \cap \set{\alpha, \beta} = \0$. This allows us to control the way the alternating paths in two consecutive steps of the construction may intersect; see Lemma~\ref{lemma:intersection_prev} and Fig.~\ref{fig:intersection_prev}.
Similarly to Algorithm \ref{alg:first_fan}, one can easily bound the runtime of Algorithm \ref{alg:next_fan} by $O(\Delta^2)$.

\section{Proof of Theorem \ref{theo:nlogn}}\label{sec:nlogn}

In this section, we prove Theorem \ref{theo:nlogn}. Specifically, we carefully describe and analyze a randomized edge-coloring algorithm based on the sketch in Algorithm~\ref{inf:simple_seq}  (see \S\ref{sec:overview}). The algorithm takes as input a graph $G$ of maximum degree $\Delta$ and returns a proper $(\Delta + 1)$-edge-coloring of $G$.
At each iteration, the algorithm picks a random uncolored edge and colors it by finding a Vizing chain.

We first give an overview of the \hyperref[alg:viz_chain]{Vizing Chain Algorithm}, which is formally presented as Algorithm~\ref{alg:viz_chain}.
It takes as input a proper partial edge-coloring $\phi$, an uncolored edge $e = xy$, and a pivot vertex $x\in e$.
We start by applying the \hyperref[alg:first_fan]{First Fan Algorithm} (Algorithm \ref{alg:first_fan}) as a subprocedure. Let $(F, \beta, j)$ be its output.
If $\beta \in M(\phi, x)$, we return the $\phi$-happy fan $F$.
If not, we pick a color $\alpha \in M(\phi, x)$ uniformly at random.
By Lemma \ref{lemma:first_fan_ber}, either $F$ or $F'\defeq F|j$ is $(\phi, \alpha\beta)$-successful.
Let \[P \,\defeq\, P(\End(F); \Shift(\phi, F), \alpha\beta) \quad \text{and} \quad P' \,\defeq\, P(\End(F'); \Shift(\phi, F'), \alpha\beta)\] (this notation is defined in \S\ref{subsec:pathchains}). If $\vend(P) \neq x$, we return $F + P$; otherwise, we return $F' + P'$.

\vspace{0.1in}
\begin{breakablealgorithm}[H]\algsize
\caption{Vizing Chain}\label{alg:viz_chain}
\begin{flushleft}
\textbf{Input}: A proper partial coloring $\phi$, an uncolored edge $e = xy$, and a vertex $x \in e$. \\
\textbf{Output}: A Vizing chain $F+P$ with $\Start(F) = e$ and $\Pivot(F) = x$. 
\end{flushleft}
\begin{algorithmic}[1]
    \State $(F, \beta, j) \gets \hyperref[alg:first_fan]{\mathsf{FirstFan}}(\phi, xy, x)$ \Comment{Algorithm \ref{alg:first_fan}}
    \If{$\beta \in M(\phi, x)$}
        \State \Return $F$ \label{step:happy}
    \EndIf
    \State Let $\alpha \in M(\phi, x)$ be chosen uniformly at random.
    \State $F' \gets F|j$
    \State $P \gets P(\End(F);\, \Shift(\phi, F),\, \alpha\beta), \quad P' \gets P(\End(F');\, \Shift(\phi, F'),\, \alpha\beta)$
    \If{$\vend(P) \neq x$}
        \State \Return $F+P$
    \Else
        \State \Return $F' + P'$
    \EndIf
\end{algorithmic}
\end{breakablealgorithm}
\vspace{0.1in}

With an appropriate choice of data structures, we can compute $P$ and $P'$ in time $O(\length(P))$ and $O(\length(P'))$ respectively (see, e.g., \cite[\S 3]{bernshteyn2024linear}).
Since the runtime of Algorithm \ref{alg:first_fan} is $O(\Delta^2)$, the runtime for Algorithm \ref{alg:viz_chain} is 
\begin{align}\label{eqn:runtime_viz_chain}
    O(\Delta^2 + \length(P) + \length(P')).
\end{align}
With this subprocedure defined, we can describe Algorithm \ref{alg:seq_log}, which iteratively colors the graph.

\begin{algorithm}[H]\algsize
\caption{Sequential Coloring with Vizing Chains}\label{alg:seq_log}
\begin{flushleft}
\textbf{Input}: A graph $G = (V, E)$ of maximum degree $\Delta$. \\
\textbf{Output}: A proper $(\Delta+1)$-edge-coloring $\phi$ of $G$.
\end{flushleft}
\begin{algorithmic}[1]
    \State $U \gets E$, \quad $\phi(e) \gets \blank$ \textbf{for each} $e \in U$
    \While{$U \neq \0$}
        \State Pick an edge $e \in U$ and a vertex $x \in e$ uniformly at random.
        \State $C \gets \hyperref[alg:viz_chain]{\mathsf{Vizing Chain}}(\phi, e, x)$ \Comment{Algorithm \ref{alg:viz_chain}}
        \State $\phi \gets \aug(\phi, C)$
        \State $U \gets U \setminus \set{e}$
    \EndWhile
    \State \Return $\phi$
\end{algorithmic}
\end{algorithm}

The correctness of this algorithm follows from Lemma \ref{lemma:first_fan_ber}. To assist with our analysis, we define the following variables:
\begin{align*}
    \phi_i &\defeq \text{ the coloring at the start of iteration } i, \\
    e_i &\defeq \text{ the edge chosen at iteration } i, \\
    x_i &\defeq \text{ the vertex chosen from } e_i = x_iy_i, \\
    \alpha_i &\defeq \text{ the color chosen from } M(\phi_i, x_i), \\
    \mathcal{P}_i &\defeq \set{P\,:\, \exists \alpha,\, \beta \in [\Delta+1] \text{ such that } P \text{ is a maximal path-component of } G(\phi_i, \alpha\beta)}, \\ 
    F_i,&\,P_i,\,F_i',\,P_i' \defeq \text{ the fans and paths computed in Algorithm \ref{alg:viz_chain} during iteration } i,  \\
    U_i &\defeq \set{e\in E\,:\, \phi_i(e) = \blank}.
\end{align*}

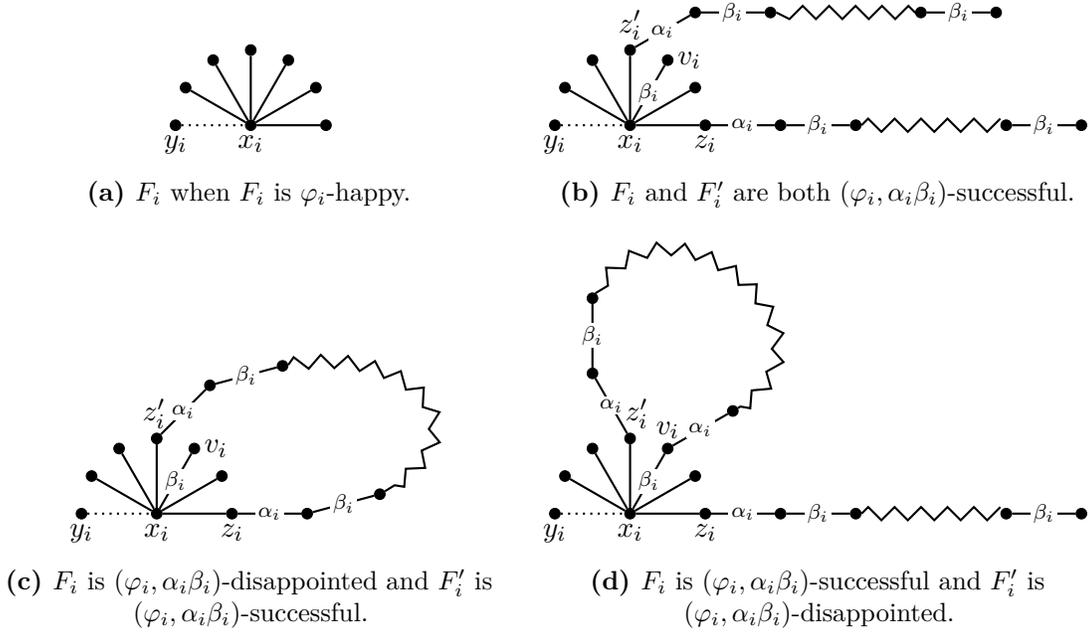
\begin{figure}[t]
    \begin{subfigure}[t]{0.45\textwidth}
        \centering
    	\begin{tikzpicture}
    	    \clip (-0.5, 1.5) rectangle (2.5, -0.5);
    	    \node[circle,fill=black,draw,inner sep=0pt,minimum size=4pt] (a) at (0,0) {};
    		\node[circle,fill=black,draw,inner sep=0pt,minimum size=4pt] (b) at (1,0) {}; \path (b) ++(150:1) node[circle,fill=black,draw,inner sep=0pt,minimum size=4pt] (c) {}; \path (b) ++(120:1) node[circle,fill=black,draw,inner sep=0pt,minimum size=4pt] (d) {}; \path (b) ++(90:1) node[circle,fill=black,draw,inner sep=0pt,minimum size=4pt] (e) {}; \path (b) ++(60:1) node[circle,fill=black,draw,inner sep=0pt,minimum size=4pt] (f) {}; \path (b) ++(30:1) node[circle,fill=black,draw,inner sep=0pt,minimum size=4pt] (g) {}; \path (b) ++(0:1) node[circle,fill=black,draw,inner sep=0pt,minimum size=4pt] (h) {};
    		
    		\draw[thick, dotted] (a) -- (b);
    		
    		\draw[thick] (b) -- (c) (b) -- (d) (b) -- (e) (b) -- (f) (b) -- (g) (b) -- (h);

                \node[anchor=north] at (b) {$x_i$};
                \node[anchor=north] at (a) {$y_i$};

    	\end{tikzpicture}
    	\caption{$F_i$ when $F_i$ is $\phi_i$-happy.}\label{subfig:happy}
    \end{subfigure}
    \begin{subfigure}[t]{0.45\textwidth}
        \centering
    	\begin{tikzpicture}
    	    \clip (-0.2, 2) rectangle (7.5, -0.5);
    	    \node[circle,fill=black,draw,inner sep=0pt,minimum size=4pt] (a) at (0,0) {};
            \node[circle,fill=black,draw,inner sep=0pt,minimum size=4pt] (b) at (1,0) {}; \path (b) ++(150:1) node[circle,fill=black,draw,inner sep=0pt,minimum size=4pt] (c) {}; \path (b) ++(120:1) node[circle,fill=black,draw,inner sep=0pt,minimum size=4pt] (d) {}; \path (b) ++(90:1) node[circle,fill=black,draw,inner sep=0pt,minimum size=4pt] (e) {}; \path (b) ++(60:1) node[circle,fill=black,draw,inner sep=0pt,minimum size=4pt] (f) {}; \path (b) ++(30:1) node[circle,fill=black,draw,inner sep=0pt,minimum size=4pt] (g) {}; \path (b) ++(0:1) node[circle,fill=black,draw,inner sep=0pt,minimum size=4pt] (h) {};

            \path (h) ++(0:1) node[circle,fill=black,draw,inner sep=0pt,minimum size=4pt] (i) {};
            \path (i) ++(0:1) node[circle,fill=black,draw,inner sep=0pt,minimum size=4pt] (j) {};
            \path (j) ++(0:2) node[circle,fill=black,draw,inner sep=0pt,minimum size=4pt] (k) {};
            \path (k) ++(0:1) node[circle,fill=black,draw,inner sep=0pt,minimum size=4pt] (l) {};

            \path (e) ++(30:1) node[circle,fill=black,draw,inner sep=0pt,minimum size=4pt] (m) {};
            \path (m) ++(0:1) node[circle,fill=black,draw,inner sep=0pt,minimum size=4pt] (n) {};
            \path (n) ++(0:2) node[circle,fill=black,draw,inner sep=0pt,minimum size=4pt] (o) {};
            \path (o) ++(0:1) node[circle,fill=black,draw,inner sep=0pt,minimum size=4pt] (p) {};
    		
    		\draw[thick, dotted] (a) -- (b);
    		
    		\draw[thick] (b) -- (c) (b) -- (d) (b) -- (e) (b) to node[font=\fontsize{8}{8},midway,inner sep=1pt,outer sep=1pt,minimum size=4pt,fill=white] {$\beta_i$} (f) (b) -- (g) (b) -- (h) to node[font=\fontsize{8}{8},midway,inner sep=1pt,outer sep=1pt,minimum size=4pt,fill=white] {$\alpha_i$} (i) to node[font=\fontsize{8}{8},midway,inner sep=1pt,outer sep=1pt,minimum size=4pt,fill=white] {$\beta_i$} (j) (k) to node[font=\fontsize{8}{8},midway,inner sep=1pt,outer sep=1pt,minimum size=4pt,fill=white] {$\beta_i$} (l) (e) to node[font=\fontsize{8}{8},midway,inner sep=1pt,outer sep=1pt,minimum size=4pt,fill=white] {$\alpha_i$} (m) to node[font=\fontsize{8}{8},midway,inner sep=1pt,outer sep=1pt,minimum size=4pt,fill=white] {$\beta_i$} (n) (o) to node[font=\fontsize{8}{8},midway,inner sep=1pt,outer sep=1pt,minimum size=4pt,fill=white] {$\beta_i$} (p);
    		
    		\draw[thick, decorate,decoration=zigzag] (j) -- (k) (n) -- (o);

                \node[anchor=north] at (b) {$x_i$};
                \node[anchor=north] at (a) {$y_i$};
                \node[anchor=south] at (e) {$z_i'$};
                \node[anchor=north] at (h) {$z_i$};
                \node[anchor=west] at (f) {$v_i$};

    	\end{tikzpicture}
    	\caption{$F_i$ and $F_i'$ are both $(\phi_i, \alpha_i\beta_i)$-successful.}\label{subfig:successful}
    \end{subfigure}
    
    \begin{subfigure}[t]{0.45\textwidth}
        \centering
    	\begin{tikzpicture}
    	    \clip (-0.5, 2.5) rectangle (5, -0.5);
    	    \node[circle,fill=black,draw,inner sep=0pt,minimum size=4pt] (a) at (0,0) {};
            \node[circle,fill=black,draw,inner sep=0pt,minimum size=4pt] (b) at (1,0) {}; \path (b) ++(150:1) node[circle,fill=black,draw,inner sep=0pt,minimum size=4pt] (c) {}; \path (b) ++(120:1) node[circle,fill=black,draw,inner sep=0pt,minimum size=4pt] (d) {}; \path (b) ++(90:1) node[circle,fill=black,draw,inner sep=0pt,minimum size=4pt] (e) {}; \path (b) ++(60:1) node[circle,fill=black,draw,inner sep=0pt,minimum size=4pt] (f) {}; \path (b) ++(30:1) node[circle,fill=black,draw,inner sep=0pt,minimum size=4pt] (g) {}; \path (b) ++(0:1) node[circle,fill=black,draw,inner sep=0pt,minimum size=4pt] (h) {};

            \path (h) ++(0:1) node[circle,fill=black,draw,inner sep=0pt,minimum size=4pt] (i) {};
            \path (i) ++(15:1) node[circle,fill=black,draw,inner sep=0pt,minimum size=4pt] (j) {};

            \path (e) ++(45:1) node[circle,fill=black,draw,inner sep=0pt,minimum size=4pt] (k) {};
            \path (k) ++(15:1) node[circle,fill=black,draw,inner sep=0pt,minimum size=4pt] (l) {};
    		
    		\draw[thick, dotted] (a) -- (b);
    		
    		\draw[thick] (b) -- (c) (b) -- (d) (b) -- (e) (b) to node[font=\fontsize{8}{8},midway,inner sep=1pt,outer sep=1pt,minimum size=4pt,fill=white] {$\beta_i$} (f) (b) -- (g) (b) -- (h) to node[font=\fontsize{8}{8},midway,inner sep=1pt,outer sep=1pt,minimum size=4pt,fill=white] {$\alpha_i$} (i) to node[font=\fontsize{8}{8},midway,inner sep=1pt,outer sep=1pt,minimum size=4pt,fill=white] {$\beta_i$} (j) (e)  to node[font=\fontsize{8}{8},midway,inner sep=1pt,outer sep=1pt,minimum size=4pt,fill=white] {$\alpha_i$} (k) to node[font=\fontsize{8}{8},midway,inner sep=1pt,outer sep=1pt,minimum size=4pt,fill=white] {$\beta_i$} (l);

            \draw[thick, decorate,decoration=zigzag] (l) to[in=30, out=10, looseness=2] (j);

            \node[anchor=north] at (b) {$x_i$};
            \node[anchor=north] at (a) {$y_i$};
            \node[anchor=south, xshift=-1pt] at (e) {$z_i'$};
            \node[anchor=north] at (h) {$z_i$};
            \node[anchor=west] at (f) {$v_i$};

    	\end{tikzpicture}
    	\caption{$F_i$ is $(\phi_i, \alpha_i\beta_i)$-disappointed and $F_i'$ is $(\phi_i, \alpha_i\beta_i)$-successful.}\label{subfig:prime_successful}
    \end{subfigure}
    \begin{subfigure}[t]{0.45\textwidth}
        \centering
    	\begin{tikzpicture}
    	    \clip (-0.2, 4) rectangle (7.5, -0.5);
    	    \node[circle,fill=black,draw,inner sep=0pt,minimum size=4pt] (a) at (0,0) {};
            \node[circle,fill=black,draw,inner sep=0pt,minimum size=4pt] (b) at (1,0) {}; \path (b) ++(150:1) node[circle,fill=black,draw,inner sep=0pt,minimum size=4pt] (c) {}; \path (b) ++(120:1) node[circle,fill=black,draw,inner sep=0pt,minimum size=4pt] (d) {}; \path (b) ++(90:1) node[circle,fill=black,draw,inner sep=0pt,minimum size=4pt] (e) {}; \path (b) ++(60:1) node[circle,fill=black,draw,inner sep=0pt,minimum size=4pt] (f) {}; \path (b) ++(30:1) node[circle,fill=black,draw,inner sep=0pt,minimum size=4pt] (g) {}; \path (b) ++(0:1) node[circle,fill=black,draw,inner sep=0pt,minimum size=4pt] (h) {};

            \path (h) ++(0:1) node[circle,fill=black,draw,inner sep=0pt,minimum size=4pt] (i) {};
            \path (i) ++(0:1) node[circle,fill=black,draw,inner sep=0pt,minimum size=4pt] (j) {};
            \path (j) ++(0:2) node[circle,fill=black,draw,inner sep=0pt,minimum size=4pt] (k) {};
            \path (k) ++(0:1) node[circle,fill=black,draw,inner sep=0pt,minimum size=4pt] (l) {};

            \path (e) ++(120:1) node[circle,fill=black,draw,inner sep=0pt,minimum size=4pt] (m) {};
            \path (m) ++(90:1) node[circle,fill=black,draw,inner sep=0pt,minimum size=4pt] (n) {};

            \path (f) ++(30:1) node[circle,fill=black,draw,inner sep=0pt,minimum size=4pt] (o) {};
    		
    		\draw[thick, dotted] (a) -- (b);
    		
    		\draw[thick] (b) -- (c) (b) -- (d) (b) -- (e) (b) to node[font=\fontsize{8}{8},midway,inner sep=1pt,outer sep=1pt,minimum size=4pt,fill=white] {$\beta_i$} (f) (b) -- (g) (b) -- (h) to node[font=\fontsize{8}{8},midway,inner sep=1pt,outer sep=1pt,minimum size=4pt,fill=white] {$\alpha_i$} (i) to node[font=\fontsize{8}{8},midway,inner sep=1pt,outer sep=1pt,minimum size=4pt,fill=white] {$\beta_i$} (j) (k) to node[font=\fontsize{8}{8},midway,inner sep=1pt,outer sep=1pt,minimum size=4pt,fill=white] {$\beta_i$} (l) (e) to node[font=\fontsize{8}{8},midway,inner sep=1pt,outer sep=1pt,minimum size=4pt,fill=white] {$\alpha_i$} (m) to node[font=\fontsize{8}{8},midway,inner sep=1pt,outer sep=1pt,minimum size=4pt,fill=white] {$\beta_i$} (n) (f) to node[font=\fontsize{8}{8},midway,inner sep=1pt,outer sep=1pt,minimum size=4pt,fill=white] {$\alpha_i$} (o);

            \draw[thick, decorate,decoration=zigzag] (j) -- (k) (n) to[in=30, out=60, looseness=2] (o);

            \node[anchor=north] at (b) {$x_i$};
            \node[anchor=north] at (a) {$y_i$};
            \node[anchor=south, xshift=2.5pt] at (e) {$z_i'$};
            \node[anchor=north] at (h) {$z_i$};
            \node[anchor=south] at (f) {$v_i$};

    	\end{tikzpicture}
    	\caption{$F_i$ is $(\phi_i, \alpha_i\beta_i)$-successful and $F_i'$ is $(\phi_i, \alpha_i\beta_i)$-disappointed.}\label{subfig:main_successful}
    \end{subfigure}
    \caption{Possible configurations of $F_i + P_i$ and $F_i'+P_i'$.}
    \label{fig:possible_path_chains}
\end{figure}
\noindent
In the case that we reach step~\ref{step:happy}, we let $F_i' = F_i$ and $P_i = P_i' = (\End(F_i))$.
It follows from \eqref{eqn:runtime_viz_chain} that we can bound the runtime of Algorithm \ref{alg:seq_log} as
\[\sum_{i = 1}^mO(\Delta^2 + \length(P_i) + \length(P_i')) \,=\, O(\Delta^2\,m) + O(\sum_{i = 1}^m\length(P_i)) + O(\sum_{i = 1}^m\length(P_i')).\]
Our goal now is to bound the lengths of the path chains $P_i$ and $P_i'$.
Let $\tilde{P}_i$ (resp.{}\ $\tilde{P}_i'$) be the path obtained by removing the first edge on the chain $P_i$ (resp.{}\ $P_i'$).
Note that $\tilde P_i$ (resp.{}\ $\tilde P_i'$) is a bicolored path in $\Shift(\phi_i, F_i)$ (resp.{}\ $\Shift(\phi_i, F_i')$).
Let the second color on these paths be $\beta_i$.
We wish to estimate the lengths of the paths $\tilde P_i$, $\tilde P_i'$ in terms of some paths $Q_i$, $Q_i' \in \mathcal{P}_i$.
As a result of Lemma~\ref{lemma:first_fan_ber}, we have four cases as shown in Fig.~\ref{fig:possible_path_chains}.
Let us consider each case carefully.
\begin{enumerate}[label=\ep{\textbf{Case \arabic*}},wide]
    \item $F_i$ is $\phi_i$-happy (Fig.~\ref{subfig:happy}).
    Here no path chain is computed.
    We let $Q_i$ and $Q_i'$ be both equal to the $\alpha_i\beta_i$-path starting at $\vend(F_i)$.
    We note that in this case $1 = \length(P_i) \leq \length(Q_i) + 1$, and the same holds for $P_i'$ and $Q_i'$.

    \item\label{item:case_successful} $F_i$ and $F_i'$ are both $(\phi_i, \alpha_i\beta_i)$-successful (Fig.~\ref{subfig:successful}).
    As shown in the figure, $\tilde P_i$ is the path starting at $z_i = \vend(F_i)$ and $\tilde P_i'$ is the path starting at $z_i' = \vend(F_i')$.
    In this case, both $\tilde P_i$ and $\tilde P_i'$ are in $\mathcal{P}_i$, so we let $Q_i \defeq \tilde P_i$, $Q_i' \defeq \tilde P_i'$.
    Note that $\length(P_i) = \length(Q_i) + 1$ and the same holds for $P_i'$, $Q_i'$ in this case.

    \item $F_i$ is $(\phi_i, \alpha_i\beta_i)$-disappointed and $F_i'$ is $(\phi_i, \alpha_i\beta_i)$-successful (Fig.~\ref{subfig:prime_successful}).
    Here, $\tilde P_i$ is the path starting at $z_i$ passing through $z_i'$ to end at $x_i$ (as $\Shift(\phi_i, F_i)(x_iz_i') = \beta_i$), and $\tilde P_i'$ is the path starting at $z_i'$ and ending at $z_i$.
    In this case, we define $Q_i = Q_i' \defeq \tilde P_i'$ and note that $\length(P_i) = \length(Q_i) + 2$ and $\length(P_i') = \length(Q_i') + 1$.

    \item\label{item:case_successful_hopeful} $F_i'$ is $(\phi_i, \alpha_i\beta_i)$-disappointed and $F_i$ is $(\phi_i, \alpha_i\beta_i)$-successful (Fig.~\ref{subfig:main_successful}).
    In this case, we note as in \ref{item:case_successful} that both $\tilde P_i$ and $\tilde P_i'$ are in $\mathcal{P}_i$ and so we let $Q_i \defeq \tilde P_i$ and $Q_i' \defeq \tilde P_i'$, implying $\length(P_i) = \length(Q_i) + 1$ and $\length(P_i') = \length(Q_i') + 1$ in this case.
    
\end{enumerate}

After covering all cases, we have now defined $Q_i$, $Q_i' \in \mathcal{P}_i$ such that $\length(P_i) \leq \length(Q_i) + 2$ and $\length(P_i') \leq \length(Q_i') + 2$.
In particular, we can now bound the runtime as
\[O(\Delta^2m) + O(\sum_{i = 1}^m\length(Q_i)) + O(\sum_{i = 1}^m\length(Q_i')).\]
It remains to bound the terms on the right. To this end, let 
\[T_i \defeq \length(Q_i), \quad T_i' \defeq \length(Q_i'), \quad T \defeq \sum_{i = 1}^mT_i, \quad T' \defeq \sum_{i = 1}^mT_i'.\]
We must bound $T$ and $T'$ from above with high probability.
Let us consider $T$. (The analysis for $T'$ is the same, \emph{mutatis mutandis}.)
We have
\[T_i \,=\, \sum_{Q\in \mathcal{P}_i}\length(Q)\,\mathbbm{1}\set{Q_i = Q}.\]
In particular, $T_i$ is determined by $\phi_i$ and the random choices made in the $i$-th iteration.
With this in mind, note the following:
\begin{align}\label{eqn:expectation}
    \E[T_i|T_1, \ldots, T_{i-1}] \,=\, \E[\E[T_i|\phi_i, T_1, \ldots, T_{i-1}]|T_1, \ldots, T_{i-1}] \,=\, \E[\E[T_i|\phi_i]|T_1, \ldots, T_{i-1}],
\end{align}
which follows since $T_i$ is independent of all $T_j$ given $\phi_i$.
In the following lemma, we shall bound the random variable $\E[T_i|\phi_i]$.

\begin{Lemma}\label{lemma:expected_runtime}
    We have
    $\E[T_i|\phi_i] \,\leq\, \frac{2\Delta^2m}{m-i+1}$.
\end{Lemma}

\begin{proof}
    The random choices made are $e_i$, $x_i$ in Algorithm \ref{alg:seq_log} and $\alpha_i$ in Algorithm \ref{alg:viz_chain}.
    Let $Q \in \mathcal{P}_i$ have endpoints $u \neq v$. Let $\gamma_u \notin M(\phi_i, u)$ and $\gamma_v \notin M(\phi_i, v)$ where $\gamma_u$, $\gamma_v$ are colors on $Q$. For $Q_i$ to take the value $Q$, we must have $x_i \in N_G(u)$ and $\alpha_i = \gamma_u$ or $x_i \in N_G(v)$ and $\alpha_i = \gamma_v$. It follows that
    \begin{align*}
        \P[Q_i = Q|\phi_i] \,&\leq\,  \sum_{x \in N_G(u)}\P[x_i = x, \,\alpha_i = \gamma_u] + \sum_{x \in N_G(v)}\P[x_i = x, \, \alpha_i = \gamma_v] \\
        &\leq\, \sum_{x \in N_G(u)}\frac{|M(\phi_i, x)|}{2|U_i|}\,\frac{1}{|M(\phi_i, x)|} + \sum_{x \in N_G(v)}\frac{|M(\phi_i, x)|}{2|U_i|}\,\frac{1}{|M(\phi_i, x)|} \,\leq\, \frac{\Delta}{|U_i|}.
    \end{align*}
    The second inequality above follows due to the following:
    \begin{align*}
        \P[x_i = x] &\leq \frac{|M(\phi_i, x)|}{|U_i|}\,\frac{1}{2} \qquad\qquad \text{since $x$ is incident to at most $|M(\phi_i, x)|$ uncolored edges}, \\
        \P[\alpha_i = \gamma_u &\mid x_i = x] \leq \frac{1}{|M(\phi_i, x)|}.
    \end{align*}
    Let us count the number of paths in $\mathcal{P}_i$ passing through a specific $\phi_i$-colored edge $e$. An alternating path is uniquely determined by its (at most $2$) colors and a single edge on it. Given $e$, we know one of the colors and there are at most $\Delta$ choices for the second color (if it exists). It follows that $e$ can be in at most $\Delta+1 \leq 2\Delta$ alternating paths.
    With this in mind, we have
    \begin{align*}
        \E[T_i|\phi_i] \,
        \leq\, \sum_{Q \in \mathcal{P}_i}\frac{\Delta}{|U_i|}\,\length(Q) 
        \,&=\, \frac{\Delta}{|U_i|}\sum_{Q \in \mathcal{P}_i}\sum_{e \in E\setminus U_i}\mathbbm{1}\set{e \in E(Q)} \\
        \,&=\, \frac{\Delta}{|U_i|}\sum_{e \in E\setminus U_i}\sum_{Q \in \mathcal{P}_i}\mathbbm{1}\set{e \in E(Q)} 
        \,\leq\, \frac{2\,\Delta^2\,m}{|U_i|}.
    \end{align*}
    This completes the proof as $|U_i| = m - i + 1$.
\end{proof}

From \eqref{eqn:expectation} and Lemma \ref{lemma:expected_runtime}, it follows that
\[\E[T_i|T_1, \ldots, T_{i-1}] \,=\, \E[\E[T_i|\phi_i]|T_1, \ldots, T_{i-1}] \,\leq\, \frac{2\Delta^2\,m}{m-i+1}.\]
To finish the time analysis of Algorithm~\ref{alg:seq_log}, we shall apply the following concentration inequality due to Kuszmaul and Qi \cite{azuma}, which is a special case of their version of multiplicative Azuma's inequality for supermartingales:

\begin{theo}[{Kuszmaul--Qi \cite[Corollary 6]{azuma}}]\label{theo:azuma_supermartingale}
    Let $c > 0$ and let $X_1$, \ldots, $X_n$ be 
    random variables taking values in $[0,c]$.
    Suppose that $\E[X_i|X_1, \ldots, X_{i-1}] \leq a_i$ for all $i$.
    Let $\mu \defeq \sum_{i = 1}^na_i$. Then, for any $\delta > 0$,
    \[\P\left[\sum_{i = 1}^nX_i \geq (1+\delta)\mu\right] \,\leq\, \exp\left(-\frac{\delta^2\mu}{(2+\delta)c}\right).\]
\end{theo}

We trivially have $T_i \in [0,n]$. 
From Lemma \ref{lemma:expected_runtime}, it follows that we can apply Theorem \ref{theo:azuma_supermartingale} with 
\[X_i = T_i, \quad c = n, \quad a_i = \frac{2\Delta^2m}{m-i+1}.\]
Assuming $G$ has no isolated vertices, we must have $n/2 \leq m \leq \Delta\,n/2$.
With this in hand, from the bounds on partial sums of the harmonic series \ep{see, e.g., \cite[\S2.4]{bressoud2022radical}}, we have
\[\Delta^2n\log n \,\leq\, \mu \,\leq\, 2\Delta^3n\log n.\]
It follows that for any $\delta > 0$, 
\[\P[T \geq 2\,(1+\delta)\,\Delta^3n\log n] \,\leq\, \exp\left(-\frac{\delta^2\,\Delta^2\,n\log n}{(2+\delta)n}\right).\]
For a suitable choice of $\delta$, it follows that $T = O(\Delta^3\,n\log n)$ with probability at least $1 - 1/\poly(n)$. As mentioned above, exactly the same analysis applies to $T_i'$ in place of $T_i$. Therefore, Algorithm~\ref{alg:seq_log} terminates in $O(\Delta^3n\log n)$ time with probability at least $1 - 1/\poly(n)$, as desired.

\section{The Coloring Procedure}\label{sec:algorithms}

This section is split into two subsections. In the first subsection, we describe the \hyperref[alg:multi_viz_chain]{Multi-Step Vizing Algorithm} (MSVA), which constructs a $\phi$-happy multi-step Vizing chain given a partial coloring $\phi$, an uncolored edge $e$, and a starting vertex $x \in e$.  In later sections we will bound the length of the chain and the runtime of the algorithm in terms of $n$, $\Delta$, and a parameter $\ell$. 
The value for $\ell$ will be specified later; for now, we shall assume that $\ell$ is a polynomial in $\Delta$ of sufficiently high degree.\footnote{To be precise, we will choose $\ell = \Theta(\Delta^{16})$ for our algorithms, but we make no attempt to optimize the degree of the polynomial.}.
In the second subsection, we prove the correctness of the \hyperref[alg:multi_viz_chain]{MSVA} as well as some technical lemmas regarding the chain computed that will be relevant to the remainder of the paper.

\subsection{The Multi-Step Vizing Algorithm}\label{subsec:algo_overview}

Our goal is to build a multi-step Vizing chain $C$. We will split the algorithm into subprocedures.
Essentially, these subprocedures involve the usual constructions of Vizing chains with an added truncation step as described in \S\ref{subsec:MSVC_inf}.
The first subprocedure (described formally in Algorithm~\ref{alg:first_chain}) is used to build the first Vizing chain on $C$. Given a partial coloring $\phi$, an uncolored edge $e = xy$, a vertex $x \in e$, and a parameter $\ell \in \N$, it returns a Vizing chain $F+P$ which satisfies certain properties which will be described and proved in \S\ref{subsec:algo_lemmas}.
The Algorithm is nearly identical to Algorithm \ref{alg:viz_chain} with the only differences that we do not compute the entire paths $P$, $P'$ and we do not pick the color $\alpha$ randomly.

\begin{algorithm}[h]\algsize
\caption{First Chain}\label{alg:first_chain}
\begin{flushleft}
\textbf{Input}: A proper partial edge-coloring $\phi$, an uncolored edge $e = xy$, a vertex $x \in e$, and a parameter $\ell \in \N$. \\
\textbf{Output}: A fan $F$ with $\Start(F) = e$ and $\Pivot(F) = x$ and a path $P$ with $\Start(P) = \End(F)$ and $\vstart(P) = \Pivot(F) = x$. 
\end{flushleft}
\begin{algorithmic}[1]
    \State $(F, \beta, j) \gets \hyperref[alg:first_fan]{\mathsf{FirstFan}}(\phi, xy, x)$ \Comment{Algorithm \ref{alg:first_fan}} \label{step:first_fan}
    \If{$\beta \in M(\phi, x)$}
        \State \Return $F$, $(\End(F))$\label{step:happy_first}
    \EndIf
    \State $\alpha \gets \min M(\phi, x)$
    \State $F' \gets F|j$
    \State $P \gets P(\End(F);\, \Shift(\phi, F),\, \alpha\beta), \quad P' \gets P(\End(F');\, \Shift(\phi, F'),\, \alpha\beta)$
    \If{$\length(P) > 2\ell$ \textbf{or} $\vend(P) \neq x$}\label{step:path_condition}
        \State \Return $F$, $P|2\ell$
    \Else
        \State \Return $F'$, $P'|2\ell$
    \EndIf
\end{algorithmic}
\end{algorithm}

    Notice that in Algorithm~\ref{alg:first_chain}, we do not need to compute the entire paths $P$, $P'$; instead, we only have to compute their initial segments of length $2\ell$, which can be done in time $O(\ell)$. The same reasoning as in the runtime analysis of Algorithm \ref{alg:viz_chain} shows that the running time of Algorithm \ref{alg:first_chain} is $O(\Delta^2+\ell) = O(\ell)$ as our choice for $\ell$ will be large \ep{considerably larger than $\Delta$}.

\begin{algorithm}[h]\algsize
\caption{Next Chain}\label{alg:next_chain}
\begin{flushleft}
\textbf{Input}: A proper partial edge-coloring $\phi$, an uncolored edge $e = xy$, a vertex $x \in e$, a parameter $\ell \in \N$, and a pair of colors $\alpha \in M(\phi, x) \setminus M(\phi, y)$ and $\beta \in M(\phi, y)$. \\
\textbf{Output}: A fan $F$ with $\Start(F) = e$ and $\Pivot(F) = x$ and a path $P$ with $\Start(P) = \End(F)$ and $\vstart(P) = \Pivot(F) = x$. 
\end{flushleft}
\begin{algorithmic}[1]
    \State $(F, \delta, j) \gets \hyperref[alg:next_fan]{\mathsf{NextFan}}(\phi, xy, x, \beta)$ \Comment{Algorithm \ref{alg:next_fan}} \label{step:next_fan}
    \If{$\delta \in M(\phi, x)$}\label{step:happy_next}
        \State \Return $F$, $(\End(F))$
    \EndIf
    \If{$\delta  = \beta$}
        \State $P \gets P(\End(F);\, \Shift(\phi, F),\, \alpha\beta)$
        \State \Return $F$, $P|2\ell$
    \EndIf
    \State $\gamma \gets \min M(\phi, x)\setminus\set{\alpha}$
    \State $F' \gets F|j$
    \State $P \gets P(\End(F); \, \Shift(\phi, F),\, \gamma\delta), \quad P' \gets P(\End(F'); \, \Shift(\phi, F'), \, \gamma\delta)$
    \If{$\length(P) > 2\ell$ \textbf{or} $\vend(P) \neq x$}\label{step:next_path_condition} 
        \State \Return $F$, $P|2\ell$
    \Else
        \State \Return $F'$, $P'|2\ell$
    \EndIf
\end{algorithmic}
\end{algorithm}

    The next algorithm we describe, namely Algorithm~\ref{alg:next_chain},
    is the subprocedure to build the ``next'' Vizing chain on $C$. Given a partial coloring $\phi$, an uncolored edge $e = xy$, a vertex $x \in e$, a parameter $\ell \in \N$, and a pair of colors $\alpha \in M(\phi, x) \setminus M(\phi, y)$, $\beta \in M(\phi, y)$, it returns a chain $F+P$ which satisfies certain properties that will be described and proved in \S\ref{subsec:algo_lemmas}. 
    The colors $\alpha$, $\beta$ represent the colors on the path of the previous Vizing chain and the coloring $\phi$ refers to the shifted coloring with respect to the multi-step Vizing chain computed so far.
    Following a nearly identical argument as for Algorithm \ref{alg:first_chain}, we see that the runtime of Algorithm \ref{alg:next_chain} is bounded by $O(\ell)$ as well.

Before we describe our \hyperref[alg:multi_viz_chain]{Multi-Step Vizing Algorithm}, we define the non-intersecting property of a chain $C$. Recall that $\IE(P)$ and $\IV(P)$ are the sets of internal edges and vertices of a path chain $P$, introduced in Definition~\ref{defn:internal}.

\begin{defn}[Non-intersecting chains]\label{defn:non-int}
    A $k$-step Vizing chain $C = F_0 + P_0 + \cdots + F_{k-1} + P_{k-1}$ is \emphd{non-intersecting} if for all $0\leq i < j < k$, $V(F_i) \cap V(F_j + P_j) = \0$ and $\IE(P_i) \cap E(F_j + P_j) = \0$.
\end{defn}

In our algorithm, we build a non-intersecting multi-step Vizing chain $C$. 
Before we state the algorithm rigorously, we provide an informal overview. To start with, we have a chain $C = (xy)$ containing just the uncolored edge. Using Algorithm~\ref{alg:first_chain}, we find the first Vizing chain $F+P$. 
With this chain defined, we begin iterating.

At the start of each iteration, we have a non-intersecting chain $C = F_0 + P_0 + \cdots + F_{k-1} +P_{k-1}$ and a \emphd{candidate chain} $F+P$ satisfying the following properties for $\psi \defeq \Shift(\phi, C)$:
\begin{enumerate}[label=\ep{\normalfont{}\texttt{Inv}\arabic*},labelindent=15pt,leftmargin=*]
    \item\label{inv:start_F_end_C} $\Start(F) = \End(C)$ and $\vstart(F) = \vend(C)$,
    \item\label{inv:non_intersecting_shiftable} $C + F + P$ is non-intersecting and $\phi$-shiftable, and 
    \item\label{inv:hopeful_length} $F$ is either $\psi$-happy or $(\psi, \alpha\beta)$-hopeful, where $P$ is an $\alpha\beta$-path; furthermore, if $F$ is $(\psi, \alpha\beta)$-disappointed, then $\length(P) = 2\ell$.
\end{enumerate}
In the next subsection, we will prove that these invariants hold. For now, let us take them to be true. Fig.~\ref{fig:iteration_start} shows an example of $C$ (in black) and $F+P$ (in red) at the start of an iteration.
We remark that in general, $C + F + P$ may be self-intersecting as we do allow paths on different chains to intersect at vertices.
For simplicity, we do not illustrate such examples in our figures.

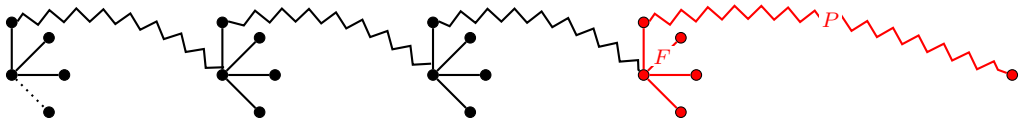
\begin{figure}[H]
    \centering
    \begin{tikzpicture}[xscale = 0.7,yscale=0.7]
        \node[circle,fill=black,draw,inner sep=0pt,minimum size=4pt] (a) at (0,0) {};
        	\path (a) ++(-45:1) node[circle,fill=black,draw,inner sep=0pt,minimum size=4pt] (b) {};
        	\path (a) ++(0:1) node[circle,fill=black,draw,inner sep=0pt,minimum size=4pt] (c) {};
        	\path (a) ++(45:1) node[circle,fill=black,draw,inner sep=0pt,minimum size=4pt] (d) {};
        	\path (a) ++(90:1) node[circle,fill=black,draw,inner sep=0pt,minimum size=4pt] (e) {};

            \path (c) ++(0:3) node[circle,fill=black,draw,inner sep=0pt,minimum size=4pt] (f) {};
            
        	\path (f) ++(-45:1) node[circle,fill=black,draw,inner sep=0pt,minimum size=4pt] (g) {};
        	\path (f) ++(0:1) node[circle,fill=black,draw,inner sep=0pt,minimum size=4pt] (h) {};
        	\path (f) ++(45:1) node[circle,fill=black,draw,inner sep=0pt,minimum size=4pt] (i) {};
        	\path (f) ++(90:1) node[circle,fill=black,draw,inner sep=0pt,minimum size=4pt] (j) {};
        	
        	\path (h) ++(0:3) node[circle,fill=black,draw,inner sep=0pt,minimum size=4pt] (k) {};

            \path (k) ++(-45:1) node[circle,fill=black,draw,inner sep=0pt,minimum size=4pt] (l) {};
        	\path (k) ++(0:1) node[circle,fill=black,draw,inner sep=0pt,minimum size=4pt] (m) {};
        	\path (k) ++(45:1) node[circle,fill=black,draw,inner sep=0pt,minimum size=4pt] (n) {};
        	\path (k) ++(90:1) node[circle,fill=black,draw,inner sep=0pt,minimum size=4pt] (o) {};
        	
        	\path (m) ++(0:3) node[circle,fill=red,draw,inner sep=0pt,minimum size=4pt] (p) {};
        	
        	\path (p) ++(-45:1) node[circle,fill=red,draw,inner sep=0pt,minimum size=4pt] (q) {};
        	\path (p) ++(0:1) node[circle,fill=red,draw,inner sep=0pt,minimum size=4pt] (r) {};
        	\path (p) ++(45:1) node[circle,fill=red,draw,inner sep=0pt,minimum size=4pt] (s) {};
        	\path (p) ++(90:1) node[circle,fill=red,draw,inner sep=0pt,minimum size=4pt] (t) {};
        	
        	\path (r) ++(0:6) node[circle,fill=red,draw,inner sep=0pt,minimum size=4pt] (u) {};

        	\draw[thick,dotted] (a) -- (b);
        	
        	\draw[thick, decorate,decoration=zigzag] (e) to[out=10,in=135] (f) (j) to[out=10,in=135] (k) (o) to[out=10,in=135] (p);
            \draw[thick, decorate,decoration=zigzag, red](t) to[out=10,in=160] node[font=\fontsize{8}{8},midway,inner sep=1pt,outer sep=1pt,minimum size=4pt,fill=white] {$P$} (u);
        	
        	\draw[thick] (a) -- (c) (a) -- (d) (a) -- (e) (f) -- (g) (f) -- (h) (f) -- (i) (f) -- (j) (k) -- (l) (k) -- (m) (k) -- (n) (k) -- (o);
            \draw[thick, red] (p) -- (q) (p) -- (r) (p) to node[font=\fontsize{8}{8},midway,inner sep=1pt,outer sep=1pt,minimum size=4pt,fill=white] {$F$} (s) (p) -- (t);
    	
    \end{tikzpicture}
    \caption{The chain $C$ and candidate chain $F+P$ at the start of an iteration.}
    \label{fig:iteration_start}
\end{figure}
We first check whether $\length(P) < 2\ell$, in which case $F$ is $(\psi, \alpha\beta)$-successful by \ref{inv:hopeful_length}).
If so, we have found a $\phi$-happy multi-step Vizing chain and we return $C+F+P$.
If not, we let $F_k = F$ and consider a random initial segment $P_k$ of $P$ of length between $\ell$ and $2\ell - 1$. 
Using Algorithm \ref{alg:next_chain}, we find a chain $\tilde F + \tilde P$ with $\Start(\tilde F) = \End(P_k)$ and $\vstart(\tilde{F}) = \vend(P_k)$.

At this point, we have two cases to consider. First, suppose $C + F_k + P_k + \tilde F + \tilde P$ is non-intersecting.
We then continue on, updating the chain to be $C+F_k+P_k$ and the candidate chain to be $\tilde F + \tilde P$. Fig.~\ref{fig:non_intersecting_iteration} shows an example of such an update with $\tilde F + \tilde P$ shown in blue.

\begin{figure}[H]
    \centering
        \begin{tikzpicture}[xscale = 0.7,yscale=0.7]
            \node[circle,fill=black,draw,inner sep=0pt,minimum size=4pt] (a) at (0,0) {};
        	\path (a) ++(-45:1) node[circle,fill=black,draw,inner sep=0pt,minimum size=4pt] (b) {};
        	\path (a) ++(0:1) node[circle,fill=black,draw,inner sep=0pt,minimum size=4pt] (c) {};
        	\path (a) ++(45:1) node[circle,fill=black,draw,inner sep=0pt,minimum size=4pt] (d) {};
        	\path (a) ++(90:1) node[circle,fill=black,draw,inner sep=0pt,minimum size=4pt] (e) {};

            \path (c) ++(0:3) node[circle,fill=black,draw,inner sep=0pt,minimum size=4pt] (f) {};
            
        	\path (f) ++(-45:1) node[circle,fill=black,draw,inner sep=0pt,minimum size=4pt] (g) {};
        	\path (f) ++(0:1) node[circle,fill=black,draw,inner sep=0pt,minimum size=4pt] (h) {};
        	\path (f) ++(45:1) node[circle,fill=black,draw,inner sep=0pt,minimum size=4pt] (i) {};
        	\path (f) ++(90:1) node[circle,fill=black,draw,inner sep=0pt,minimum size=4pt] (j) {};
        	
        	\path (h) ++(0:3) node[circle,fill=black,draw,inner sep=0pt,minimum size=4pt] (k) {};

            \path (k) ++(-45:1) node[circle,fill=black,draw,inner sep=0pt,minimum size=4pt] (l) {};
        	\path (k) ++(0:1) node[circle,fill=black,draw,inner sep=0pt,minimum size=4pt] (m) {};
        	\path (k) ++(45:1) node[circle,fill=black,draw,inner sep=0pt,minimum size=4pt] (n) {};
        	\path (k) ++(90:1) node[circle,fill=black,draw,inner sep=0pt,minimum size=4pt] (o) {};
        	
        	\path (m) ++(0:3) node[circle,fill=black,draw,inner sep=0pt,minimum size=4pt] (p) {};
        	
        	\path (p) ++(-45:1) node[circle,fill=black,draw,inner sep=0pt,minimum size=4pt] (q) {};
        	\path (p) ++(0:1) node[circle,fill=black,draw,inner sep=0pt,minimum size=4pt] (r) {};
        	\path (p) ++(45:1) node[circle,fill=black,draw,inner sep=0pt,minimum size=4pt] (s) {};
        	\path (p) ++(90:1) node[circle,fill=black,draw,inner sep=0pt,minimum size=4pt] (t) {};
        	
        	\path (r) ++(0:3) node[circle,fill=red,draw,inner sep=0pt,minimum size=4pt] (u) {};
        	
        	\path (u) ++(-45:1) node[circle,fill=red,draw,inner sep=0pt,minimum size=4pt] (v) {};
        	\path (u) ++(0:1) node[circle,fill=red,draw,inner sep=0pt,minimum size=4pt] (w) {};
        	\path (u) ++(45:1) node[circle,fill=red,draw,inner sep=0pt,minimum size=4pt] (x) {};
        	\path (u) ++(90:1) node[circle,fill=red,draw,inner sep=0pt,minimum size=4pt] (y) {};
        	
        	\path (w) ++(0:6) node[circle,fill=red,draw,inner sep=0pt,minimum size=4pt] (z) {};
        	
        	\draw[thick,dotted] (a) -- (b);
        	
        	\draw[thick, decorate,decoration=zigzag] (e) to[out=10,in=135] (f) (j) to[out=10,in=135] (k) (o) to[out=10,in=135] (p) (t) to[out=10,in=135] node[font=\fontsize{8}{8},midway,inner sep=1pt,outer sep=1pt,minimum size=4pt,fill=white] {$P_k$} (u);
            \draw[thick, decorate,decoration=zigzag, red](y) to[out=10,in=160] node[font=\fontsize{8}{8},midway,inner sep=1pt,outer sep=1pt,minimum size=4pt,fill=white] {$\tilde P$} (z);
        	
        	\draw[thick] (a) -- (c) (a) -- (d) (a) -- (e) (f) -- (g) (f) -- (h) (f) -- (i) (f) -- (j) (k) -- (l) (k) -- (m) (k) -- (n) (k) -- (o) (p) -- (q) (p) -- (r) (p) to node[font=\fontsize{8}{8},midway,inner sep=1pt,outer sep=1pt,minimum size=4pt,fill=white] {$F_k$} (s) (p) -- (t);
            \draw[thick, red] (u) -- (v) (u) -- (w) (u) to node[font=\fontsize{8}{8},midway,inner sep=1pt,outer sep=1pt,minimum size=4pt,fill=white] {$\tilde F$} (x) (u) -- (y);

        \begin{scope}[yshift=5.5cm]
            \node[circle,fill=black,draw,inner sep=0pt,minimum size=4pt] (a) at (0,0) {};
        	\path (a) ++(-45:1) node[circle,fill=black,draw,inner sep=0pt,minimum size=4pt] (b) {};
        	\path (a) ++(0:1) node[circle,fill=black,draw,inner sep=0pt,minimum size=4pt] (c) {};
        	\path (a) ++(45:1) node[circle,fill=black,draw,inner sep=0pt,minimum size=4pt] (d) {};
        	\path (a) ++(90:1) node[circle,fill=black,draw,inner sep=0pt,minimum size=4pt] (e) {};

            \path (c) ++(0:3) node[circle,fill=black,draw,inner sep=0pt,minimum size=4pt] (f) {};
            
        	\path (f) ++(-45:1) node[circle,fill=black,draw,inner sep=0pt,minimum size=4pt] (g) {};
        	\path (f) ++(0:1) node[circle,fill=black,draw,inner sep=0pt,minimum size=4pt] (h) {};
        	\path (f) ++(45:1) node[circle,fill=black,draw,inner sep=0pt,minimum size=4pt] (i) {};
        	\path (f) ++(90:1) node[circle,fill=black,draw,inner sep=0pt,minimum size=4pt] (j) {};
        	
        	\path (h) ++(0:3) node[circle,fill=black,draw,inner sep=0pt,minimum size=4pt] (k) {};

            \path (k) ++(-45:1) node[circle,fill=black,draw,inner sep=0pt,minimum size=4pt] (l) {};
        	\path (k) ++(0:1) node[circle,fill=black,draw,inner sep=0pt,minimum size=4pt] (m) {};
        	\path (k) ++(45:1) node[circle,fill=black,draw,inner sep=0pt,minimum size=4pt] (n) {};
        	\path (k) ++(90:1) node[circle,fill=black,draw,inner sep=0pt,minimum size=4pt] (o) {};
        	
        	\path (m) ++(0:3) node[circle,fill=red,draw,inner sep=0pt,minimum size=4pt] (p) {};
        	
        	\path (p) ++(-45:1) node[circle,fill=red,draw,inner sep=0pt,minimum size=4pt] (q) {};
        	\path (p) ++(0:1) node[circle,fill=red,draw,inner sep=0pt,minimum size=4pt] (r) {};
        	\path (p) ++(45:1) node[circle,fill=red,draw,inner sep=0pt,minimum size=4pt] (s) {};
        	\path (p) ++(90:1) node[circle,fill=red,draw,inner sep=0pt,minimum size=4pt] (t) {};
        	
        	\path (r) ++(0:3) node[circle,fill=blue,draw,inner sep=0pt,minimum size=4pt] (u) {};
        	
        	\path (u) ++(-45:1) node[circle,fill=blue,draw,inner sep=0pt,minimum size=4pt] (v) {};
        	\path (u) ++(0:1) node[circle,fill=blue,draw,inner sep=0pt,minimum size=4pt] (w) {};
        	\path (u) ++(45:1) node[circle,fill=blue,draw,inner sep=0pt,minimum size=4pt] (x) {};
        	\path (u) ++(90:1) node[circle,fill=blue,draw,inner sep=0pt,minimum size=4pt] (y) {};
        	
        	\path (w) ++(0:6) node[circle,fill=blue,draw,inner sep=0pt,minimum size=4pt] (z) {};
        	
        	\draw[thick,dotted] (a) -- (b);
        	
        	\draw[thick, decorate,decoration=zigzag] (e) to[out=10,in=135] (f) (j) to[out=10,in=135] (k) (o) to[out=10,in=135] (p);
            \draw[thick, decorate,decoration=zigzag, red](t) to[out=10,in=135] node[font=\fontsize{8}{8},midway,inner sep=1pt,outer sep=1pt,minimum size=4pt,fill=white] {$P_k$} (u);
            \draw[thick, decorate,decoration=zigzag, blue](y) to[out=10,in=160] node[font=\fontsize{8}{8},midway,inner sep=1pt,outer sep=1pt,minimum size=4pt,fill=white] {$\tilde P$} (z);
        	
        	\draw[thick] (a) -- (c) (a) -- (d) (a) -- (e) (f) -- (g) (f) -- (h) (f) -- (i) (f) -- (j) (k) -- (l) (k) -- (m) (k) -- (n) (k) -- (o);
            \draw[thick, red] (p) -- (q) (p) -- (r) (p) to node[font=\fontsize{8}{8},midway,inner sep=1pt,outer sep=1pt,minimum size=4pt,fill=white] {$F_k$} (s) (p) -- (t);
            \draw[thick, blue] (u) -- (v) (u) -- (w) (u) to node[font=\fontsize{8}{8},midway,inner sep=1pt,outer sep=1pt,minimum size=4pt,fill=white] {$\tilde F$} (x) (u) -- (y);
        \end{scope}

        \begin{scope}[yshift=3cm]
            \draw[-{Stealth[length=3mm,width=2mm]},very thick,decoration = {snake,pre length=3pt,post length=7pt,},decorate] (11.5,1) -- (11.5,-1);
        \end{scope}
        	
        \end{tikzpicture}
    \caption{Example of a non-intersecting update.}
    \label{fig:non_intersecting_iteration}
\end{figure}
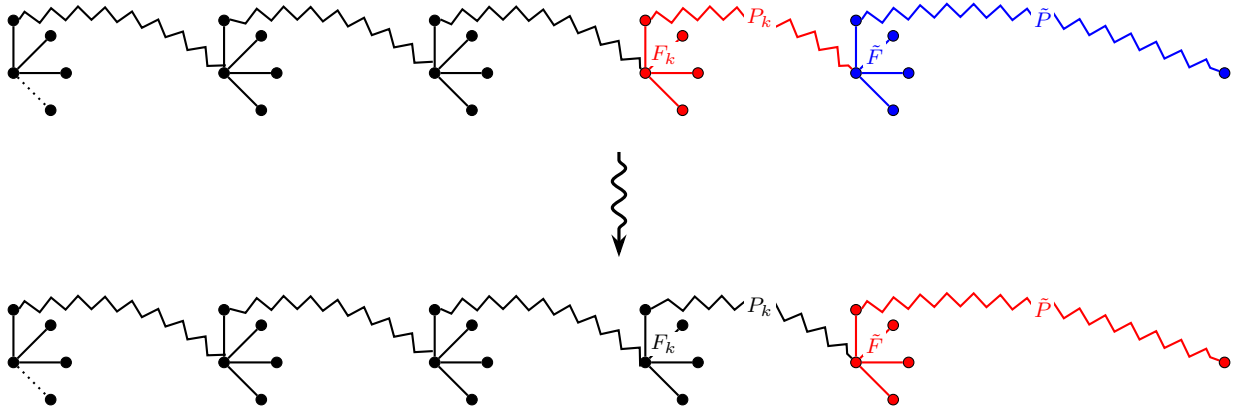

Now suppose $\tilde F + \tilde P$ intersects $C + F_k + P_k$. The edges and vertices of $\tilde{F} + \tilde{P}$ are naturally ordered, and we let $0 \leq j \leq k$ be the index such that the first intersection point of $\tilde{F} + \tilde{P}$ with $C + F_k + P_k$ occurred at $F_j+P_j$. Then we update $C$ to $C' \defeq F_0 + P_0 + \cdots + F_{j-1} + P_{j-1}$ and $F+P$ to $F_j + P'$, where $P'$ is the path of length $2\ell$ from which $P_j$ was obtained as an initial segment. Now, during the next iteration, we will ``resample'' $P_j$ by randomly truncating $P'$ again.
Fig.~\ref{fig:intersecting_iteration} shows an example of such an update.

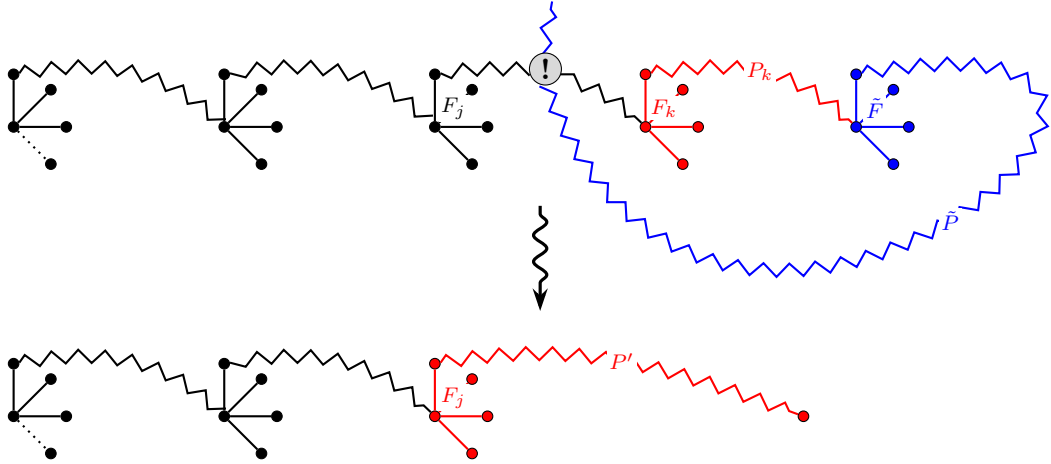
\begin{figure}[H]
    \centering
        \begin{tikzpicture}[xscale = 0.7, yscale=0.7]
            \clip (-0.5,-6.5) rectangle (20,2.5);
    
            \node[circle,fill=black,draw,inner sep=0pt,minimum size=4pt] (a) at (0,0) {};
        	\path (a) ++(-45:1) node[circle,fill=black,draw,inner sep=0pt,minimum size=4pt] (b) {};
        	\path (a) ++(0:1) node[circle,fill=black,draw,inner sep=0pt,minimum size=4pt] (c) {};
        	\path (a) ++(45:1) node[circle,fill=black,draw,inner sep=0pt,minimum size=4pt] (d) {};
        	\path (a) ++(90:1) node[circle,fill=black,draw,inner sep=0pt,minimum size=4pt] (e) {};

            \path (c) ++(0:3) node[circle,fill=black,draw,inner sep=0pt,minimum size=4pt] (f) {};
            
        	\path (f) ++(-45:1) node[circle,fill=black,draw,inner sep=0pt,minimum size=4pt] (g) {};
        	\path (f) ++(0:1) node[circle,fill=black,draw,inner sep=0pt,minimum size=4pt] (h) {};
        	\path (f) ++(45:1) node[circle,fill=black,draw,inner sep=0pt,minimum size=4pt] (i) {};
        	\path (f) ++(90:1) node[circle,fill=black,draw,inner sep=0pt,minimum size=4pt] (j) {};
        	
        	\path (h) ++(0:3) node[circle,fill=black,draw,inner sep=0pt,minimum size=4pt] (k) {};

            \path (k) ++(-45:1) node[circle,fill=black,draw,inner sep=0pt,minimum size=4pt] (l) {};
        	\path (k) ++(0:1) node[circle,fill=black,draw,inner sep=0pt,minimum size=4pt] (m) {};
        	\path (k) ++(45:1) node[circle,fill=black,draw,inner sep=0pt,minimum size=4pt] (n) {};
        	\path (k) ++(90:1) node[circle,fill=black,draw,inner sep=0pt,minimum size=4pt] (o) {};
        	
        	\path (m) ++(0:3) node[circle,fill=red,draw,inner sep=0pt,minimum size=4pt] (p) {};
        	
        	\path (p) ++(-45:1) node[circle,fill=red,draw,inner sep=0pt,minimum size=4pt] (q) {};
        	\path (p) ++(0:1) node[circle,fill=red,draw,inner sep=0pt,minimum size=4pt] (r) {};
        	\path (p) ++(45:1) node[circle,fill=red,draw,inner sep=0pt,minimum size=4pt] (s) {};
        	\path (p) ++(90:1) node[circle,fill=red,draw,inner sep=0pt,minimum size=4pt] (t) {};
        	
        	\path (r) ++(0:3) node[circle,fill=blue,draw,inner sep=0pt,minimum size=4pt] (u) {};
        	
        	\path (u) ++(-45:1) node[circle,fill=blue,draw,inner sep=0pt,minimum size=4pt] (v) {};sep=0pt,minimum size=4pt] (q) {};
        	\path (u) ++(0:1) node[circle,fill=blue,draw,inner sep=0pt,minimum size=4pt] (w) {};
        	\path (u) ++(45:1) node[circle,fill=blue,draw,inner sep=0pt,minimum size=4pt] (x) {};
        	\path (u) ++(90:1) node[circle,fill=blue,draw,inner sep=0pt,minimum size=4pt] (y) {};
        	
        	\node[circle,inner sep=1pt] (z) at (10, 1.1) {\textbf{!}};
        	\path (z) ++(80:1.5) node (aa) {};

        	\draw[thick,dotted] (a) -- (b);
        	
        	\draw[thick, decorate,decoration=zigzag] (e) to[out=10,in=135] (f) (j) to[out=10,in=135] (k) (o) to[out=10,in=180] (z) to[out=-20,in=135] (p);
            \draw[thick,decorate,decoration=zigzag, red] (t) to[out=10,in=135] node[font=\fontsize{8}{8},midway,inner sep=1pt,outer sep=1pt,minimum size=3pt,fill=white] {$P_k$} (u); 
            
            \draw[thick,decorate,decoration=zigzag, blue] (y) to[out=10,in=-70, looseness=4] node[font=\fontsize{8}{8},midway,inner sep=1pt,outer sep=1pt,minimum size=3pt,fill=white] {$\tilde P$} (z) -- (aa);
        	
        	\draw[thick] (a) -- (c) (a) -- (d) (a) -- (e) (f) -- (g) (f) -- (h) (f) -- (i) (f) -- (j) (k) -- (l) (k) -- (m) (k) to node[font=\fontsize{8}{8},midway,inner sep=1pt,outer sep=1pt,minimum size=4pt,fill=white] {$F_j$} (n) (k) -- (o);
            \draw[thick, red] (p) -- (q) (p) -- (r) (p) to node[font=\fontsize{8}{8},midway,inner sep=1pt,outer sep=1pt,minimum size=4pt,fill=white] {$F_k$} (s) (p) -- (t);
            \draw[thick, blue] (u) -- (v) (u) -- (w) (u) to node[font=\fontsize{8}{8},midway,inner sep=1pt,outer sep=1pt,minimum size=4pt,fill=white] {$\tilde F$} (x) (u) -- (y);

            \node[circle,fill=gray!30,draw,inner sep=1.3pt] at (10.1, 1.1) {\textbf{!}};
        
        \begin{scope}[yshift=-2.5cm]
            \draw[-{Stealth[length=3mm,width=2mm]},very thick,decoration = {snake,pre length=3pt,post length=7pt,},decorate] (10,1) -- (10,-1);
        \end{scope}
        
        \begin{scope}[yshift=-5.5cm]
            \node[circle,fill=black,draw,inner sep=0pt,minimum size=4pt] (a) at (0,0) {};
        	\path (a) ++(-45:1) node[circle,fill=black,draw,inner sep=0pt,minimum size=4pt] (b) {};
        	\path (a) ++(0:1) node[circle,fill=black,draw,inner sep=0pt,minimum size=4pt] (c) {};
        	\path (a) ++(45:1) node[circle,fill=black,draw,inner sep=0pt,minimum size=4pt] (d) {};
        	\path (a) ++(90:1) node[circle,fill=black,draw,inner sep=0pt,minimum size=4pt] (e) {};

            \path (c) ++(0:3) node[circle,fill=black,draw,inner sep=0pt,minimum size=4pt] (f) {};
            
        	\path (f) ++(-45:1) node[circle,fill=black,draw,inner sep=0pt,minimum size=4pt] (g) {};
        	\path (f) ++(0:1) node[circle,fill=black,draw,inner sep=0pt,minimum size=4pt] (h) {};
        	\path (f) ++(45:1) node[circle,fill=black,draw,inner sep=0pt,minimum size=4pt] (i) {};
        	\path (f) ++(90:1) node[circle,fill=black,draw,inner sep=0pt,minimum size=4pt] (j) {};
        	
        	\path (h) ++(0:3) node[circle,fill=red,draw,inner sep=0pt,minimum size=4pt] (k) {};

            \path (k) ++(-45:1) node[circle,fill=red,draw,inner sep=0pt,minimum size=4pt] (l) {};
        	\path (k) ++(0:1) node[circle,fill=red,draw,inner sep=0pt,minimum size=4pt] (m) {};
        	\path (k) ++(45:1) node[circle,fill=red,draw,inner sep=0pt,minimum size=4pt] (n) {};
        	\path (k) ++(90:1) node[circle,fill=red,draw,inner sep=0pt,minimum size=4pt] (o) {};
        	
        	\path (m) ++(0:6) node[circle,fill=red,draw,inner sep=0pt,minimum size=4pt] (p) {};

        \draw[thick, dotted] (a) -- (b);
         \draw[thick] (a) -- (c) (a) -- (d) (a) -- (e) (f) -- (g) (f) -- (h) (f) -- (i) (f) -- (j);
         \draw[thick, red] (k) -- (l) (k) -- (m) (k) to node[font=\fontsize{8}{8},midway,inner sep=1pt,outer sep=1pt,minimum size=4pt,fill=white] {$F_j$} (n) (k) -- (o);

         \draw[thick, decorate,decoration=zigzag] (e) to[out=10,in=135] (f) (j) to[out=10,in=135] (k);
         \draw[thick, decorate,decoration=zigzag, red] (o) to[out=10,in=160] node[font=\fontsize{8}{8},midway,inner sep=1pt,outer sep=1pt,minimum size=4pt,fill=white] {$P'$} (p);
        \end{scope}
        \end{tikzpicture}
    \caption{Example of an intersecting update.}
    \label{fig:intersecting_iteration}
\end{figure}

In order to track intersections we define a hash map $\visited$ with key set $V\cup E$. We let $\visited(v) = 1$ if and only if $v$ is a vertex on a fan in the current chain $C$ and $\visited(e) = 1$ if and only if $e$ is an internal edge of a path in the current chain $C$.
The formal statement of our \hyperref[alg:multi_viz_chain]{Multi-Step Vizing Algorithm} is given in Algorithm~\ref{alg:multi_viz_chain}.

\begin{algorithm}[h]\algsize
\caption{Multi-Step Vizing Algorithm (MSVA)}\label{alg:multi_viz_chain}
\begin{flushleft}
\textbf{Input}: A proper partial coloring $\phi$, an uncolored edge $xy$, a vertex $x \in e$, and a parameter $\ell \in \N$. \\
\textbf{Output}: A $\phi$-happy multi-step Vizing chain $C$ with $\Start(C) = xy$.
\end{flushleft}
\begin{algorithmic}[1]
    \State $\visited(e) \gets 0, \quad \visited(v) \gets 0$ \quad \textbf{for each} $e \in E$, $v \in V$
    \State $(F,P) \gets \hyperref[alg:first_chain]{\mathsf{FirstChain}}(\phi, xy, x, \ell)$ \label{step:first_chain} \Comment{Algorithm \ref{alg:first_chain}}
    \State $C\gets (xy), \quad \psi \gets \phi, \quad k \gets 0$
    \medskip
    \While{true}
        \If{$\length(P) < 2\ell$}
            \State \Return $C+F+P$ \label{step:success} \Comment{Success}
        \EndIf
        \State\label{step:random_choice} Let $\ell' \in [\ell,2\ell-1]$ be an integer chosen uniformly at random.
        \State $F_k \gets F,\quad P_k\gets P|\ell'$ \label{step:Pk} \Comment{Randomly shorten the path}
        \State Let $\alpha$, $\beta$ be such that $P_k$ is an $\alpha\beta$-path where $\psi(\End(P_k)) = \beta$.
        \State $\psi \gets \Shift(\psi, F_k+P_k)$ 
        \State $\visited(v) \gets 1$ \textbf{for each} $v \in V(F_k)$
        \State $\visited(e) \gets 1$ \textbf{for each} $e \in \IE(P_k)$
        \State $uv \gets \End(P_k), \quad v \gets \vend(P_k)$
        \State $(\tilde F , \tilde P) \gets \hyperref[alg:next_chain]{\mathsf{NextChain}}(\psi, uv, u, \ell, \alpha, \beta)$ \label{step:alpha_beta_order} \Comment{Algorithm \ref{alg:next_chain}}
        \If{$\visited(v) = 1$ or $\visited(e) = 1$ for some $v\in V(\tilde F + \tilde P)$, $e \in E(\tilde F + \tilde P)$}
            \State Let $0 \leq j \leq k$ be such that the first intersection occurs at $F_j + P_j$.\label{step:choosej}
            \State $\psi \gets \Shift(\psi, (F_j + P_j + \cdots + F_k + P_k)^*)$ \label{step:psi}
            \State $\visited(v) \gets 0$ \textbf{for each} $v \in V(F_j) \cup \ldots \cup V(F_k)$
            \State $\visited(e) \gets 0$ \textbf{for each} $e \in \IE(P_j) \cup \ldots \cup \IE(P_k)$\label{step:visited}
            \State $C\gets F_0 + P_0 + \cdots + F_{j-1} + P_{j-1}, \quad k \gets j$ \label{step:truncate1} \Comment{Return to step $j$}
            \State $F\gets F_j, \quad P \gets P'$ \quad where $P_j$ is an initial segment of $P'$ as described earlier. 
            \label{step:truncate}
        \ElsIf{$2 \leq \length(\tilde P) < 2\ell$ \textbf{and} $\vend(\tilde P) = \Pivot(\tilde F)$}
            \State\label{step:fail_chain} \Return \textbf{\textsf{FAIL}} \Comment{Failure}
        \Else
            \State $ C \gets C + F_k + P_k, 
        \quad F \gets \tilde F, \quad P \gets \tilde P, \quad k \gets k + 1$ \label{step:append} \Comment{Append and move on to the next step}
        \EndIf
    \EndWhile
\end{algorithmic}
\end{algorithm}

A few remarks are in order. 
Note that in steps \ref{step:psi}--\ref{step:visited} of the algorithm, we can update $\visited$ and the missing colors hash map $M(\cdot)$ while simultaneously updating $\psi$. By construction, $\length(P_k) \geq \ell > 2$ for all $k$. This ensures that at least one edge of each color $\alpha$, $\beta$ is on the $\alpha\beta$-path $P_k$ and also guarantees that $V(\End(F_j)) \cap V(\Start(F_{j+1})) = \0$ for all $j$.
In step \ref{step:truncate1} of the algorithm, we truncate the current chain at the {first} vertex $v$ or edge $e$ on $\tilde F + \tilde P$ such that $\visited(v) = 1$ or $\visited(e) = 1$.
These observations will be important for the proofs in the sequel.

\subsection{Proof of Correctness}\label{subsec:algo_lemmas}

In this subsection, we prove the correctness of Algorithm \ref{alg:multi_viz_chain} as well as some auxiliary results on the chain it outputs.
These results will be important for the analysis in later sections.
First, let us consider the output of Algorithm \ref{alg:first_chain}.

\begin{Lemma}\label{lemma:first_chain}
    Let $\phi$ be a proper partial coloring, let $\ell \in \N$, and let $xy$ be an uncolored edge. Let $F$ and $P$ be the fan and path returned by Algorithm \ref{alg:first_chain} on input $(\phi, xy, x, \ell)$, where $P$ is an $\alpha\beta$-path. Then
    \begin{itemize}
        \item either $F$ is $\phi$-happy and $P = (\End(F))$, or
        \item $F$ is $(\phi, \alpha\beta)$-hopeful and $\length(P) = 2\ell$, or
        \item $F$ is $(\phi, \alpha\beta)$-successful.
    \end{itemize}
\end{Lemma}

\begin{proof}
    Let $(\tilde F, \beta, j)$ be the output of Algorithm~\ref{alg:first_fan} at step~\ref{step:first_fan} of Algorithm~\ref{alg:first_chain}.
    If the algorithm reaches step \ref{step:happy_first}, then $F = \tilde F$ is $\phi$-happy.
    If not, let $\tilde F' \defeq \tilde F|j$.
    By Lemma \ref{lemma:first_fan_ber}, either $\tilde F$ or $\tilde F'$ is $(\phi, \alpha\beta)$-successful.
    In the case when $\tilde F$ is $(\phi, \alpha\beta)$-disappointed, step \ref{step:path_condition} ensures that $P$ is the initial segment of a path of length greater than $2\ell$. 
    If not, we pick the $(\phi, \alpha\beta)$-successful fan $\tilde F'$.
\end{proof}

Note that in the setting of Lemma \ref{lemma:first_chain}, the chain $F+P$ is $\phi$-shiftable.
Furthermore, if $\length(P) < 2\ell$, then $F+P$ is $\phi$-happy.
Let us now consider the output of Algorithm \ref{alg:next_chain}.

\begin{Lemma}\label{lemma:next_chain}
    Suppose that $\phi$ is a proper partial coloring, $xy$ is an uncolored edge, $\ell \in \N$, and $\alpha$, $\beta$ are colors such that $\alpha \in M(\phi, x) \setminus M(\phi, y)$ and $\beta \in M(\phi, y)$. Let $\tilde F$ and $\tilde P$ be the fan and the path returned by Algorithm \ref{alg:next_chain} on input $(\phi, xy, x, \ell, \alpha,\beta)$, where $\tilde P$ is a $\gamma\delta$-path. Then no edge in $\tilde F$ is colored $\alpha$ or $\beta$ and
    \begin{itemize}
        \item either $\tilde F$ is $\phi$-happy and $\tilde P = (\End(\tilde F))$, or
        \item $\tilde F$ is $(\phi, \gamma\delta)$-hopeful and $\{\gamma, \delta\} = \{\alpha, \beta\}$, or
        \item $\tilde F$ is $(\phi, \gamma\delta)$-hopeful, $\{\gamma, \delta\} \cap \{\alpha, \beta\} = \0$, and $\length(\tilde P) = 2\ell$, or
        \item $\tilde F$ is $(\phi, \gamma\delta)$-successful and $\{\gamma, \delta\} \cap \{\alpha, \beta\} = \0$.
    \end{itemize}
\end{Lemma}

\begin{proof}
    Let $(F, \delta, j)$ be the output of Algorithm~\ref{alg:next_fan} at step~\ref{step:next_fan} of Algorithm~\ref{alg:next_chain}.
    By Lemma \ref{lemma:next_fan_ber},
    no edge in $F$ is colored $\alpha$ or $\beta$ and at least one of the following holds:
    \begin{itemize}
        \item $F$ is $\phi$-happy, or
        \item $\delta = \beta$ and the fan $F$ is $(\phi, \alpha\beta)$-hopeful, or
        \item either $F$ or $F'$ is $(\phi, \gamma\delta)$-successful for some $\gamma \in M(\phi, x) \setminus \{\alpha\}$.
    \end{itemize}
    In the first two cases, we let $\tilde F = F$ and the claim follows.
    For $F' \defeq F|j$, by Lemma~\ref{lemma:next_fan_ber}, either $F$ or $F'$ is $(\phi, \gamma\delta)$-successful.
    In the case that $F$ is $(\phi, \gamma\delta)$-disappointed and $\set{\gamma, \delta} \cap \set{\alpha, \beta} = \0$, step \ref{step:next_path_condition} of Algorithm \ref{alg:next_chain} ensures that $\tilde P$ is the initial segment of a path of length greater than $2\ell$.
\end{proof}

Note that Algorithm \ref{alg:multi_viz_chain} outputs a $\phi$-happy chain as long as the algorithm terminates and
\begin{enumerate}[label=\ep{\normalfont{}\texttt{Happy}\arabic*},labelindent=15pt,leftmargin=*]
    \item\label{item:valid_input} the input to Algorithm \ref{alg:next_chain} at step \ref{step:alpha_beta_order} is valid,

    \item\label{item:never_fail} we never reach step \ref{step:fail_chain}, and

    \item\label{item:invariants} the invariants \ref{inv:start_F_end_C}--\ref{inv:hopeful_length} hold for each iteration of the \textsf{while} loop.
    
\end{enumerate}
Let us first show that item \ref{item:valid_input} holds.

\begin{Lemma}\label{lemma:valid_input}
    The input given to Algorithm \ref{alg:next_chain} at step \ref{step:alpha_beta_order} of Algorithm \ref{alg:multi_viz_chain} is valid.
\end{Lemma}

\begin{proof}
    In order for the input to be valid, we must ensure the following:
    \begin{enumerate}[label=\ep{\normalfont{}\texttt{Input}\arabic*},labelindent=15pt,leftmargin=*]
        \item\label{item:uncolored} $\psi(uv) = \blank$, 

        \item\label{item:beta} $\beta \in M(\psi, v)$, and

        \item\label{item:alpha} $\alpha \in M(\psi, u) \setminus M(\psi, v)$.
    \end{enumerate}
    As $uv = \End(P_k)$, it follows that item \ref{item:uncolored} holds. 
    To verify \ref{item:beta} and \ref{item:alpha}, consider the path chains $P$ and $P_k = P|\ell'$ (see Fig.~\ref{fig:next_chain_invariants} for an illustration).
    Let $\psi' \defeq \Shift(\psi, (F_k + P_k)^*)$ be the coloring at the start of this iteration of the \textsf{while} loop.
    Recall that the colors $\alpha$ and $\beta$ are defined so that $\beta = \psi'(\End(P_k))$ and $\alpha$ is the other color on the path $P$.
    As $\psi = \Shift(\psi', F_k + P_k)$, it follows that $\beta \in M(\psi, v)$ and $\alpha \in M(\psi, u)$ as desired. 
    Note that
    \[\length(P_k) \,=\, \ell' \,\leq\, 2\ell - 1 \,<\, 2\ell \,=\, \length(P).\]
    In particular, $v$ is not the endpoint of $P$ and so there is an edge colored $\alpha$ incident to $v$ under $\psi$. 
    Therefore, $\alpha \notin M(\psi, v)$, as desired.
\end{proof}

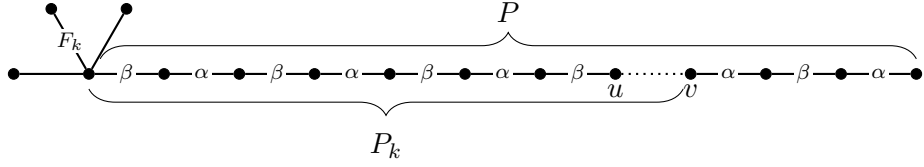
\begin{figure}[t]
    \centering
        \begin{tikzpicture}
            \node[circle,fill=black,draw,inner sep=0pt,minimum size=4pt] (a) at (0,0) {};
            \path (a) ++(0:1) node[circle,fill=black,draw,inner sep=0pt,minimum size=4pt] (b) {};
            \path (b) ++(120:1) node[circle,fill=black,draw,inner sep=0pt,minimum size=4pt] (c) {};
            \path (b) ++(60:1) node[circle,fill=black,draw,inner sep=0pt,minimum size=4pt] (d) {};
            \path (b) ++(0:1) node[circle,fill=black,draw,inner sep=0pt,minimum size=4pt] (e) {};

            \path (e) ++(0:1) node[circle,fill=black,draw,inner sep=0pt,minimum size=4pt] (f) {};
            \path (f) ++(0:1) node[circle,fill=black,draw,inner sep=0pt,minimum size=4pt] (g) {};
            \path (g) ++(0:1) node[circle,fill=black,draw,inner sep=0pt,minimum size=4pt] (h) {};
            \path (h) ++(0:1) node[circle,fill=black,draw,inner sep=0pt,minimum size=4pt] (i) {};
            \path (i) ++(0:1) node[circle,fill=black,draw,inner sep=0pt,minimum size=4pt] (j) {};
            \path (j) ++(0:1) node[circle,fill=black,draw,inner sep=0pt,minimum size=4pt] (k) {};
            \path (k) ++(0:1) node[circle,fill=black,draw,inner sep=0pt,minimum size=4pt] (l) {};
            \path (l) ++(0:1) node[circle,fill=black,draw,inner sep=0pt,minimum size=4pt] (m) {};
            \path (m) ++(0:1) node[circle,fill=black,draw,inner sep=0pt,minimum size=4pt] (n) {};
            \path (n) ++(0:1) node[circle,fill=black,draw,inner sep=0pt,minimum size=4pt] (o) {};

            \node[anchor=north] at (k) {$u$};
            \node[anchor=north] at (l) {$v$};

            \draw[thick] (a) -- (b) to node[font=\fontsize{8}{8},midway,inner sep=1pt,outer sep=1pt,minimum size=4pt,fill=white] {$F_k$} (c) (b) -- (d) (b) to node[font=\fontsize{8}{8},midway,inner sep=1pt,outer sep=1pt,minimum size=4pt,fill=white] {$\beta$} (e) to node[font=\fontsize{8}{8},midway,inner sep=1pt,outer sep=1pt,minimum size=4pt,fill=white] {$\alpha$} (f) to node[font=\fontsize{8}{8},midway,inner sep=1pt,outer sep=1pt,minimum size=4pt,fill=white] {$\beta$} (g) to node[font=\fontsize{8}{8},midway,inner sep=1pt,outer sep=1pt,minimum size=4pt,fill=white] {$\alpha$} (h) to node[font=\fontsize{8}{8},midway,inner sep=1pt,outer sep=1pt,minimum size=4pt,fill=white] {$\beta$} (i) to node[font=\fontsize{8}{8},midway,inner sep=1pt,outer sep=1pt,minimum size=4pt,fill=white] {$\alpha$} (j) to node[font=\fontsize{8}{8},midway,inner sep=1pt,outer sep=1pt,minimum size=4pt,fill=white] {$\beta$} (k) (l) to node[font=\fontsize{8}{8},midway,inner sep=1pt,outer sep=1pt,minimum size=4pt,fill=white] {$\alpha$} (m) to node[font=\fontsize{8}{8},midway,inner sep=1pt,outer sep=1pt,minimum size=4pt,fill=white] {$\beta$} (n) to node[font=\fontsize{8}{8},midway,inner sep=1pt,outer sep=1pt,minimum size=4pt,fill=white] {$\alpha$} (o);

            \draw[thick, dotted] (k) -- (l);

            \draw[decoration={brace,amplitude=10pt,mirror},decorate] (1, -0.2) -- node [midway,above,xshift=0pt,yshift=-30pt] {$P_k$} (8.9,-0.2);

            \draw[decoration={brace,amplitude=10pt},decorate] (1.15, 0.2) -- node [midway,above,xshift=0pt,yshift=10pt] {$P$} (12,0.2);
        \end{tikzpicture}    
    \caption{$F_k$, $P_k$, and $P$ under the coloring $\psi$ at step \ref{step:alpha_beta_order} of Algorithm \ref{alg:multi_viz_chain}.}
    \label{fig:next_chain_invariants}
\end{figure}

Next, we verify item \ref{item:never_fail}, that is, we show that Algorithm \ref{alg:multi_viz_chain} never fails.

\begin{Lemma}\label{lemma:never_fail}
    We never reach step \ref{step:fail_chain} in Algorithm \ref{alg:multi_viz_chain}.
\end{Lemma}

\begin{proof}
    Suppose we reach step \ref{step:fail_chain} on some iteration of the \textsf{while} loop in Algorithm~\ref{alg:multi_viz_chain}.
    Let $C$, $F+P$ be the chain and the candidate chain at the start of this iteration, where $P$ is an $\alpha\beta$-path. 
    Let $P_k$ be the random truncation of $P$ computed on step \ref{step:Pk} and let $\tilde F$, $\tilde P$ be the output of Algorithm \ref{alg:next_chain} computed on step \ref{step:alpha_beta_order}.
    As we have reached step \ref{step:fail_chain}, $2 \leq \length(\tilde P) < 2\ell$ and $\vend(\tilde P) = \Pivot(\tilde F)$.
    For $\psi \defeq \Shift(\phi, C+F+P_k)$, by Lemma \ref{lemma:next_chain}, it must be that no edge in $\tilde F$ is colored $\alpha$ or $\beta$ and
    \begin{itemize}
        \item either $\tilde F$ is $\psi$-happy and $\tilde P = (\End(\tilde F))$, or
        \item $\tilde F$ is $(\psi, \gamma\delta)$-hopeful and $\{\gamma, \delta\} = \{\alpha, \beta\}$, or
        \item $\tilde F$ is $(\psi, \gamma\delta)$-hopeful, $\{\gamma, \delta\} \cap \{\alpha, \beta\} = \0$, and $\length(\tilde P) = 2\ell$, or
        \item $\tilde F$ is $(\psi, \gamma\delta)$-successful and $\{\gamma, \delta\} \cap \{\alpha, \beta\} = \0$.
    \end{itemize}
    As $\length(\tilde P) \geq 2$, $\tilde F$ cannot be $\psi$-happy.
    Additionally, since $\vend(\tilde P) = \Pivot(\tilde F)$, it must be the case that $\tilde F$ is $(\psi, \gamma\delta)$-disappointed.
    Finally, since $\length(\tilde P) < 2\ell$, we must have  $\set{\gamma, \delta} = \set{\alpha, \beta}$.
    We now claim that $\End(\tilde P) \in \IE(P_k)$, implying we would have reached step \ref{step:truncate1} earlier.

    Let $\psi'\defeq \Shift(\psi, \tilde F)$. 
    By construction, $\deg(\Pivot(\tilde F); \psi, \alpha\beta) = 1$ and so $\deg(\Pivot(\tilde F); \psi', \alpha\beta) = 1$ as well, since $M(\psi', \Pivot(\tilde F)) = M(\psi, \Pivot(\tilde F))$. 
    Furthermore, as $\alpha \in M(\psi, \Pivot(\tilde F))$ and since no edge in $\tilde F$ is colored $\alpha$ or $\beta$, the edge incident to $\Pivot(\tilde F)$ colored $\beta$ is the same under both $\psi$ and $\psi'$.
    This edge $e$ must lie on $P_k$ and must also be $\End(\tilde P)$.
    Finally, as $\End(P_k) = \Start(\tilde F)$, $e$ is the penultimate edge on $P_k$.
    Since $\length(P_k) \geq \ell - 1 > 2$, it follows that $e \neq \Start(P_k)$, implying $\End(\tilde P) \in \IE(P_k)$.
\end{proof}

Now we are ready to verify item \ref{item:invariants}, i.e., prove that the multi-step Vizing chain $C$ and candidate chain $F+P$ at the start of each iteration satisfy \ref{inv:start_F_end_C}--\ref{inv:hopeful_length}.

\begin{Lemma}\label{lemma:correctness_of_algo}
    Consider running Algorithm \ref{alg:multi_viz_chain} on input $(\phi, xy, x, \ell)$.
    Let $C$ and $F + P$ be the multi-step chain and the candidate chain at the beginning of some iteration of the {\upshape\textsf{while}} loop.
    Then $C$ and $F+P$ satisfy the invariants \ref{inv:start_F_end_C}--\ref{inv:hopeful_length}.
\end{Lemma}

\begin{proof}
    We induct on the iterations. At the first iteration, we have $C = (xy)$ and $(F, P)$ is the output of Algorithm \ref{alg:first_chain},
    so the invariants are satisfied by Lemma \ref{lemma:first_chain}.
    
    Suppose the invariants are satisfied for all iterations up to $i$. Let $C$ and $F+P$ be the multi-step chain and the candidate chain at the start of iteration $i$ such that $P$ is an $\alpha\beta$-path.
    We must show that the resulting multi-step chain $C'$ and the candidate chain $F'+P'$ at the end of the iteration satisfy the invariants as well.
    Let us consider each possible outcome of Algorithm \ref{alg:multi_viz_chain}.

    If we are in the situation of step \ref{step:success}, then $C+F+P$ is $\phi$-happy and the algorithm terminates. 
    Thus, we may assume that $\length(P) = 2\ell$ and step \ref{step:random_choice} yields a path chain $P_k \defeq P|\ell'$ of length $\ell'$.

    Let $(\tilde F, \tilde P)$ be the output of Algorithm \ref{alg:next_chain} on step \ref{step:alpha_beta_order} where $\tilde P$ is a $\gamma\delta$-path.
    By Lemma \ref{lemma:never_fail}, we need only consider the cases where we reach step \ref{step:truncate1} or \ref{step:append}.
    Note that if we reach step \ref{step:truncate1}, then $C'$ is an initial segment of $C$ formed at some iteration $j < i$ and therefore $C'$ and $F'+P'$ satisfy the invariants by the induction hypothesis. 
    If we reach step \ref{step:append}, then $C' = C + F + P_k$, $F' = \tilde F$, and $P' = \tilde P$.
    By our choice of input to Algorithm \ref{alg:next_chain} and since we assume we do not reach step~\ref{step:truncate1}, \ref{inv:start_F_end_C} and \ref{inv:non_intersecting_shiftable} are satisfied. 
    As a result of Lemma \ref{lemma:next_chain}, the only case to verify \ref{inv:hopeful_length} is if $\set{\gamma, \delta} = \set{\alpha, \beta}$ and $\tilde F$ is $(\Shift(\phi, C + F + P_k), \alpha\beta)$-disappointed.
    If $\length(\tilde P) < 2\ell$, we would reach step \ref{step:fail_chain}, which is not possible due to Lemma \ref{lemma:never_fail}, and so $\length(\tilde P) = 2\ell$.
    This covers all the cases and completes the proof.
\end{proof}

We conclude this subsection with two lemmas.
The first describes some properties of non-intersecting chains that will be useful in the proofs presented in the next section.

\begin{Lemma}\label{lemma:non-intersecting_degrees}
    Let $\phi$ be a proper partial coloring and let $xy$ be an uncolored edge. 
    Consider running Algorithm \ref{alg:multi_viz_chain} with input $(\phi, xy, x, \ell)$.
    Let $C = F_0+P_0+\cdots+F_{k-1}+P_{k-1}$ be the multi-step Vizing chain at the beginning of an iteration of the {\upshape\textsf{while}} loop and let $F_k + P_k$ be the chain formed at step \ref{step:Pk} such that, for each $j$, $P_j$ is an $\alpha_j\beta_j$-path in the coloring $\Shift(\phi, F_0+P_0+\cdots +F_{j-1}+P_{j-1})$. Then:
    \begin{enumerate}[label=\ep{\normalfont{}\texttt{Chain}\arabic*},labelindent=15pt,leftmargin=*]
        \item\label{item:degree_end} $\deg(\vend(F_j); \phi, \alpha_j\beta_j) = 1$  for each $0 \leq j \leq k$, and
        
        \item\label{item:related_phi} for each $0 \leq j \leq k$, all edges of $P_j$ except $\Start(P_j)$ are colored $\alpha_j$ or $\beta_j$ under $\phi$.
    \end{enumerate}
\end{Lemma}

\begin{proof}

    Let $\psi_0 \defeq \phi$ and $\psi_j \defeq \Shift(\phi, F_0+P_0+\cdots +F_{j-1}+P_{j-1})$ for $1 \leq j \leq k$. 
    By construction, items \ref{item:degree_end}, \ref{item:related_phi} hold with $\phi$ replaced by $\psi_{j}$.
    Note that by Lemma \ref{lemma:correctness_of_algo}, condition \ref{inv:non_intersecting_shiftable} is satisfied, and hence the chain $C' \defeq C + F_k + P_k$ is non-intersecting and $\phi$-shiftable.

    As $\psi_0 = \phi$, \ref{item:degree_end} holds for $\vend(F_0)$.
    Consider $1 \leq j \leq k$. 
    We will show that for $0 \leq i < j$, $M(\psi_i, \vend(F_j)) = M(\psi_{i+1}, \vend(F_j))$, which implies \ref{item:degree_end} since $\deg(\vend(F_j); \psi_j, \alpha_j\beta_j) = 1$.
    As $C'$ is non-intersecting, $\vend(F_j) \notin V(F_i)$, so if $\vend(F_j) \notin V(P_i)$, then clearly $M(\psi_i, \vend(F_j)) = M(\psi_{i+1}, \vend(F_j))$.
    If $\vend(F_j) \in V(P_i)$,
    then in fact $\vend(F_j) \in \IV(P_i) \setminus V(F_i)$, for if not, then it either belongs to $\End(F_i)$ or $\Start(F_{i+1})$, violating the non-intersecting property 
    (for the case that $j = i+1$, we have $\vend(F_j) \notin \Start(F_j)$ by construction).
    If $\vend(F_j) \in \IV(P_i) \setminus V(F_i)$, the only changes in its neighborhood between $\psi_i$ and $\psi_{i+1}$ are that the edges colored $\alpha_i$ and $\beta_i$ swap colors.
    In particular, the missing set of $\vend(F_j)$ is not altered, as desired. 
    
    Now suppose that $0 \leq j \leq k$ is such that some edge of $P_{j}$ (apart from $\Start(P_j)$) is not colored $\alpha_j$ or $\beta_j$ under $\phi$. 
    This implies that $E(P_j) \cap E(F_0+P_0 + \cdots +F_{j-1} + P_{j-1}) \neq \0$.
    By the non-intersecting property, it must be the case that $E(P_j) \cap E(F_i) \neq \0$ for some $i < j$. But then we would have $V(P_j) \cap V(F_i) \neq \0$
    violating the non-intersecting property. 
    So item \ref{item:related_phi} holds as well.
\end{proof}

The next lemma describes an implication of Lemma \ref{lemma:next_chain} on a certain kind of intersection.

    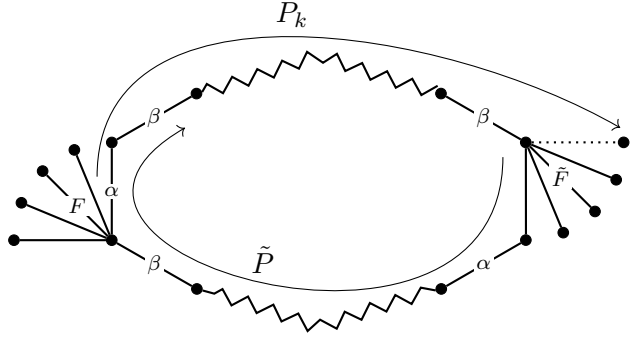
\begin{figure}[t]
        \centering
    	\begin{tikzpicture}[scale=1.3]
            \clip (-0.2,-1) rectangle (6.5,2.5);
     
            \node[circle,fill=black,draw,inner sep=0pt,minimum size=4pt] (a) at (0,0) {};
            \node[circle,fill=black,draw,inner sep=0pt,minimum size=4pt] (b) at (1,0) {};
            \path (b) ++(157.5:1) node[circle,fill=black,draw,inner sep=0pt,minimum size=4pt] (c) {};
            \path (b) ++(135:1) node[circle,fill=black,draw,inner sep=0pt,minimum size=4pt] (d) {};
            \path (b) ++(112.5:1) node[circle,fill=black,draw,inner sep=0pt,minimum size=4pt] (e) {};
            \path (b) ++(90:1) node[circle,fill=black,draw,inner sep=0pt,minimum size=4pt] (f) {};
            \path (f) ++(30:1) node[circle,fill=black,draw,inner sep=0pt,minimum size=4pt] (g) {};
            \path (g) ++(0:2.5) node[circle,fill=black,draw,inner sep=0pt,minimum size=4pt] (h) {};
            \path (h) ++(-30:1) node[circle,fill=black,draw,inner sep=0pt,minimum size=4pt] (i) {};
            \path (i) ++(0:1) node[circle,fill=black,draw,inner sep=0pt,minimum size=4pt] (j) {};
            \path (i) ++(-22.5:1) node[circle,fill=black,draw,inner sep=0pt,minimum size=4pt] (k) {};
            \path (i) ++(-45:1) node[circle,fill=black,draw,inner sep=0pt,minimum size=4pt] (l) {};
            \path (i) ++(-67.5:1) node[circle,fill=black,draw,inner sep=0pt,minimum size=4pt] (m) {};
            \path (i) ++(-90:1) node[circle,fill=black,draw,inner sep=0pt,minimum size=4pt] (n) {};
            \path (n) ++(-150:1) node[circle,fill=black,draw,inner sep=0pt,minimum size=4pt] (o) {};
            \path (o) ++(180:2.5) node[circle,fill=black,draw,inner sep=0pt,minimum size=4pt] (p) {};
            
            \draw[->, decorate] (0.85, 0.65) to[out=90, in=150, looseness=1] node [midway,above,xshift=0pt,yshift=0pt] {$P_k$} (6.2,1.15);
    
            \draw[->, decorate] (5, 0.85) to[out=-90, in=-150, looseness=2] node [midway,above,xshift=0pt,yshift=0pt] {$\tilde P$} (1.75,1.15);
            
            \draw[thick,dotted] (i) -- (j);
            \draw[thick] (a) -- (b) -- (c) (b) to node[font=\fontsize{8}{8},midway,inner sep=1pt,outer sep=1pt,minimum size=4pt,fill=white] {$F$} (d) (e) -- (b) to node[font=\fontsize{8}{8},midway,inner sep=1pt,outer sep=1pt,minimum size=4pt,fill=white] {$\alpha$} (f) to node[font=\fontsize{8}{8},midway,inner sep=1pt,outer sep=1pt,minimum size=4pt,fill=white] {$\beta$} (g) (h) to node[font=\fontsize{8}{8},midway,inner sep=1pt,outer sep=1pt,minimum size=4pt,fill=white] {$\beta$} (i) -- (k) (l) to node[font=\fontsize{8}{8},midway,inner sep=1pt,outer sep=1pt,minimum size=4pt,fill=white] {$\tilde F$} (i) -- (m) (i) -- (n) to node[font=\fontsize{8}{8},midway,inner sep=1pt,outer sep=1pt,minimum size=4pt,fill=white] {$\alpha$} (o) (p) to node[font=\fontsize{8}{8},midway,inner sep=1pt,outer sep=1pt,minimum size=4pt,fill=white] {$\beta$} (b);
            
            \draw[thick, decorate,decoration=zigzag] (g) to[out=15,in=165,looseness=2] (h) (o) to[out=195,in=-15,looseness=2] (p);
            
    	\end{tikzpicture}
        \caption{The situation in Lemma~\ref{lemma:intersection_prev}: $F+P_k+\tilde F+\tilde P$ under $\psi$.}
        \label{fig:intersection_prev}
    \end{figure}

\begin{Lemma}\label{lemma:intersection_prev}
    Suppose we reach step \ref{step:truncate1} during an iteration of the {\upshape\textsf{while}} loop in Algorithm \ref{alg:multi_viz_chain}.
    If $j = k$, then the first intersection occurred at a vertex in $V(F_k)$.
\end{Lemma}

\begin{proof}
    Consider such an iteration of the \textsf{while} loop.
    Let $C = F_0 + P_0 + \cdots + F_{k-1} + P_{k-1}$ and $F+P$ be the chain and the candidate chain at the start of the iteration such that $P$ is an $\alpha\beta$-path. 
    Let $P_k$ be the random truncation of $P$ computed on step \ref{step:Pk}
    and let $\tilde F$, $\tilde P$ be the output of Algorithm \ref{alg:next_chain} on step \ref{step:alpha_beta_order}.
    Since we reach step \ref{step:truncate1} with $j = k$, we must have
    \begin{itemize}
        \item either $V(\tilde F + \tilde P) \cap V(F) \neq \0$, or
        \item $E(\tilde F + \tilde P) \cap \IE(P_k) \neq \0$.
    \end{itemize}
    Suppose the first intersection is at an edge.
    Let $\psi \defeq \Shift(\phi, C + F + P_k)$, and,
    without loss of generality, let $\alpha \in M(\psi, \Pivot(\tilde F))$.
    By Lemma \ref{lemma:next_chain}, no edge in $\tilde F$ is colored $\alpha$ or $\beta$ under $\psi$ and
    \begin{itemize}
        \item either $\tilde F$ is $\psi$-happy and $\tilde P = (\End(\tilde F))$, or
        \item $\tilde F$ is $(\psi, \gamma\delta)$-hopeful and $\{\gamma, \delta\} = \{\alpha, \beta\}$, or
        \item $\tilde F$ is $(\psi, \gamma\delta)$-hopeful, $\{\gamma, \delta\} \cap \{\alpha, \beta\} = \0$, and $\length(\tilde P) = 2\ell$, or
        \item $\tilde F$ is $(\psi, \gamma\delta)$-successful and $\{\gamma, \delta\} \cap \{\alpha, \beta\} = \0$.
    \end{itemize}
    As no edge in $\tilde F$ is colored $\alpha$ or $\beta$ under $\psi$, $E(\tilde F) \cap E(P_k) = \set{\Start(\tilde F)}$, which is not a violation since $\Start(\tilde F)$ is not an internal edge of $P_k$.
    So we must have $E(\tilde P) \cap \IE(P_k) \neq \0$.
    Let $\psi'\defeq \Shift(\psi, \tilde F)$.
    As $E(\tilde F) \cap \IE(P_k) = \0$, $\psi'(e) = \psi(e) \in \set{\alpha, \beta}$ for all $e \in \IE(P_k)$.
    Therefore, we can conclude that $\set{\gamma, \delta} \cap \set{\alpha, \beta} \neq \0$ and hence $\set{\gamma, \delta} = \set{\alpha, \beta}$.
    See Fig.~\ref{fig:intersection_prev} for an illustration.

    Let $wz \in E(\tilde P) \cap \IE(P_k)$ be the first intersection.
    Note that as $wz \in \IE(P_k)$ and $\length(P_k) \geq \ell > 3$, either $w \in \IV(P_k)$ or $z \in \IV(P_k)$. Without loss of generality, let $w \in \IV(P_k)$. We claim that then $z = \vend(F)$, and hence the first forbidden intersection of $\tilde F + \tilde P$ with $F + P_k$ occurs at the vertex $z \in V(F)$, as desired.
    Suppose not, that is, either $z \in \IV(P_k)$ or $z = \Pivot(\tilde F)$. 
    As $z$ and $\vend(\tilde F)$ both lie on $\tilde P$, they must be $(\psi', \alpha\beta)$-related.
    Similarly, as $z$ and $\Pivot(\tilde F)$ both lie on $P_k$, they must also be $(\psi', \alpha\beta)$-related.
    It follows that $\Pivot(\tilde F)$ and $\vend(\tilde F)$ are $(\psi', \alpha\beta)$-related.
    Furthermore, $\deg(\Pivot(\tilde F); \psi', \alpha\beta) = \deg(\vend(\tilde F); \psi', \alpha\beta) = 1$ by construction.
    We conclude that $\tilde P$ is an initial segment of $P'\defeq P(\End(\tilde F); \psi', \alpha\beta)$, where $\vend(P') = \Pivot(\tilde F)$.
    Note that the edge $f$ that immediately precedes $wz$ on $P'$ must also lie on $\tilde P$.
    But since $w \in \IV(P_k)$ and $z \neq \vend(F)$, the edge $f$ must also belong to $\IE(P_k)$, which contradicts the choice of $wz$.
\end{proof}

\section{Analysis of Algorithm \ref{alg:multi_viz_chain}}\label{sec:msva_analysis}

For this section, we fix a proper partial coloring $\phi$.
The main result of this section is the following:

\begin{theo}\label{theo:entropy_compression_dist}
    Let $e = xy$ be an uncolored edge. For any $t > 0$ and $\ell \geq 1200\Delta^{16}$, Algorithm \ref{alg:multi_viz_chain} with input $(\phi, xy, x, \ell)$ computes an $e$-augmenting multi-step Vizing chain $C$ of length $\length(C) = O(\ell\,t)$ in time $O(\ell\,t)$ with probability at least $1 - 4m(1200\Delta^{15}/\ell)^{t/2}$.
\end{theo}

This result yields a constructive proof of Theorem \ref{theo:christlog}.
Indeed, taking 
$t = \Theta(\log n)$ and $\ell = 1200\Delta^{16}$, Theorem~\ref{theo:entropy_compression_dist} asserts that, starting with any uncolored edge $e$, Algorithm \ref{alg:multi_viz_chain} outputs an $e$-augmenting subgraph $C$ with at most $\length(C) = \poly(\Delta)  \log n$ edges in time $\poly(\Delta) \log n$ with probability at least $1 - 1/\poly(n)$.

As mentioned in \S\ref{sec:overview}, to prove Theorem \ref{theo:entropy_compression_dist}, we employ an entropy compression argument. 
Consider running Algorithm \ref{alg:multi_viz_chain} on input $(\phi, f,z, \ell)$, where $f \in E$ and $z \in f$. We emphasize that throughout this section, $\phi$ is assumed to be fixed. Hence, the execution process during the first $t$ iterations of the \textsf{while} loop in Algorithm \ref{alg:multi_viz_chain} is uniquely determined by the \emphd{input sequence} $(f, z, \ell_1, \ldots, \ell_t)$, where $\ell_i \in [\ell, 2\ell-1]$ is the random choice made at step \ref{step:random_choice} during the $i$-th iteration (note that at each iteration, we make precisely one such random choice).
We let $\mathcal{I}^{(t)}$ be the set of all input sequences for which Algorithm \ref{alg:multi_viz_chain} does not terminate within the first $t$ iterations.

\begin{defn}[Records and termini]\label{def:record}
    Let $I = (f, z, \ell_1, \ldots, \ell_t) \in \mathcal{I}^{(t)}$. Consider running Algorithm~\ref{alg:multi_viz_chain} for the first $t$ iterations with this input sequence. We define the \emphd{record} of $I$ to be the tuple $D(I) = (d_1, \ldots, d_t) \in \Z^t$, where for each $i$, $d_i$ is computed at the $i$-th iteration as follows:
\[
    d_i \,\defeq\, \begin{cases}
        1 &\text{if we reach step \ref{step:append}};\\
        j-k &\text{if we reach step \ref{step:choosej}}.
    \end{cases}
\]
Note that in the second case, $d_i \leq 0$. Let $C$ be the multi-step Vizing chain produced after the $t$-th iteration. The \emphd{terminus} of $I$ is the pair $\tau(I) = (\End(C), \vend(C))$
(in the case that $C = (f)$, we define $\vend(C)$ so that $\vend(C) \neq z$).
    \end{defn}
    Let $\mathcal{D}^{(t)}$ denote the set of all tuples $D \in \Z^t$ such that $D = D(I)$ for some $I \in \mathcal{I}^{(t)}$. Given $D \in \mathcal{D}^{(t)}$ and a pair $(uv, u)$ such that $uv \in E$, we let $\mathcal{I}^{(t)}(D, uv, u)$ be the set of all input sequences $I \in \mathcal{I}^{(t)}$ such that $D(I) = D$ and $\tau(I) = (uv, u)$.
The following functions will assist with our proofs:

\begin{align*}
    \val(z) \,&\defeq\, \left\{\begin{array}{cc}
        \Delta^4 & \text{if }z = 1 \\
        20\Delta^9 & \text{if } z = 0 \\
        75\ell\Delta^{11} & \text{if } z < 0
    \end{array}\right., &z \in \Z,\ z\leq 1, \\
    \wt(D) \,&\defeq\, \prod_{i = 1}^t\val(d_i), &D = (d_1, \ldots, d_t) \in \mathcal{D}^{(t)}.
\end{align*}

We are now ready to state the key result underlying our analysis of Algorithm~\ref{alg:multi_viz_chain}.

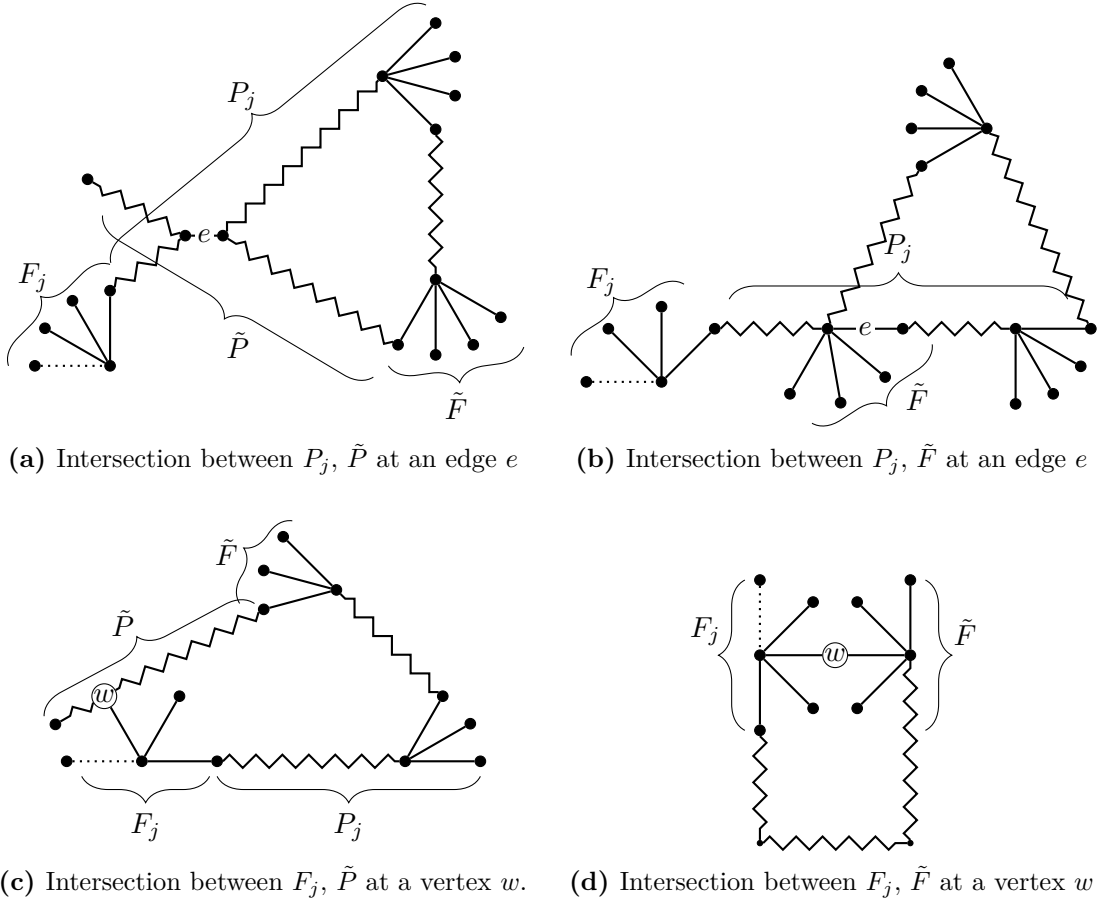
\begin{figure}[t]
    \begin{subfigure}[t]{0.45\textwidth}
        \centering
    	\begin{tikzpicture}
    	    \node[circle,fill=black,draw,inner sep=0pt,minimum size=4pt] (a) at (0,0) {};
    		\node[circle,fill=black,draw,inner sep=0pt,minimum size=4pt] (b) at (1,0) {};
    		\path (b) ++(150:1) node[circle,fill=black,draw,inner sep=0pt,minimum size=4pt] (c) {};
    		\path (b) ++(120:1) node[circle,fill=black,draw,inner sep=0pt,minimum size=4pt] (d) {};
    		\path (b) ++(90:1) node[circle,fill=black,draw,inner sep=0pt,minimum size=4pt] (e) {};
    		
    		\draw[decoration={brace,amplitude=10pt},decorate] (-0.35,0) -- node [midway,above,yshift=5pt,xshift=-10pt] {$F_j$} (1,1.35);
    		
    		\path (b) ++(60:2) node[circle,fill=black,draw,inner sep=0pt,minimum size=4pt] (f) {};
    		\path (f) ++(0:0.5) node[circle,fill=black,draw,inner sep=0pt,minimum size=4pt] (g) {};
    		
    		\draw[decoration={brace,amplitude=10pt},decorate] (1.01,1.37) -- node [midway,above,yshift=5pt,xshift=-10pt] {$P_j$} (5.2,4.8);
    		
    		\path (g) ++(45:3) node[circle,fill=black,draw,inner sep=0pt,minimum size=4pt] (h) {};
    		\path (h) ++(15:1) node[circle,fill=black,draw,inner sep=0pt,minimum size=4pt] (i) {};
    		\path (h) ++(45:1) node[circle,fill=black,draw,inner sep=0pt,minimum size=4pt] (j) {};
    		\path (h) ++(-15:1) node[circle,fill=black,draw,inner sep=0pt,minimum size=4pt] (k) {};
    		\path (h) ++(-45:1) node[circle,fill=black,draw,inner sep=0pt,minimum size=4pt] (l) {};
    		
    		\draw[decoration={brace,amplitude=10pt,mirror},decorate] (4.7,0) -- node [midway,below,yshift=-10pt,xshift=0pt] {$\tilde F$} (6.5,0.25);
    		
    		\path (l) ++(270:2) node[circle,fill=black,draw,inner sep=0pt,minimum size=4pt] (m) {};
    		\path (m) ++(270:1) node[circle,fill=black,draw,inner sep=0pt,minimum size=4pt] (n) {};
    		\path (m) ++(-60:1) node[circle,fill=black,draw,inner sep=0pt,minimum size=4pt] (o) {};
    		\path (m) ++(-30:1) node[circle,fill=black,draw,inner sep=0pt,minimum size=4pt] (p) {};
    		\path (m) ++(-120:1) node[circle,fill=black,draw,inner sep=0pt,minimum size=4pt] (q) {};
    		
    		\draw[decoration={brace,amplitude=10pt,mirror}, decorate] (0.9, 2) -- node [midway,below,xshift=0pt,yshift=-10pt] {$\tilde P$} (4.5,-0.15);
    		
    		\path (f) ++(150:1.5) node[circle,fill=black,draw,inner sep=0pt,minimum size=4pt] (r) {};
    		
    		\draw[thick,dotted] (a) -- (b);
    		\draw[thick] (f) to node[midway,inner sep=1pt,outer sep=1pt,minimum size=4pt,fill=white] {$e$} (g);
    		\draw[thick, snake=zigzag] (e) -- (f) -- (r) (g) -- (h) (l) -- (m) (q) -- (g);
    		\draw[ thick] (b) -- (c) (b) -- (d) (b) -- (e) (h) --  (i) (h) -- (j) (h) - -(k) (h) -- (l) (m) -- (n) (m) -- (o) (m) -- (p) (m) -- (q);
    		
    	\end{tikzpicture}
    	\caption{Intersection between $P_j,\, \tilde P$ at an edge $e$}\label{fig:Viz_path_path_intersect}
    \end{subfigure}
    \begin{subfigure}[t]{0.45\textwidth}
        \centering
    	\begin{tikzpicture}
    	    \node[circle,fill=black,draw,inner sep=0pt,minimum size=4pt] (a) at (0,0) {};
    		\node[circle,fill=black,draw,inner sep=0pt,minimum size=4pt] (b) at (1,0) {};
    		\path (b) ++(135:1) node[circle,fill=black,draw,inner sep=0pt,minimum size=4pt] (c) {};
    		\path (b) ++(90:1) node[circle,fill=black,draw,inner sep=0pt,minimum size=4pt] (d) {};
    		\path (b) ++(45:1) node[circle,fill=black,draw,inner sep=0pt,minimum size=4pt] (e) {};
    		
    		\draw[decoration={brace,amplitude=10pt},decorate] (-0.2,0.35) -- node [midway,above,yshift=5pt,xshift=-10pt] {$F_j$} (1.3, 1.3);
    		
    		\path (e) ++(0:1.5) node[circle,fill=black,draw,inner sep=0pt,minimum size=4pt] (f) {};
    		\path (f) ++(0:1) node[circle,fill=black,draw,inner sep=0pt,minimum size=4pt] (g) {};
    		\path (g) ++(0:1.5) node[circle,fill=black,draw,inner sep=0pt,minimum size=4pt] (h) {};
    		
    		\draw[decoration={brace,amplitude=10pt}, decorate] (1.9, 1.1) -- node [midway,above,xshift=0pt,yshift=10pt] {$P_j$} (6.4,1.1);
    		
    		\path (h) ++(-30:1) node[circle,fill=black,draw,inner sep=0pt,minimum size=4pt] (i) {};
    		\path (h) ++(-60:1) node[circle,fill=black,draw,inner sep=0pt,minimum size=4pt] (j) {};
    		\path (h) ++(-90:1) node[circle,fill=black,draw,inner sep=0pt,minimum size=4pt] (k) {};
    		\path (h) ++(0:1) node[circle,fill=black,draw,inner sep=0pt,minimum size=4pt] (l) {};

    		\path (f) ++(-120:1) node[circle,fill=black,draw,inner sep=0pt,minimum size=4pt] (m) {};
    		\path (f) ++(-80:1) node[circle,fill=black,draw,inner sep=0pt,minimum size=4pt] (n) {};
    		\path (f) ++(-40:1) node[circle,fill=black,draw,inner sep=0pt,minimum size=4pt] (o) {};

    		\draw[decoration={brace,amplitude=10pt,mirror},decorate] (3,-0.55) -- node [midway,yshift=-5pt,xshift=17pt] {$\tilde F$} (4.6, 0.5);
    		
    		\path (f) ++(60:2.5) node[circle,fill=black,draw,inner sep=0pt,minimum size=4pt] (p) {};
    		\path (p) ++(30:1) node[circle,fill=black,draw,inner sep=0pt,minimum size=4pt] (q) {};
    		\path (q) ++(180:1) node[circle,fill=black,draw,inner sep=0pt,minimum size=4pt] (r) {};
    		\path (q) ++(150:1) node[circle,fill=black,draw,inner sep=0pt,minimum size=4pt] (s) {};
    		\path (q) ++(120:1) node[circle,fill=black,draw,inner sep=0pt,minimum size=4pt] (t) {};

    		\draw[thick,dotted] (a) -- (b);
    		\draw[thick] (f) to node[midway,inner sep=1pt,outer sep=1pt,minimum size=4pt,fill=white] {$e$} (g);
    		\draw[thick, snake=zigzag] (e) -- (f) (g) -- (h) (f) -- (p) (q) -- (l);
    		\draw[ thick] (b) -- (c) (b) -- (d) (b) -- (e) (h) -- (i) (h) -- (j) (h) -- (k) (h) -- (l) (f) -- (m) (f) -- (n) (f) -- (o) (p) -- (q) -- (r) (q) -- (s) (q) -- (t);

    	\end{tikzpicture}
    	\caption{Intersection between $P_j,\,\tilde F$ at an edge $e$}\label{fig:Viz_path_fan_intersect}
    \end{subfigure}

    \vspace{15pt}
    
    \begin{subfigure}[b]{0.45\textwidth}
        \centering
    	\begin{tikzpicture}
    	    \node[circle,fill=black,draw,inner sep=0pt,minimum size=4pt] (a) at (0,0) {};
    		\node[circle,fill=black,draw,inner sep=0pt,minimum size=4pt] (b) at (1,0) {};
    		\path (b) ++(120:1) node[circle,draw,inner sep=0pt,minimum size=4pt] (c) {$w$};
    		\path (b) ++(60:1) node[circle,fill=black,draw,inner sep=0pt,minimum size=4pt] (d) {};
    		\path (b) ++(0:1) node[circle,fill=black,draw,inner sep=0pt,minimum size=4pt] (e) {};
    		
    		\draw[decoration={brace,amplitude=10pt,mirror}, decorate] (0.2, -0.2) -- node [midway,below,xshift=0pt,yshift=-10pt] {$F_j$} (1.9,-0.2);
    		
    		\path (e) ++(0:2.5) node[circle,fill=black,draw,inner sep=0pt,minimum size=4pt] (f) {};
    		\path (f) ++(0:1) node[circle,fill=black,draw,inner sep=0pt,minimum size=4pt] (g) {};
    		\path (f) ++(30:1) node[circle,fill=black,draw,inner sep=0pt,minimum size=4pt] (h) {};
    		\path (f) ++(60:1) node[circle,fill=black,draw,inner sep=0pt,minimum size=4pt] (i) {};
    		
    		\draw[decoration={brace,amplitude=10pt,mirror}, decorate] (2, -0.2) -- node [midway,below,xshift=0pt,yshift=-10pt] {$P_j$} (5.5,-0.2);
    		
    		\path (i) ++(135:2) node[circle,fill=black,draw,inner sep=0pt,minimum size=4pt] (j) {};
    		\path (j) ++(135:1) node[circle,fill=black,draw,inner sep=0pt,minimum size=4pt] (k) {};
    		\path (j) ++(165:1) node[circle,fill=black,draw,inner sep=0pt,minimum size=4pt] (l) {};
    		\path (j) ++(195:1) node[circle,fill=black,draw,inner sep=0pt,minimum size=4pt] (m) {};
    		
    		\draw[decoration={brace,amplitude=10pt}, decorate] (2.3, 2.1) -- node [midway,above,xshift=-15pt,yshift=-5pt] {$\tilde F$} (3,3.2);
    		
    		\path (c) ++(-150:0.75) node[circle,fill=black,draw,inner sep=0pt,minimum size=4pt] (n) {};
    		
    		\draw[decoration={brace,amplitude=10pt},decorate] (-0.3,0.6) -- node [midway,above,yshift=5pt,xshift=-10pt] {$\tilde P$} (2.5,2.1);

    		\draw[thick,dotted] (a) to (b);
    		\draw[thick, snake=zigzag] (e) -- (f) (i) -- (j) (m) -- (c) (c) -- (n);
    		\draw[thick] (b) -- (c) (b) -- (d) (b) -- (e) (f) -- (g) (f) -- (h) (f) -- (i) (j) -- (k) (j) -- (l) (j) -- (m);
    		
    	\end{tikzpicture}
    	
    	\caption{Intersection between $F_j,\, \tilde P$ at a vertex $w$.}\label{fig:Viz_fan_path_intersect}
    \end{subfigure}
    \begin{subfigure}[b]{0.45\textwidth}
        \centering
    	\begin{tikzpicture}
    	    \node[circle,fill=black,draw,inner sep=0pt,minimum size=4pt] (a) at (0,0) {};
    		\node[circle,fill=black,draw,inner sep=0pt,minimum size=4pt] (b) at (0,-1) {};
    		\path (b) ++(45:1) node[circle,fill=black,draw,inner sep=0pt,minimum size=4pt] (c) {};
    		\path (b) ++(0:1) node[circle,draw,inner sep=0pt,minimum size=4pt] (d) {$w$};
    		\path (b) ++(-45:1) node[circle,fill=black,draw,inner sep=0pt,minimum size=4pt] (e) {};
    		\path (b) ++(-90:1) node[circle,fill=black,draw,inner sep=0pt,minimum size=4pt] (f) {};
    		
    		\draw[decoration={brace,amplitude=10pt},decorate] (-0.2,-2) -- node [midway,above,yshift=0pt,xshift=-15pt] {$F_j$} (-0.2, 0);
    		
    		\path (f) ++(-90:1.5) node[circle,fill=black,draw,inner sep=0pt,minimum size=2pt] (g) {};
    		\path (g) ++(0:2) node[circle,fill=black,draw,inner sep=0pt,minimum size=2pt] (h) {};
    		\path (h) ++(90:2.5) node[circle,fill=black,draw,inner sep=0pt,minimum size=4pt] (i) {};
    		
    		\draw[decoration={brace,amplitude=10pt},decorate] (2.2,0) -- node [midway,above,yshift=0pt,xshift=15pt] {$\tilde F$} (2.2, -2);
    		
    		\path (i) ++(90:1) node[circle,fill=black,draw,inner sep=0pt,minimum size=4pt] (j) {};
    		\path (i) ++(135:1) node[circle,fill=black,draw,inner sep=0pt,minimum size=4pt] (k) {};
    		\path (i) ++(-135:1) node[circle,fill=black,draw,inner sep=0pt,minimum size=4pt] (l) {};
    		
    		\draw[thick,dotted] (a) to (b);
    		\draw[thick, snake=zigzag] (f) -- (g) -- (h) -- (i);
    		\draw[thick] (b) -- (c) (b) -- (d) (b) -- (e) (b) -- (f) (i) -- (j) (i) -- (k) (i) -- (l) (i) -- (d);
    		
    	\end{tikzpicture}
    	
    	\caption{Intersection between $F_j,\, \tilde F$ at a vertex $w$}\label{fig:Viz_fan_fan_intersect}
    \end{subfigure}
    \caption{Intersecting Vizing Chains at step \ref{step:truncate1}.}
    \label{fig:Viz_intersect}
\end{figure}

\begin{Lemma}\label{lemma:counting_wD}
    Let $D = (d_1, \ldots, d_t) \in \mathcal{D}^{(t)}$ and $uv \in E$.
    Then $|\mathcal{I}^{(t)}(D, uv, u)| \leq \wt(D)$, for $\ell \geq \Delta$.
\end{Lemma}
\begin{proof}\stepcounter{ForClaims} \renewcommand{\theForClaims}{\ref{lemma:counting_wD}}
    We will prove this by induction on $t$.    
    Let us first consider $t = 0$.
    In this case, we have 
    \[\mathcal{I}^{(0)} = \set{(e, x)\,:\, e \in E,\, x \in e,\, \phi(e) = \blank}.\]
    In particular, $\mathcal{I}^{(0)}((), uv, u) = \set{(uv, v)}$ contains just $1$ element, where $\wt(D) = \wt(()) = 1$.

    Now assume that $t \geq 1$. Consider any $I = (f, z, \ell_1, \ldots, \ell_t) \in \mathcal{I}^{(t)}(D, uv, u)$  
    and run Algorithm~\ref{alg:multi_viz_chain} with input sequence $I$ for $t$ iterations of the \textsf{while} loop. Let the chain at the start of the $t$-th iteration be $C = F_0 + P_0 + \cdots + F_{k-1} + P_{k-1}$ and let the candidate chain be $F+P$. On step \ref{step:Pk}, we compute a shortened path $P_k \defeq P|\ell_t$ and then compute $\tilde F + \tilde P$ on step \ref{step:alpha_beta_order}.
    Note that for
    \[
        D' \,\defeq\, (d_1, \ldots, d_{t-1}), \qquad v' \,\defeq\, \Pivot(F), \qquad \text{and} \qquad u' \,\defeq\, \vstart(F),
    \]
    we have $(f, z, \ell_1, \ldots, \ell_{t - 1}) \,\in\, \mathcal{I}^{(t-1)}(D', u'v', u')$. 
    Since $|\mathcal{I}^{(t-1)}(D', u'v', u')| \leq \wt(D')$ by the inductive hypothesis, it is enough to show that the number of choices for
    $(u'v', u', \ell_t)$ given $(uv, u)$ and $\phi$ is at most $\val(d_t)$. 

    \begin{claim*}
        Given the vertex $\vend(P_k) = \vstart(\tilde F)$, there are at most $\Delta^4$ options for $(u'v', u', \ell_t)$.
    \end{claim*}
    \begin{claimproof}
       Knowing $\vstart(\tilde F)$, we have at most $\Delta$ choices for $\Pivot(\tilde F)$. Note that $\End(P_k)$ is the edge joining $\Pivot(\tilde F)$ to $\vstart(\tilde F) = \vend(P_k)$. By Lemma~\ref{lemma:non-intersecting_degrees} \ref{item:related_phi}, all edges of $P_k$ except $\Start(P_k)$ are colored with two colors, say $\alpha$ and $\beta$, one of which is $\phi(\End(P_k))$, and there are at most $\Delta$ choices for the other color. By Lemma~\ref{lemma:non-intersecting_degrees} \ref{item:degree_end} and \ref{item:related_phi}, given $\alpha$ and $\beta$, we can locate $\vend(F)$ by considering the maximal $\alpha\beta$-path under the coloring $\phi$ containing $\End(P_k)$ and picking its endpoint on the side of $\Pivot(\tilde{F})$. Since $v' \in N_G(\vend(F))$ and $u' \in N_G(v')$, there are at most $\Delta^2$ choices for $(u'v',u')$ given $\vend(F)$. The path $P_k$ starts at $v'$, goes to $\vend(F)$, and then follows the edges colored $\alpha$ and $\beta$ to $uv$. Thus,
        we can now determine $\ell_t = \length(P_k)$. In conclusion, there are at most $\Delta \cdot \Delta \cdot \Delta^2 = \Delta^4$ possible tuples $(u'v', u', \ell_t)$, as desired.
    \end{claimproof}

    \smallskip

    In view of the above claim, it suffices to show that the number of choices for the vertex $\vend(P_k) = \vstart(\tilde F)$ is at most $\val(d_t)/\Delta^4$. We now split into cases depending on the value of $d_t$.
    
    \begin{enumerate}[label=\ep{\textbf{Case \arabic*}}, wide] 
        \item\label{case:dt=1} $d_t = 1$. Then $u = \vend(P_k)$, so the number of choices for $\vend(P_k)$ is $1 = \val(1)/\Delta^4$. 

        \item\label{case:dtnegative} $d_t < 0$. Then we have $v = \Pivot(F_j)$ and $u = \vstart(F_j)$, where $j = k + d_t$ and the first intersection between $\tilde F + \tilde P$ and $C + F + P_k$ was with $F_j + P_j$. The vertex $\vend(F_j)$ is a neighbor of $v$, so there are at most $\Delta$ choices for $\vend(F_j)$. By Lemma~\ref{lemma:non-intersecting_degrees} \ref{item:related_phi}, all edges of $P_j$ except $\Start(P_j)$ are colored with two colors, say $\gamma$ and $\delta$, for which there are at most $(\Delta + 1)^2$ choices. Let $P' \defeq G(\vend(F_j); \phi, \gamma\delta)$. By Lemma~\ref{lemma:non-intersecting_degrees} \ref{item:degree_end}, $P'$ is a path starting at $\vend(F_j)$, and the chain $P_j$ consists of the edge $\End(F_j)$ followed by an initial segment of $P'$ of length less than $2\ell$.
        
        Recall that we have $V(\tilde F + \tilde P) \cap V(F_j) \neq \0$ or $E(\tilde F + \tilde P) \cap \IE(P_j) \neq \0$. The possible configurations of this intersection are shown in Fig.~\ref{fig:Viz_intersect}. In particular, we must have $V(\tilde F + \tilde P) \cap V(F_j + P_j) \neq \0$. We now consider two subcases.

        \begin{enumerate}[label=\ep{\textit{Subcase 2\alph*}}, wide]
            \item\label{subcase:tildefan} $V(\tilde F) \cap V(F_j + P_j) \neq \0$. If $V(\tilde F) \cap V(F_j) \neq \0$, then $\vstart(\tilde F)$ is at distance at most $3$ from $v$ in $G$, so there are at most $\Delta^3 + 1$ options for it. On the other hand, if $V(\tilde F) \cap V(P_j) \neq \0$, then we have at most $2\ell$ choices for the intersection point on $P'$, and hence at most $2\ell (\Delta^2 + 1)$ choices for $\vstart(\tilde F)$. In total, we have at most $\Delta^3 + 1 + \Delta(\Delta+1)^2\,2\ell (\Delta^2 + 1)$ options for $\vstart(\tilde F)$. 

            \item\label{subcase:tildepath} $V(\tilde F) \cap V(F_j + P_j) = \0$, so $(V(\tilde P) \setminus \Start(\tilde P)) \cap V(F_j + P_j) \neq \0$. Let $w$ be the first vertex in $V(\tilde P) \setminus \Start(\tilde P)$ that belongs to $V(F_j + P_j)$. There are at most $\Delta + 1 + \Delta(\Delta+1)^2\,2\ell$ choices for $w$, because it can be found either in $N_G[v]$ or among the first $2\ell$ vertices on $P'$. 
            Next, there are at most $(\Delta + 1)^2$ choices for the colors $\zeta$ and $\eta$ such that $\tilde P$ is a $\zeta\eta$-path in the coloring $\Shift(\phi, C + F + P_k)$. 
            By the choice of $j$ and since we are not in \ref{subcase:tildefan}, $V(\tilde F) \cap V(F_i) = \0$ for all $i \leq k$. 
            Hence, the same argument as in the proof of Lemma \ref{lemma:non-intersecting_degrees} \ref{item:degree_end} shows that $\deg(\vend(\tilde F); \phi, \zeta\eta) = 1$. 
            Clearly, the vertices $\vend(\tilde F)$ and $w$ are $(\Shift(\phi, C + F + P_k), \zeta\eta)$-related, and  by the choice of $j$ and $w$, 
            they must be $(\phi, \zeta\eta)$-related as well. Therefore, $\vend(\tilde F)$ is an endpoint of the path $G(w; \phi, \zeta \eta)$, and thus there are at most $2$ choices for $\vend(\tilde F)$ given $w$, $\zeta$, and $\eta$. Finally, $\vstart(\tilde{F})$ and $\vend(\tilde{F})$ are joined by a path of length $2$ in $G$, so we have at most $\Delta^2$ choices for $\vstart(\tilde F)$ given $\vend(\tilde F)$. In summary, we have at most $2(\Delta + 1 + \Delta(\Delta+1)^2\,2\ell)(\Delta + 1)^2 \Delta^2$ options for $\vstart(\tilde F)$.
        \end{enumerate}

        Combining the two subcases, we see that the total number of choices for $\vstart(\tilde F)$ is at most
        \begin{align*}
            \Delta^3 + 1 + &\Delta(\Delta+1)^2\,2\ell (\Delta^2 + 1) + 2(\Delta + 1 + \Delta(\Delta+1)^2\,2\ell)(\Delta + 1)^2 \Delta^2 \\
            &=\, \Delta^3 + 1 + 2(\Delta + 1)^3\Delta^2 + \Delta(\Delta+1)^2\,2\ell (\Delta^2 + 1 + 2(\Delta+1)^2 \Delta^2) \\
            &\leq\, 18\Delta^5 + 72\ell \Delta^7 
            \,\leq\,
            \frac{\val(d_t)}{\Delta^4},
        \end{align*}
        where in the last step we use that $\ell \geq \Delta \geq 2$.

        \item $d_t = 0$. This case is almost verbatim the same as \ref{case:dtnegative} with $j = k$. The only difference is that, by Lemma~\ref{lemma:intersection_prev}, the first intersection between $\tilde F + \tilde P$ and $F_k + P_k$ must occur at a vertex in $V(F_k)$. Therefore, the number of choices for $\vstart(\tilde F)$ in \ref{subcase:tildefan} becomes only $\Delta^3 + 1$, while the number of choices for $w$ in \ref{subcase:tildepath} is $\Delta + 1$. With this adjustment, the total number of choices for $\vstart(\tilde F)$ becomes at most
        \[
            \Delta^3 + 1 + 2(\Delta + 1)^3 \Delta^2 \,\leq\, 18\Delta^5 \,\leq\, \frac{\val(0)}{\Delta^4}.
        \]
    \end{enumerate}
    
    This covers all the cases and completes the proof.
\end{proof}

    Next we prove an upper bound on $\wt(D)$ in terms of the sum of the entries of $D$. To this end, let
\[\mathcal{D}_s^{(t)} \defeq \left\{D = (d_1, \ldots, d_t)\in \mathcal{D}^{(t)}\,:\, \sum_{i = 1}^td_i = s\right\}.\]

\begin{Lemma}\label{lemma:wD_bound}
    Let $D \in \mathcal{D}_s^{(t)}$. Then $\wt(D) \leq (75\Delta^{15}\ell)^{t/2}\,(75\Delta^7\ell)^{-s/2}$, for $\ell \geq 6\Delta^3$.
\end{Lemma}

\begin{proof}
    Let $D = (d_1, \ldots, d_t) \in \mathcal{D}_s^{(t)}$.
    Define
    \begin{align*}
        I\defeq \{i\,:\,d_i = 1\}, \quad J\defeq \{i\,:\,d_i = 0\}, \quad K\defeq [t] \setminus (I\cup J).
    \end{align*}
    Using the definition of $\wt(D)$ and $\val(z)$, we write
    \begin{align*}
        \wt(D) \,&=\, \Delta^{4|I|}\,(20\Delta^9)^{|J|}\, (75\ell\Delta^{11})^{|K|} \\
        &=\, 20^{|J|}\,75^{|K|}\,\ell^{|K|}\,\Delta^{4|I|+9|J|+11|K|} \\
        &=\, 20^{|J|}\,75^{|K|}\,\ell^{|K|}\,\Delta^{4t+5|J|+7|K|},
    \end{align*}
    where we use the fact that $|I| + |J| + |K| = t$.
    Note the following:
    \[s \,=\, |I| - \sum_{k \in K}|d_k| \,\leq\, |I| - |K| \,=\, t - |J| - 2|K|.\]
    It follows that $|K| \leq (t - s - |J|)/2$. With this in mind, we have 
    \begin{align*}
        \wt(D) \,&\leq\, 20^{|J|}\, (75\,\ell)^{(t-s - |J|)/2} \, \Delta^{4t+ 5|J| + 7(t-s-|J|)/2}\\
        &=\, (75\Delta^{15}\ell)^{t/2}\,(75\Delta^7\ell)^{-s/2}\,\left(\frac{400\Delta^3}{75\ell}\right)^{|J|/2} \\
        & \leq\,  (75\Delta^{15}\ell)^{t/2}\,(75\Delta^7\ell)^{-s/2},
    \end{align*}
    for $\ell \geq 6\Delta^3$.
\end{proof}

    Now we bound the size of $\mathcal{D}_s^{(t)}$.

    \begin{Lemma}\label{lemma:Dst}
        $|\mathcal{D}_s^{(t)}|\leq 4^t$.
    \end{Lemma}
    \begin{proof}
        If $D = (d_1, \ldots, d_t) = D(I)$ for some $I \in \mathcal{I}^{(t)}$, then at the end of the $j$-th iteration of the \textsf{while} loop with input sequence $I$, $C$ is a $k(j)$-step Vizing chain, where $k(j) \defeq \sum_{i = 1}^j d_i$.
        Therefore, $k(j)$ must be nonnegative.
        Following the seminal paper of Grytczuk, Kozik, and Micek \cite{Grytczuk}, we consider the set $N(s,t)$ of all sequences of integers $(d_1, \ldots, d_t)$ such that
        \begin{enumerate}
            \item $\sum\limits_{i = 1}^j d_i \geq 1$ for all $1 \leq j \leq t$,
            \item $d_i \leq 1$ for all $1 \leq i \leq t$, and
            \item $\sum\limits_{i = 1}^td_i = s$.
        \end{enumerate}
        In the proof of \cite[Theorem 1]{Grytczuk}, Grytczuk \emph{et al.} observe the connection of the numbers $|N(s,t)|$ to the Catalan numbers. In particular, they show that
        \[|N(s,t)| \,\leq\, \frac{1}{t} \binom{2(t-1)}{t-1} \,\leq\, 4^{t-1}.\]
        Let $D = (d_1, \ldots, d_t) \in \mathcal{D}_s^{(t)}$. Then the sequence $(1, d_1, \ldots, d_t)$ is in $N(s+1, t+1)$. Therefore, 
        \[|\mathcal{D}_s^{(t)}| \,\leq\, |N(s+1, t+1)| \,\leq\, 4^t,\]
        as desired.
    \end{proof}

    Lastly, we need a lemma that bounds the running time of Algorithm~\ref{alg:multi_viz_chain} in terms of the number of iterations of the \textsf{while} loop:

    \begin{Lemma}\label{lemma:runtime}
        The first $t$ iterations of the {\upshape\textsf{while}} loop in Algorithm~\ref{alg:multi_viz_chain} are performed in time $O(\ell \, t)$.
    \end{Lemma}
    \begin{proof}
        Let $I \in \mathcal{I}^{(t)}$ be an input sequence and let $D = (d_1, \ldots, d_t) \defeq D(I)$. We bound the runtime of the $i$-th iteration in terms of $d_i$. As shown in \S\ref{subsec:algo_overview}, Algorithms \ref{alg:first_chain} and \ref{alg:next_chain} run in $O(\ell)$ time. Also, while updating the coloring $\psi$, we must update the sets of missing colors as well. Each update to $M(\cdot)$ takes $O(\Delta)$ time. Let us now consider two cases depending on whether $d_i$ is positive.
        
    \begin{enumerate}[label=\ep{\textbf{Case \arabic*}}, wide] 
        \item $d_i = 1$. Here, we test for success, shorten the path $P$ to $P_k$, update $\psi$ and the hash map $\visited$,
        call Algorithm \ref{alg:next_chain}, and check that the resulting chain is non-intersecting. All of these steps can be performed in time $O(\ell)$.
        
        \item $d_i \leq 0$. Here we again conduct all the steps mentioned in the previous case, which takes $O(\ell)$ time.
        Once we locate the intersection, we then update $\psi$ and the hash map $\visited$, which takes $O((|d_i| + 1)\ell)$ time (we are adding $1$ to account for the case $d_i = 0$).
    \end{enumerate}
    
    Since $\sum_{i=1}^t d_i \geq 0$, we have
    \[\sum_{i,\, d_i \leq 0}|d_i| \leq \sum_{i,\, d_i = 1}d_i \,\leq\, t.\]
    It follows that the running time of $t$ iterations of the \textsf{while} loop is $O(\ell\,t)$, as desired.
    \end{proof}
    
We are now ready to prove Theorem \ref{theo:entropy_compression_dist}.

\begin{proof}[Proof of Theorem \ref{theo:entropy_compression_dist}]
    Note that, since each $\ell_i$ is chosen independently and uniformly at random from the $\ell$-element set $[\ell, 2\ell-1] \cap \N$, the probability that Algorithm \ref{alg:multi_viz_chain} lasts for more than $t$ iterations of the \textsf{while} loop is bounded above by $|\mathcal{I}^{(t)}|/\ell^t$. We can bound $|\mathcal{I}^{(t)}|$ by considering all possible records and termini. In particular, as a result of Lemmas \ref{lemma:counting_wD}, \ref{lemma:wD_bound}, and \ref{lemma:Dst}, we have
    \begin{align*}
        |\mathcal{I}^{(t)}| \,\leq\, \sum_{\substack{e \in E, u \in e, \\ D \in D^{(t)}}}|\mathcal{I}^{(t)}(D, e, u)| \,&\leq\, \sum_{\substack{e \in E, u \in e, \\ D \in D^{(t)}}}\wt(D) \\
        &\leq\, 2m\sum_{s = 0}^t\sum_{D \in \mathcal{D}_s^{(t)}}\wt(D) \\
        &\leq\, 2m\sum_{s = 0}^t|\mathcal{D}_s^{(t)}|(75\Delta^{15}\ell)^{t/2}\,(75\Delta^7\ell)^{-s/2} \\
        &\leq\, 2m(1200\Delta^{15}\ell)^{t/2}\sum_{s = 0}^t(75\Delta^7\ell)^{-s/2} \\
        &\leq\, 4m(1200\Delta^{15}\ell)^{t/2}.
    \end{align*}
    Therefore, we may conclude that
    \begin{align*}
        \P[\text{Algorithm \ref{alg:multi_viz_chain} does not terminate after } t \text{ iterations}] \,&\leq\, \frac{4m\,(1200\Delta^{15}\ell)^{t/2}}{\ell^t} \,=\, 4m\,\left(\frac{1200\Delta^{15}}{\ell}\right)^{t/2},
    \end{align*}
    as desired. 
    The bound $O(\ell\, t)$ on the runtime of the algorithm follows by Lemma~\ref{lemma:runtime}.
\end{proof}

\section{Proof of Theorem \ref{theo:seq}}\label{sec:sequential}

In this section, we prove Theorem \ref{theo:seq} by describing a randomized algorithm similar to the one in \S \ref{sec:nlogn} and based on Algorithm Sketch \ref{inf:seq}. The algorithm takes as input a graph $G$ of maximum degree $\Delta$ and a parameter $\ell \in \N$, and returns a proper $(\Delta + 1)$-edge-coloring of $G$.
At each iteration, the algorithm picks a random uncolored edge and colors it by finding a Multi-Step Vizing Chain using Algorithm \ref{alg:multi_viz_chain}.

\begin{algorithm}[h]\algsize
\caption{Sequential Coloring with Multi-Step Vizing Chains}\label{alg:seq}
\begin{flushleft}
\textbf{Input}: A graph $G = (V, E)$ of maximum degree $\Delta$ and a parameter $\ell \in \N$. \\
\textbf{Output}: A proper $(\Delta(G)+1)$-edge-coloring $\phi$ of $G$.
\end{flushleft}
\begin{algorithmic}[1]
    \State $U \gets E$, $\phi(e) \gets \blank$ for each $e \in U$.
    \While{$U \neq \0$}
        \State Pick an edge $e \in U$ and a vertex $x \in e$ uniformly at random.
        \State $C \gets \hyperref[alg:multi_viz_chain]{\mathsf{MSVA}}(\phi, e, x, \ell)$ \Comment{Algorithm \ref{alg:multi_viz_chain}}
        \State $\phi \gets \aug(\phi, C)$
        \State $U \gets U \setminus \set{e}$
    \EndWhile
    \State \Return $\phi$
\end{algorithmic}
\end{algorithm}

The correctness of this algorithm follows from the results in \S \ref{subsec:algo_lemmas}. 
To assist with the time analysis, we define the following variables:
\begin{align*}
    \phi_i &\defeq \text{ the coloring at the start of iteration } i, \\
    T_i&\defeq \text{ the number of iterations of the \textsf{while} loop in the $i$-th call to Algorithm \ref{alg:multi_viz_chain}},  \\
    C_i &\defeq \text{ the chain produced by Algorithm \ref{alg:multi_viz_chain} at the $i$-th iteration}, \\
    U_i &\defeq \set{e\in E\,:\, \phi_i(e) = \blank}.
\end{align*}
Augmenting the chain $C_i$ takes time $O(\length(C_i))$.
From Lemma \ref{lemma:runtime} and since $\length(C_i) = O(\ell \,T_i)$, we see that the runtime of Algorithm \ref{alg:seq} is $O(\ell\,T)$, where $T = \sum_{i = 1}^mT_i$.
It remains to show that $T = O(m)$ with high probability. 
To this end, we prove the following lemma.

\begin{Lemma}\label{lemma:t_i_prob_bound}
    For all $t_1$, \ldots, $t_i \in \N$, we have $\P[T_i > t_i|T_1 > t_1, \ldots, T_{i-1} > t_{i-1}] \leq \frac{2m}{|U_i|}\,\left(\frac{1200\Delta^{15}}{\ell}\right)^{t_i/2}$.
\end{Lemma}

\begin{proof}
    To analyze the number of iterations of the \textsf{while} loop in the $i$-th call to Algorithm \ref{alg:multi_viz_chain}, we use the definition of input sequences from \S\ref{sec:msva_analysis}. Note that an input sequence $I = (f, z, \ell_1, \ldots, \ell_t)$ contains not only the random integers $\ell_1$, \ldots, $\ell_t$, but also the starting edge $f$ and the vertex $z \in f$. Since at every iteration of Algorithm~\ref{alg:seq}, the starting edge and vertex are chosen uniformly at random, we have
    \[\P[T_i > t_i|\phi_i] \,\leq\, \frac{|\mathcal{I}^{(t_i)}|}{2|U_i|\ell^{t_i}}.\]
    As shown in the proof of Theorem \ref{theo:entropy_compression_dist},
    \[|\mathcal{I}^{(t_i)}| \,\leq\,
    4m(1200\Delta^{15}\ell)^{t_i/2}.\]
    Furthermore, $T_i$ is independent of $T_1$, \ldots, $T_{i-1}$ given $\phi_i$. Therefore, 
    \begin{align*}
        \P[T_i > t_i|T_1 > t_1, \ldots T_{i-1} > t_{i-1}] \,\leq\, \frac{|\mathcal{I}^{(t_i)}|}{2|U_i|\ell^{t_i}} \,\leq\, \frac{2m}{|U_i|}\,\left(\frac{1200\Delta^{15}}{\ell}\right)^{t_i/2},
    \end{align*}
    as desired.
\end{proof}
Let us now bound $\P[T > t]$. 
Note that if $T > t$, then there exist non-negative integers $t_1$, \ldots, $t_m$ such that $\sum_it_i = t$ and $T_i > t_i$ for each $i$.
Using the union bound
we obtain
\begin{align*}
    \P[T > t] \,&\leq\, \sum_{\substack{t_1, \ldots, t_m \\ \sum_it_i = t}}\P[T_1 > t_1, \ldots, T_m > t_m] \\
    &=\, \sum_{\substack{t_1, \ldots, t_m \\ \sum_it_i = t}}\,\prod_{i = 1}^m\P[T_i > t_i|T_1 > t_1, \ldots, T_{i-1} > t_{i-1}] \\
    [\text{by Lemma \ref{lemma:t_i_prob_bound}}]\qquad &\leq\, \sum_{\substack{t_1, \ldots, t_m \\ \sum_it_i = t}}\,\prod_{i = 1}^m\frac{2m}{|U_i|}\,\left(\frac{1200\Delta^{15}}{\ell}\right)^{t_i/2} \\
    &=\, \sum_{\substack{t_1, \ldots, t_m \\ \sum_it_i = t}}\frac{(2m)^m}{m!}\,\left(\frac{1200\Delta^{15}}{\ell}\right)^{t/2} \\
    [m! \geq (m/e)^m] \qquad &\leq\, \binom{t+m-1}{m-1} \,(2e)^{m}\,\left(\frac{1200\Delta^{15}}{\ell}\right)^{t/2} \\
    [{\textstyle {a \choose b} \leq (ea/b)^b}] \qquad &\leq\, \left(e\left(1 + \frac{t}{m-1}\right)\right)^{m-1}\,(2e)^{m}\,\left(\frac{1200\Delta^{15}}{\ell}\right)^{t/2} \\
    [(1 + a)\leq e^a] \qquad &\leq\, (2e^2)^m\,\left(\frac{1200\,e^2\,\Delta^{15}}{\ell}\right)^{t/2}.
\end{align*}
For $t = \Theta(m)$ and $\ell = \Theta(\Delta^{16})$, the above is at most $1/\Delta^n$. It follows that Algorithm \ref{alg:seq} computes a proper $(\Delta+1)$-edge-coloring in $O(\Delta^{17}\,n)$ time with probability at least $1 - 1/\Delta^n$, as desired.

\section{Proof of Theorem \ref{theo:dist}}\label{sec:distributed}

\subsection{Finding many disjoint augmenting subgraphs}\label{subsec:many_aug_subgraphs}

In this subsection we prove the following version of Theorem \ref{theo:disjoint} with explicit dependence on the bounds on $\Delta$:

\begin{theo}[Exact version of Theorem~\ref{theo:disjoint}]\label{theo:disjoint_with_bounds}
    Let $\phi \colon E \pto [\Delta + 1]$ be a proper partial $(\Delta + 1)$-edge-coloring with domain $\dom(\phi) \subset E$ and let $U \defeq E \setminus \dom(\phi)$ be the set of uncolored edges. Then there exists a randomized \LOCAL algorithm that in $O(\Delta^{16} \log n)$ rounds outputs a set $W \subseteq U$ of expected size $\E[|W|] = \Omega(|U|/\Delta^{20})$ and an assignment of 
    connected $e$-augmenting subgraphs ${H_e}$ to the edges $e \in W$ such that:
    \begin{itemize}
        \item the graphs ${H_e}$, $e \in W$ are pairwise vertex-disjoint, and
        \item $|E({H_e})| = O(\Delta^{16}\log n)$ for all $e \in W$.
    \end{itemize}
\end{theo}

    Let $\phi$, $U$ be as in Theorem~\ref{theo:disjoint_with_bounds}. Throughout this subsection, we fix $\ell = \Theta(\Delta^{16})$ and $t = \Theta(\log n)$, where the implicit constant factors are assumed to be sufficiently large. Our algorithm for Theorem~\ref{theo:disjoint_with_bounds} starts with every uncolored edge $e \in U$ running Algorithm~\ref{alg:multi_viz_chain} in parallel for $t$ iterations of the \textsf{while} loop.  Note that the first $t$ iterations of Algorithm \ref{alg:multi_viz_chain} with starting edge $e$ only rely on the information within distance $O(\ell t)$ from $e$, and thus this stage of the algorithm takes $O(\ell t)$ rounds. For each $e \in U$, let $C_e$ be the resulting multi-step Vizing chain. 
Let $S\subseteq U$ be the set of all edges for which Algorithm \ref{alg:multi_viz_chain} succeeds within the first $t$ iterations, and so the chain $C_e$ is $e$-augmenting. By Theorem~\ref{theo:entropy_compression_dist}, for each $e \in U$, we have $\P[e \in S] \geq 1 - 4m(1200\Delta^{15}/\ell)^{t/2}$. Therefore, by taking $\ell = \Theta(\Delta^{16})$ and $t = \Theta(\log n)$, we can ensure that $\E[|S|] \geq 3|U|/4$. 

The problem now is that some of the graphs $C_e$ may intersect each other. To get rid of such intersections, we define an auxiliary graph $\graphrand$ as follows:
\[V(\graphrand) \defeq S, \qquad E(\graphrand) \defeq \{ef\,:\, V(C_e)\cap V(C_f) \neq \0\}.\]
Our goal is to find a large independent set $W$ in $\graphrand$. 
We do this using a simple randomized procedure described in Algorithm \ref{alg:independent_set}. Note that, since each $C_e$ has diameter $O(\ell t)$, running Algorithm~\ref{alg:independent_set} on the graph $\graphrand$ takes $O(\ell t)$ rounds. 
The following well-known lemma gives a lower bound on the expected size of the independent set produced by this algorithm (see, e.g., \cite[100]{AS}):

\begin{algorithm}[t]\algsize
    \caption{Distributed Random Independent Set}\label{alg:independent_set}
    \begin{flushleft}
        \textbf{Input}: A graph $\Gamma$. \\
        \textbf{Output}: An independent set $W$.
    \end{flushleft}
    \begin{algorithmic}[1]
        \State Let $x_v$ be i.i.d. $\mathcal{U}(0,1)$ for each $v \in V(\Gamma)$.
        \If{$x_v > \max_{u \in N_\Gamma(v)}x_u$}
            \State $W \gets W \cup \set{v}$
        \EndIf
        \State \Return $W$
    \end{algorithmic}
\end{algorithm}

\begin{Lemma}\label{lemma:alg_indep_set}
    Let $W$ be the independent set returned by Algorithm \ref{alg:independent_set} on a graph $\Gamma$ of average degree $d(\Gamma)$. 
    Then $\E[|W|] \geq |V(\Gamma)|/(d(\Gamma)+1)$.
\end{Lemma}

\begin{proof}
    The set $W$ is independent, since otherwise there would exist an edge $uv$ such that $x_u > x_v$ and $x_v > x_u$. To get a lower bound on $\E[|W|]$, consider any $v \in V(G)$. It is a standard result on the order statistics of i.i.d.~$\mathcal{U}(0,1)$ random variables that
    \[\P\left[x_v = \max_{u\in N_{\Gamma}[v]}x_u\right] \,=\, \frac{1}{\deg_{\Gamma}(v)+1}.\]
    By Jensen's inequality, it follows that
    \[\E[|W|] \,=\, \sum_{v \in V(\Gamma)}\frac{1}{\deg_{\Gamma}(v) + 1} \,\geq\, \frac{|V(\Gamma)|}{d(\Gamma) + 1}. \qedhere\]
\end{proof}

To apply the above result, we must have an upper bound on the average degree $d(\graphrand)$. This is essentially equivalent to bounding the expected number of edges in $\graphrand$.

\begin{Lemma}\label{lemma:edges_bound_Gamma_r}
    $\E[|E(\graphrand)|] \leq 240 \ell\,\Delta^4|U|$ .
\end{Lemma}

\begin{proof}
    We employ the machinery of input sequences and records developed in \S\ref{sec:msva_analysis}. For a vertex $x \in V$ and an integer $1 \leq j \leq t$, 
    let $\mathcal{I}_j(x)$ be the set of all input sequences $I = (f,z, \ell_1, \ldots, \ell_{t_0})$ such that:
    \begin{itemize}
        \item $t_0 \leq t$, and if $t_0 < t$, then Algorithm~\ref{alg:multi_viz_chain} with this input sequence terminates after the $t_0$-th iteration of the \textsf{while} loop,
        \item letting $C(I) = F_0 + P_0 + \cdots + F_{k-1} + P_{k-1}$ be the resulting multi-step Vizing chain after the first $t_0$ iterations of Algorithm \ref{alg:multi_viz_chain}, we have $x \in V(F_{j-1} + P_{j-1})$.
    \end{itemize}
    For an input sequence $I \in \mathcal{I}^{(t_0)}$, we let $|I| \defeq t_0$.
    Let us define $\tilde{\mathcal{I}}_j(x)$ to contain all input sequences $I \in \mathcal{I}^{(t)}$ such that some initial segment $I'$ of $I$ is in $\mathcal{I}_j(x)$.
    Note that since no input sequence in $\mathcal{I}_j(x)$ is an initial segment of another sequence in $\mathcal{I}_j(x)$, $I'$ is unique.
    In particular, for each $I \in \mathcal{I}_j(x)$, there are $\ell^{t - |I|}$ input sequences in $\tilde{\mathcal{I}}_j(x)$.
    It now follows that
    \[\E[c_j(x)] \leq \sum_{I \in \mathcal{I}_j(x)}\frac{1}{\ell^{|I|}} = \frac{1}{\ell^t}\sum_{I \in \mathcal{I}_j(x)}\ell^{t - |I|} = |\tilde{\mathcal{I}}_j(x)|/\ell^t.\]
    Consider any $I \in \mathcal{I}_j(x)$ and let $D(I) = (d_1, \ldots, d_{t_0})$ and $C(I) =  F_0 + P_0 + \cdots + F_{k-1} + P_{k-1}$ be the corresponding record and the resulting multi-step Vizing chain.
    There must be an index $j-1 \leq t' < t_0$ such that after the first $t'$ iterations of the \textsf{while} loop in Algorithm \ref{alg:multi_viz_chain} with input sequence $I$, the current chain is $F_0 + P_0 + \cdots + F_{j-2} + P_{j-2}$ and the candidate chain is $F_{j-1} + P'$, where $P_{j-1}$ is an initial segment of $P'$.
    Letting $I' \defeq (f, z, \ell_1, \ldots, \ell_{t'})$, we have
    \begin{itemize}

        \item $I' \in \mathcal{I}^{(t')}(D(I'), \Start(F_{j-1}), \vstart(F_{j-1}))$, and
        
        \item $D(I') \in \mathcal{D}_{j-1}^{(t')}$.
    \end{itemize}
    Now we consider two cases depending on whether $x$ belongs to $F_{j-1}$ or $P_{j-1}$.
    
    \begin{enumerate}[label=\ep{\textbf{Case \arabic*}}, wide]
        \item $x \in V(F_{j-1})$. Here, $\Pivot(F_{j-1}) \in N_G[x]$ and $\vstart(F_{j-1}) \in N_G(\Pivot(F_{j-1}))$, and so there are at most $(\Delta + 1)\Delta$ choices for $(\Start(F_{j-1}), \vstart(F_{j-1}))$.

        \item $x \in V(P_{j-1}) \setminus V(F_{j-1})$.
        By Lemma \ref{lemma:non-intersecting_degrees} \ref{item:related_phi}, all edges of $P_{j-1}$ are colored with two colors, say $\alpha$ and $\beta$, and there are at most $(\Delta + 1)^2$ choices for them.
        By Lemma \ref{lemma:non-intersecting_degrees} \ref{item:degree_end} and \ref{item:related_phi}, given $\alpha$ and $\beta$, we can locate $\vend(F_{j-1})$ as one of the two endpoints of the maximal $\alpha\beta$-path under the coloring $\phi$ containing $x$.
        Since $\Pivot(F_{j-1}) \in N_G(\vend(F_{j-1}))$ and $v \in N_G(\Pivot(F_{j-1}))$, there are at most $\Delta^2$ choices for $(\Start(F_{j-1}), \vstart(F_{j-1}))$ given $\vend(F_{j-1})$.
        In conclusion, there are at most $2(\Delta + 1)^2\Delta^2$ choices for $(\Start(F_{j-1}), \vstart(F_{j-1}))$ in this case.
    \end{enumerate}

    Therefore, the total number of options for $(\Start(F_{j-1}), \vstart(F_{j-1}))$ is at most
    \[
        (\Delta + 1)\Delta + 2(\Delta + 1)^2\Delta^2 \,\leq\, 10\Delta^4.
    \]
    Note that, given $I'$, we have at most $\ell^{t_0-t'}$ choices for $(\ell_{t'+1}, \ldots, \ell_{t_0})$ and another $\ell^{t - t_0}$ choices for $(\ell_{t_0+1}, \ldots, \ell_{t})$ to define sequences in $\tilde{\mathcal{I}}_j(x)$.
    With this and Lemmas \ref{lemma:counting_wD}, \ref{lemma:wD_bound}, and \ref{lemma:Dst}, we obtain the following chain of inequalities:
    \begin{align}
        |\tilde{\mathcal{I}}_j(x)| \,&\leq\, 10 \Delta^4 \sum_{t' = j-1}^t\sum_{D \in \mathcal{D}_{j-1}^{(t')}}\wt(D)\,\ell^{t-t'} \nonumber\\
        &\leq\, 10\Delta^4\ell^t\sum_{t' = j-1}^t|\mathcal{D}_{j-1}^{(t')}|\left(\frac{75\Delta^{15}}{\ell}\right)^{t'/2}(75\Delta^7\ell)^{-(j-1)/2} \nonumber\\
        &\leq\, 10\Delta^4\ell^t(75\Delta^7\ell)^{-(j-1)/2}\sum_{t' = j-1}^t\left(\frac{1200\Delta^{15}}{\ell}\right)^{t'/2} \nonumber\\
        [\text{assuming $\ell \geq 4800\Delta^{15}$}] \qquad &\leq\, 20\Delta^4\ell^t\left(\frac{4\Delta^{4}}{\ell}\right)^{j-1} \nonumber\\
        &=\, 5\ell^{t+1} \left(\frac{4\Delta^{4}}{\ell}\right)^j. \label{eq:Ij}
    \end{align}
    
    Let $c_j(x)$ be the random variable equal to the number of multi-step Vizing chains $C_e$ for $e \in U$ such that $x$ lies on the $j$-th Vizing chain in $C$. Since an input sequence for Algorithm~\ref{alg:multi_viz_chain} includes the information about the initial uncolored edge, it follows from \eqref{eq:Ij} that 
    \[\E[c_j(x)] \,\leq\, \frac{|\tilde{\mathcal{I}}_j(x)|}{\ell^t} \,\leq\, 5\ell\left(\frac{4\Delta^4}{\ell}\right)^j.\]
    
    For each uncolored edge $e \in U$, let
    \[
        C_e \,=\, F_0^e + P_0^e + \cdots + F_{k_e-1}^e + P_{k_e-1}^e \qquad \text{and} \qquad V^e_j \,\defeq\, V\left(F^e_{j-1} + P^e_{j-1}\right).
    \]
    Let $N^+(e)$ be the set of all edges $f \in U \setminus \set{e}$ such that for some $j \leq j'$, we have
    \[
        V^e_j \cap V^f_{j'} \,\neq\, \0.
    \]
    By definition, if $ef \in E(\graphrand)$, then $V(C_e) \cap V(C_f) \neq \0$, and hence $f \in N^+(e)$ or $e \in N^+(f)$ (or both). Therefore, we have $|E(\graphrand)| \leq \sum_{e \in U} |N^+(e)|$. Note that
    \[
        |N^+(e)| \,\leq\, \sum_{j = 1}^{k_e} \sum_{x \in V^e_j} \sum_{j' = j}^t \sum_{f \in U \setminus \set{e}} \bbone\left\{x \in V^f_{j'}\right\}.
    \]
    Since Algorithm~\ref{alg:multi_viz_chain} is executed independently for every uncolored edge, we have
    \begin{align*}
        \E[|N^+(e)|\big| C_e] \,&\leq\, \sum_{j = 1}^{k_e} \sum_{x \in V^e_j} \sum_{j' = j}^t \sum_{f \in U \setminus \set{e}} \P\left[x \in V^f_{j'} \,\middle\vert\, C_e\right] \\
        &=\, \sum_{j = 1}^{k_e} \sum_{x \in V^e_j} \sum_{j' = j}^t \sum_{f \in U \setminus \set{e}} \P\left[x \in V^f_{j'} \right] \\
        &\leq\, \sum_{j = 1}^{k_e} \sum_{x \in V^e_j} \sum_{j' = j}^t \E[c_{j'}(x)] \\
        &\leq\, \sum_{j = 1}^{k_e} \sum_{x \in V^e_j} \sum_{j' = j}^t 5\ell\left(\frac{4\Delta^4}{\ell}\right)^{j'}.
    \end{align*}
    Using that $k_e \leq t$ and $|V^e_j| \leq \Delta + 1 + 2\ell \leq 3\ell$ for all $e \in U$ and $j \leq k_e$ and assuming that $\ell \geq 8\Delta^4$, we can bound the last expression as follows:
    \begin{align*}
        \sum_{j = 1}^{k_e} \sum_{x \in V^e_j} \sum_{j' = j}^t 5\ell\left(\frac{4\Delta^4}{\ell}\right)^{j'} \,&\leq\, 15 \ell^2 \sum_{j = 1}^{t} \sum_{j' = j}^t \left(\frac{4\Delta^4}{\ell}\right)^{j'} \\
        &\leq\, 30\ell^2 \sum_{j = 1}^{t} \left(\frac{4\Delta^4}{\ell}\right)^{j} \\
        &\leq\, 240 \ell\, \Delta^4.
    \end{align*}
    As this bound does not depend on $C_e$, we conclude that $\E[|N^+(e)|] \leq 240 \ell\, \Delta^4$ unconditionally, and
    \[
        \E[|E(\graphrand)|] \,\leq\, \sum_{e \in U} \E[|N^+(e)|] \,\leq\, 240 \ell\, \Delta^4 |U|. \qedhere
    \]

\end{proof}

    Since $\E[|S|] \geq 3|U|/4$ and $|S| \leq |U|$ always, Markov's inequality implies that $\P[|S| \geq |U|/2] \geq 1/2$. Using Lemma~\ref{lemma:edges_bound_Gamma_r}, we see that
    \begin{align*}
        240 \ell\,\Delta^4 |U| \,&\geq\, \E[|E(\graphrand)|] \\
        &\geq\, \E[|E(\graphrand)| \,\big|\, |S| \geq |U|/2] \,\P[|S| \geq |U|/2] \\
        &\geq\, \frac{\E[|E(\graphrand)| \,\big|\, |S| \geq |U|/2]}{2},
    \end{align*}
    and hence $\E[|E(\graphrand)| \,\big|\, |S| \geq |U|/2] \leq 480 \ell\, \Delta^4 |U|$. Since $d(\graphrand) = 2|E(\graphrand)|/|S|$, we have
    \[
        \E[d(\graphrand) \,\big\vert\, |S| \geq |U|/2] \,\leq\, \frac{2 \cdot 480 \ell\, \Delta^4 |U|}{|U|/2} \,=\, 1920 \ell \,\Delta^4.
    \]
    Therefore, by Lemma~\ref{lemma:alg_indep_set} we have
    \begin{align*}
        \E[|W|] \,\geq\, \E\left[\frac{|S|}{d(\graphrand) + 1}\right] \,&\geq\, \E\left[\frac{|S|}{d(\graphrand) + 1}\,\middle\vert\, |S| \geq |U|/2\right] \,\P[|S| \geq |U|/2] \\
        &\geq\, \frac{|U|}{4} \, \E\left[\frac{1}{d(\graphrand) + 1}\,\middle\vert\, |S| \geq |U|/2\right] \\
        [\text{by Jensen's inequality}]\qquad &\geq\, \frac{|U|}{4} \,\dfrac{1}{\E[d(\graphrand) \,\big|\, |S| \geq |U|/2] + 1} \\
        &\geq\, \frac{|U|}{4} \, \frac{1}{1920 \ell \,\Delta^4 + 1} \\
        &=\, \Theta(|U|/(\ell \, \Delta^4)) \\
        &=\, \Theta (|U|/\Delta^{20}).
    \end{align*}
    It follows that in $O(\ell t) = O(\Delta^{16} \log n)$ rounds we find a set $W$ and $e$-augmenting subgraphs ${H_e} \defeq C_e$ for all $e \in W$ with all the desired properties. This concludes the proof of Theorem~\ref{theo:disjoint_with_bounds}.

\subsection{Randomized distributed algorithm for $(\Delta + 1)$-edge-coloring}\label{subsec:dist_randomised}

With Theorem~\ref{theo:disjoint_with_bounds} in hand, we can establish Theorem~\ref{theo:dist}\ref{item:dist_rand}, i.e., present a randomized \LOCAL algorithm for $(\Delta+1)$-edge-coloring that runs in $\poly(\Delta) \log^2 n$ rounds. The algorithm proceeds in stages. We start by setting $\phi_0(e) \defeq \blank$ for all $e \in E$. At the start of stage $i$, we have a proper partial $(\Delta + 1)$-edge-coloring $\phi_{i-1}$. Let $U_{i-1} \defeq E \setminus \dom(\phi_{i-1})$. We then perform the following steps:
\begin{enumerate}[label=\ep{\arabic*}]
    \item Run the algorithm from Theorem~\ref{theo:disjoint_with_bounds} to compute a set $W_i \subseteq U_{i-1}$ and vertex-disjoint augmenting subgraphs ${H_e}$ for all $e \in W_i$.
    
    \item\label{step:augment} Let $\phi_{i}$ be the coloring obtained by augmenting $\phi_{i-1}$ with ${H_e}$ for all $e \in W_i$ simultaneously.
\end{enumerate}
Note that all the augmentations on step \ref{step:augment} can indeed be performed simultaneously, because the graphs ${H_e}$ are vertex-disjoint. By construction, $U_{i} = U_{i-1} \setminus W_i$ for all $i$. By Theorem~\ref{theo:disjoint_with_bounds}, we have
\[
    \E\left[|W_i|\,\bigg| U_{i-1}\right] \,=\, \Omega\left(\frac{|U_{i-1}|}{\Delta^{20}}\right).
\]
It follows that after $T$ stages, we have
\[\E[|U_T|] \,\leq\, \left(1 - \Omega(\Delta^{-20})\right)^T|U_0| \,=\, m\exp\left(-\Omega(T \Delta^{-20})\right).\]
By Markov's inequality, we conclude that
\[\P\left[|U_T| \geq 1\right] \,\leq\, \E[|U_T|] \,\leq\, m\exp\left(-\Omega(T \Delta^{-20})\right).\]
For $T = \Theta\left(\Delta^{20} \log n\right)$, we have $|U_T| < 1$ with probability at least $1 - 1/\poly(n)$. By Theorem~\ref{theo:disjoint_with_bounds}, each stage takes $O(\Delta^{16} \log n)$ rounds, and hence all edges will be colored after $O(\Delta^{36} \log^2 n)$ rounds with probability at least $1 - 1/\poly(n)$, as desired.

\subsection{Deterministic Distributed Algorithm}\label{subsec:dist_deterministic}

In this subsection we prove Theorem~\ref{theo:dist}\ref{item:dist_det}, i.e., we present a deterministic \LOCAL algorithm for $(\Delta + 1)$-edge-coloring with running time $\poly(\Delta, \log \log n) \log^5 n$. We derive this result from Theorem~\ref{theo:disjoint_with_bounds} by utilizing approximation algorithms for hypergraph maximum matching in exactly the same way as in \cite{VizingChain, Christ}.

Let us first introduce some new definitions and notation. Recall that $\log^*n$ denotes the iterated logarithm of $n$, i.e., the number of iterative applications of the logarithm function to $n$ after which the output becomes less than $1$. The asymptotic notation $\tilde{O}(\cdot)$ hides polylogarithmic factors, i.e., $\tilde{O}(x) = O(x\,\poly(\log x))$.
Let $\mathcal{H}$ be a hypergraph. To avoid potential confusion, we refer to the edges of $\mathcal{H}$ as ``hyperedges.'' The \emphd{rank} of $\mathcal{H}$, denoted by $r(\mathcal{H})$, is the maximum size of a hyperedge of $\mathcal{H}$. The \emphd{degree} of a vertex $x \in V(\mathcal{H})$ is the number of hyperedges containing $x$. The \emphd{maximum degree} of $\mathcal{H}$, denoted by $d(\mathcal{H})$, is the largest degree of a vertex of $\mathcal{H}$. A \emphd{matching} in $\mathcal{H}$ is a set of disjoint hyperedges. We let $\mu(\mathcal{H})$ be the maximum number of edges in a matching in $\mathcal{H}$.

While the usual \LOCAL model is defined for graphs, there is an analogous model operating on a hypergraph $\mathcal{H}$. Namely, in a single communication round of the \LOCAL model on $\mathcal{H}$, each vertex $x \in V(\mathcal{H})$ is allowed to send messages to every vertex $y \in V(\mathcal{H})$ such that $x$ and $y$ belong to a common hyperedge.

We shall use the following result due to Harris (based on earlier work of Ghaffari, Harris, and Kuhn \cite{derandomizing}):

\begin{theo}[{Harris \cite[Theorem 1.1]{harrisdistributed}}]\label{theo:hypergraph_matching}
There exists a deterministic distributed algorithm in the \LOCAL model on an $n$-vertex hypergraph $\mathcal{H}$ that outputs a matching $M \subseteq E(\mathcal{H})$ with $|M| = \Omega(\mu/r)$ in
$\tilde{O}(r \log d + \log^2 d + \log^* n)$
rounds, where $r \defeq r(H)$, $d \defeq d(H)$, and $\mu \defeq \mu(H)$.
\end{theo}

Let us now describe a deterministic \LOCAL algorithm for $(\Delta + 1)$-edge-coloring a graph $G$ of maximum degree $\Delta$. Suppose we are given a proper partial $(\Delta+1)$-edge-coloring $\phi \colon E \pto [\Delta + 1]$ and let $U \defeq E \setminus \dom(\phi)$ be the set of uncolored edges. Define an auxiliary hypergraph $\graphdet$ as follows: Set $V(\graphdet) \defeq V(G)$ and let $S \subseteq V(G)$ be a hyperedge of $\graphdet$ if and only if there exist an uncolored edge $e \in U$ and a connected $e$-augmenting subgraph $H$ of $G$ with $S = V(H)$ and $|E(H)| = O(\Delta^{16}\log n)$ (the implied constant factors in this bound are the same as in Theorem~\ref{theo:disjoint_with_bounds}). 

By definition, $r(\graphdet) = O(\Delta^{16} \log n)$, and, by Theorem~\ref{theo:disjoint_with_bounds}, $\mu(\graphdet) = \Omega(|U|/\Delta^{20})$. To bound $d(\graphdet)$, consider an edge $S \in E(\graphdet)$. Since the graph $G[S]$ is connected, we can order the vertices of $S$ as $(x_0, \ldots, x_{|S| - 1})$ so that each $x_i$, $i \geq 1$ is adjacent to at least one of $x_0$, \ldots, $x_{i-1}$. This means that once $(x_0, \ldots, x_{i-1})$ are fixed, there are at most $i\Delta \leq r(\graphdet)\Delta$ choices for $x_i$, and hence,
\[d(\graphdet) \,\leq\, (r(\graphdet)\Delta)^{r(\graphdet)} \,\leq\, \exp\left(2r(\graphdet) \log r(\graphdet)\right).\]
By Theorem~\ref{theo:hypergraph_matching}, we can find a matching $M \subseteq E(\mathcal{H})$ of size
\[
    |M| \,=\, \Omega\left(\frac{|U|}{\Delta^{36}\,\log n}\right)
\]
in $\tilde{O}(\Delta^{32} \log^2 n)$ rounds in the \LOCAL model on $\mathcal{H}$. Since a single round of the \LOCAL model on $\mathcal{H}$ can be simulated by $O(\Delta^{16}\log n)$ rounds in the \LOCAL model on $G$, $M$ can be found in $\tilde{O}(\Delta^{48} \log^3 n)$ rounds in the \LOCAL model on $G$. Once such a matching $M$ is found, in $O(\Delta^{16} \log n)$ rounds we may pick one augmenting subgraph $H$ with $V(H) = S$ for each $S \in M$ and simultaneously augment $\phi$ using all these subgraphs. (All the augmentations can be performed simultaneously because the sets in $M$ are disjoint.) This results in a new proper partial coloring with at most
\[
    \left(1 - \frac{1}{O(\Delta^{36} \log n)}\right) |U|
\]
uncolored edges.

Starting with the empty coloring and repeating this process $O(\Delta^{36} \log^2 n)$ times, we get a proper $(\Delta + 1)$-edge-coloring of the entire graph. The resulting deterministic \LOCAL algorithm runs in
\[
    \tilde{O}(\Delta^{48} \log^3 n) \cdot O(\Delta^{36} \log^2 n) \,=\, \tilde{O}(\Delta^{84} \log^5 n),
\] rounds, as desired.

\subsection*{Acknowledgment}

We are very grateful to the anonymous referees for their helpful comments.

\printbibliography

\end{document}